\documentclass[first, review, 3p, times, authoryear]{elsarticle}

\pdfoutput=1

\journal{arXiv}

\usepackage[utf8]{inputenc}

\usepackage{algorithm}
\usepackage{algpseudocode}

\usepackage{amssymb}
\usepackage{amsmath}
\usepackage{amsthm}
\usepackage{bbm}
\usepackage{booktabs}
\usepackage{calc}
\usepackage{chngcntr}
\usepackage{graphicx}
\usepackage{microtype}
\usepackage[figuresright]{rotating}
\usepackage{setspace} 
\usepackage[justification=centering]{subcaption}
\usepackage{standalone}
\usepackage{booktabs}
\usepackage{tikz}

\usepackage{xfrac}
\usepackage{color}

 \renewcommand{\textcolor}[2]{#2}

\newtheorem{theorem}{Theorem}
\newtheorem{proposition}[theorem]{Proposition}

\graphicspath{ {./figures/}{./figures/benchmark/} } 

\algdef{SE}[DOWHILE]{Do}{doWhile}{\algorithmicdo}[1]{\algorithmicwhile\ #1}

\usepackage{natbib}

\begin{document}

\begin{frontmatter}



\title{Rectangle Blanket Problem: Binary integer linear programming formulation and solution algorithms}

\author[CompEng]{Bar\i\c s Evrim Demir\" oz}
\ead{baris.evrim.demiroz@gmail.com}

\author[IndEng]{\.{I}. Kuban Alt{\i}nel}
\ead{altinel@boun.edu.tr}

\author[CompEng]{Lale Akarun}
\ead{akarun@boun.edu.tr}

\address[CompEng]{Department of Computer Engineering, Bo\u{g}azi\c{c}i University, 34342, Bebek, \.{I}stanbul, Turkey}
\address[IndEng]{Department of Industrial Engineering, Bo\u{g}azi\c{c}i University, 34342, Bebek, \.{I}stanbul, Turkey}


\begin{abstract}
A \textit{rectangle blanket} is a set of non-overlapping axis-aligned rectangles,  used to \textcolor{red}{approximately}  represent  the two-dimensional image of a shape approximately.  The use of a rectangle blanket is a widely considered strategy for speeding-up the computations in many computer vision applications. Since neither the rectangles nor the image have to be fully covered by the other, the blanket becomes more precise as the non-overlapping area of the image and the blanket decreases. In this work, we focus on the \textit{rectangle blanket problem}, which involves the determination of an optimum blanket minimizing \textcolor{red}{the} non-overlapping area with a given image subject to an upper bound on the total number of rectangles the blanket can include. This problem has similarities with rectangle covering, rectangle partitioning and cutting / packing problems. 
\textcolor{red}{The image replaces an irregular master object by an approximating set of smaller axis-aligned rectangles. The union of these rectangles, namely, the rectangle blanket, is neither restricted to remain entirely within the master object, nor required to cover the master object completely. } We first develop a binary integer linear programming formulation of the problem. Then, we introduce four methods for its solution. The first one is a branch-and-price algorithm that computes an exact optimal solution. The second one is a new constrained simulated annealing heuristic. The last two are heuristics adopting ideas available in the literature for other computer vision related problems. Finally, we realize extensive computational tests and report results on the performances of these algorithms.
\end{abstract}

\begin{keyword}
\texttt{Integer programming; branch-and-price; computer vision; cutting / packing; heuristics}
\end{keyword}

\end{frontmatter}


\section{Introduction}\label{intro}
The problem of  \textcolor{red}{approximately} representing a two-dimensional image approximately using multiple non-overlapping axis-aligned rectangles arises in many computer vision problems such as template matching \citep{MohrZachmann2010a,MohrZachmann2010b} and people tracking \citep{FleuretBerclazLengagneFua2008,DemirozSalahAkarun2014}. A rectangle is axis-aligned if its adjacent (orthogonal) edges are parallel to $x$ and $y$ axes.  An illustration for an approximation of the image with three rectangles is provided in Figure~\ref{fig:example}. Here, the rectangles determine a \textit{rectangle blanket}. We define the blanket as a set of non-overlapping axis-aligned rectangles. It should be noticed that a blanket does not have to cover the image perfectly. There can be uncovered parts of the image as well as uncovering parts of the blanket. Besides, it can be connected or disconnected depending on the image it approximates. Nevertheless, it is possible to say that the number of rectangles forming a blanket directly effects the quality of the approximation: the higher it is, the finer the approximation becomes. 

\begin{figure}[!h]
	\begin{center}
   		\includegraphics[width=0.25\linewidth]{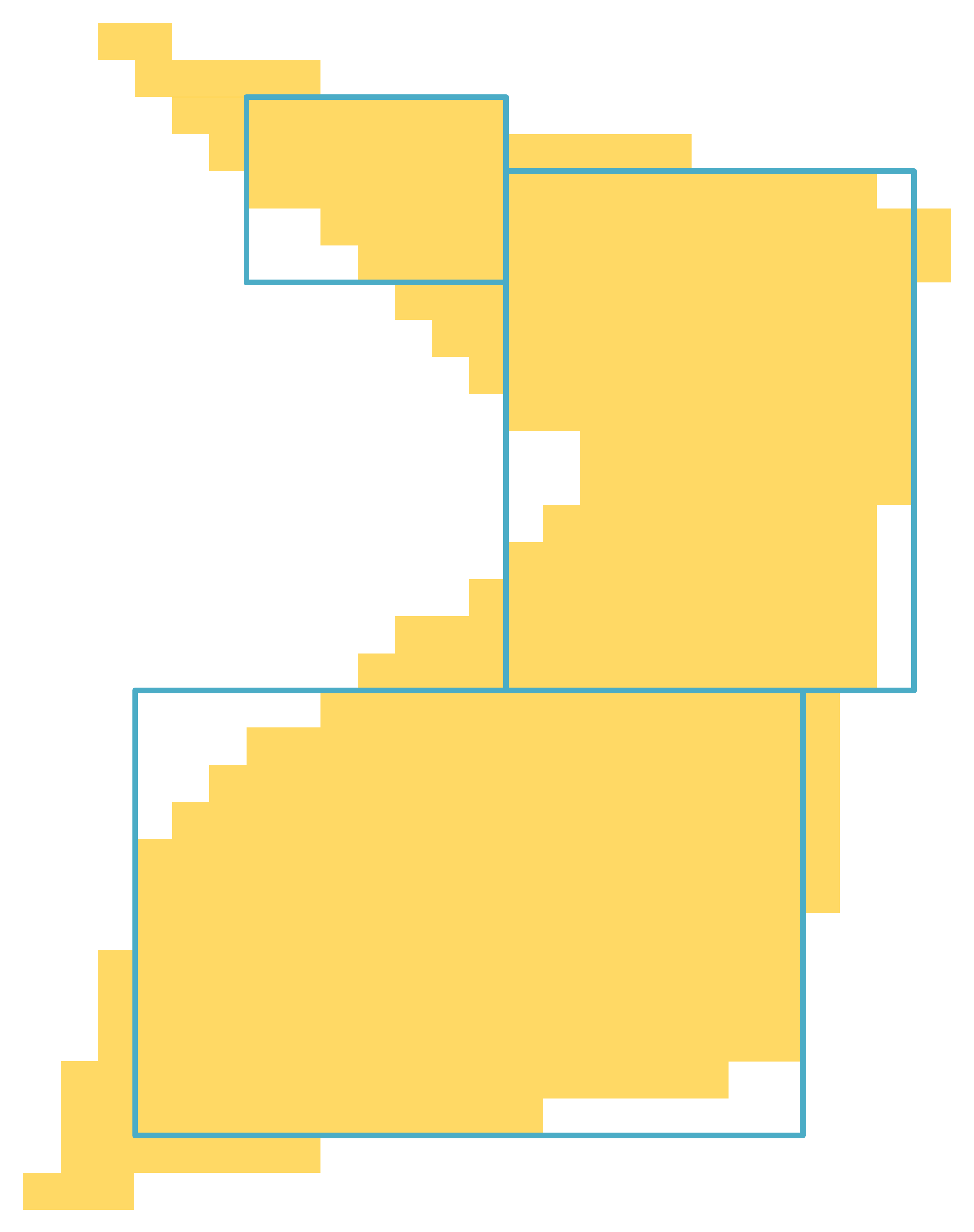}
   		\caption{An example of a rectangle blanket having $K=3$ rectangles}\label{fig:example}
   	\end{center}
\end{figure}

We formally define the \textit{Rectangle Blanket Problem} (RBP) as determining a rectangle set that contains no more than $K$ non-overlapping axis-aligned rectangles that minimizes the non-overlapping area with the given two-dimensional image (i.e. the uncovered area of the image and uncovering area of the blanket). This problem has similarities with \emph{rectangle covering}~\citep{ChaikenKleitmanSaksShearer1981,Heinrich-LitanLubbecke2006, StoyanRomanovaScheithauerKrivulya2011}, \emph{rectangle partitioning}~\citep{Ohtsuki1982,ORourkeTewari2001} and  \emph{cutting / packing}~\citep{DyckhoffScheithauerTerno1997,WascherHaussnerSchumann2007} problems, which we try to outline briefly in the following lines. It is not our aim to provide a comprehensive review of the similarities and dissimilarites between all the works in these three fertile research areas and this one, but to use a number of  examples to illustrate potentials for cross-fertilization of ideas and methodologies among RBP and the problems they study. 

Rectangle covering and partitioning problems, and cutting / packing problems have attracted considerable research interest in the last three decades and sophisticated solution approaches have been proposed. As a consequence of the inherent difficulty of the problems \citep{FowlerPatersonTanimoto1981, HochbaumMaass1985, HaesslerSweeney1991, CulbersonReckhow1994} these are mostly heuristic methods. We direct interested readers to \citep{DanielsInkulu2001, BaldacciBoschettiGanovelliManiezzo2014} for didactic overviews. 
\textcolor{red}{Despite the emergence of heuristics with considerably better quality of solutions, the research interest in methods based on mathematical formulation, has not waned. Furthermore, the availability of increased computational power made available by advances in hardware technology has accelerated the search for exact solution methods.}

\subsection{Cutting and packing problems}
RBP and cutting / packing problems have a similar structure. First of all, two sets of elements, namely, a master object (i.e. input, supply, image) and a set of small items (i.e. output, demand, nonidentical rectangles) are defined. Second,  small items are selected and grouped  into a subset (i.e. a rectangle blanket) which is assigned to a large object such that the geometric condition (i.e. rectangles are non-overlapping and axis-aligned, but not necessarily lying within the image) holds, and a given objective function is optimized. As a result, it is possible to say that RBP is somewhat related to cutting / packing problems coded as 2/B/O/R or 2/B/O/M in particular, according to Dyckhoff's typology \citep{Dyckhoff1990}, which is improved later on by \cite{WascherHaussnerSchumann2007} in order to clear away the ambiguities. RBP can be seen as an output maximization problem, and within this category of problems, it can be treated as a close relative of the two-dimensional single large object placement problem (SLOPP) according to Wacher et al.'s improved typology. However, although small items are rectangular, \textcolor{red}{for two reasons}, it is not possible to say that RBP exactly belongs to the family of the two-dimensional rectangular SLOPP with explicit upper bounds on the number of times a small item of particular type can be cut from  the large object \citep{ChristofidesWhitlock1977, Wang1983, Beasley1985b,ChristofidesHadjiconstantinou1995,Beasley2004}. First of all, in RBP, the upper bound is aggregated and the total number of rectangles is bounded. Besides, the master object, namely the image, has an irregular shape rather than rectangular, which pushes RBP towards irregular cutting / packing problems \citep{DownslandDownsland1995,BaldacciBoschettiGanovelliManiezzo2014,CherriMundimAndrettaToledoOliveiraCarravilla2016}. 

Mathematical programming formulations \textcolor{red}{of cutting and packing problems} are mixed-integer linear programs (MILPs) in general and their complexity depends on the geometric methodologies used to model constraints related to the type of interactions between the small items and their layout on the master object. The core geometric methodologies \textcolor{red}{available in the literature} are very well explained in the tutorial by \cite{BennellOliveira2008}. There are four of them; but they can be grouped into two: the pixel / raster method and polygonal methods. The pixel / raster method divides the surface of the master object into discrete regions in order to represent it as a grid coded in the form of a matrix. One drawback of this representation is the inability to represent irregular shapes accurately. However, this is not the case when the D-function is used; it enables benefitting from the well known tests for line intersection and point intersection of direct trigonometry. Unfortunately the inefficiency in checking the feasibility is its weakness; it takes exponential time in the number of edges of the items, which is quadratic in the grid size for the pixel / raster method. The nofit polygon (NFP) is a polygonal construct that offers higher efficiency than direct trigonometry, and more accuracy, since it uses the original edges. The nofit polygon of two items $i$ and $j$ (NFP$_{ij}$), is the locus of all points where the reference point of item $j$ cannot be placed without overlapping item $i$. Unfortunately, calculating NFP$_{ij}$ is still a non-trivial task and can be very time consuming. Finally, the $\Phi$-function is the most recent polygonal methodology invented to represent all mutual positions of two polygons \citep{StoyanTernoScheithauerGilRomanova2001, StoyanScheithauerGilRomanova2004}. The major problem is the determination of the expression of a suitable $\Phi$-function.

The proposed MILP formulations are not \textcolor{red}{numerous}. Some of them are compaction models hybridized with meta-heuristics. Compaction is known to be a difficult problem as well \citep{LiMilenkovic1993}. Compaction models can improve layouts by moving small items continuously on the master object without overlapping, by changing the relative positions of the pairs. A good example of this type of work is due to \cite{GomesOliveira2006}. Their hybrid algorithm solves the Linear Programming (LP) relaxation of their MILP compaction formulation under the guidance of a simulated annealing search. The formulation uses NFP to model constraints preventing overlapping. \textcolor{red}{The main drawback of this formulation is that its definition of non-overlapping constraints does not limit the relative positions of the small items strictly enough so that many different branches of a branch-and-bound (BB) tree can contain the same solution, when this formulation is used in a BB algorithm. }
Clearly, this can increase the inefficiency of the search. \cite{FischettiLuzzi2009} develop another MILP formulation. \textcolor{red}{They use a different formulation of NFP than the one used by Gomes and Oliveira (2006). }
Their approach takes advantage of the earlier effort of \cite{DanielsMilenkovicLi1994} and \cite{Li1994}. It is based on slicing, i.e. partitioning into convex disjoint areas, of the region outside the NFP, which corresponds to the region the second small item can be placed onto without overlapping the first one, for every pair of small items. However, they do not specify the way in which the slices are defined. \cite{Alvarez-ValdesMartinezTamarit2013} use their slices more specifically and define binary variables for each region in which the reference point of a small item can be placed with respect to another one without overlapping. Using variables associated with slices resolves the inefficiency problem of Gomes and Oliveira's BB algorithm \citep{GomesOliveira2006}. They define slices horizontally, which helps to control relative vertical position of the small items while developing two new MILP formulations and BB algorithms for their solutions. Besides, \cite{Alvarez-ValdesMartinezTamarit2013} prefer a set partitioning constraint to describe the exterior of the NFP as  \cite{FischettiLuzzi2009}, which is done using a set covering constraint by \cite{GomesOliveira2006}. \textcolor{red}{Algorithms to build the NFP are time consuming, complex and numerically unstable. This limits real world applications of these models.} \cite{CherriMundimAndrettaToledoOliveiraCarravilla2016} propose two directions to overcome these limitations and built robust mathematical optimization models. The first one derives non-overlapping constraints based on direct trigonometry without using NFP. The second one decomposes small objects into convex pieces prior to the computation of the NFPs. Notice that both of the approaches are polygonal, contrasting the preference of  \cite{BaldacciBoschettiGanovelliManiezzo2014}; their MILP formulations for nesting with defects is based on the pixel / raster approach with the simplest coding scheme \citep{OliveiraFerreira1993}.

RBP is not the first close encounter of computer vision with cutting / packing. The NFP is related to Minkowski or vector sums. Any two-dimensional region can be considered as a set of vectors, and the Minkowski sum of two regions is the region obtained by summing all pairs of vectors. This relation was first pointed out by \cite{StoyanPonomarenko1977} and used extensively by \cite{MilenkovicDanielsLi1992,LiMilenkovic1993} in order to determine the constraints in their compaction processes. Minkowski sum is in fact a particular case of a more general concept which is known as mathematical morphology \citep{Shih2017}. It is widely used in computer vision and image processing and motivates some interesting applications for the cutting / packing of highly irregular items in an irregular master object \citep{WhelanBachelor1991, WhelanBachelor1996, BouganisShanahan2006}.

\subsection{Covering and partitioning problems}
The rectangle covering problem mainly deals with the covering of a compact polygonal target region with a finite family of small rectangles. There can be different \textcolor{red}{objectives}. For instance, one can search for a minimum number of rectangles needed to cover the target region \citep{Heinrich-LitanLubbecke2006}, or in case there exists several covers, one can look for the best one with respect to some objective \citep{StoyanRomanovaScheithauerKrivulya2011}.  The rectangle partitioning problem has similar features; but  the  \textcolor{red}{aim} is  to obtain a partition or dissection of the target region \citep{Ohtsuki1982}, this time. 

As can be noticed there are many relations between covering and cutting / packing problems. One example is the use of $\Phi$-function as a geometry modeling tool \citep{StoyanRomanovaScheithauerKrivulya2011}. In the same work the authors exploit the particular nature of the covering problem and discuss the use of an extension of the $\Phi$-function called $\Gamma$-function \citep{Stoyan2007}, which they apply for determining whether a given set of rectangles, with respect to their configuration, form a cover of the target region.

The rectangle covering problem is hard and exact solution methods have an enumerative nature. Some of them use the above mentioned $\Phi$ and $\Gamma$ functions such as the one developed by \cite{StoyanRomanovaScheithauerKrivulya2011} for which the choice of a suitable initial configuration is particularly important. Another possibility is to try to solve the problem after formulating it as an MILP. →\textcolor{red}{The major difference is the definition of the target object: in our case, it consists of finite number of rasterized points, in contrast with the branch-and-bound algorithm of Stoyan et al. (2011), where the target region includes infinite number of points.}

MILP models of the rectangle covering problem have both advantages and disadvantages as shown by \citep{ScheithauerStoyanRomanova2009} based on the new formulations they propose. They adopt Beasley's approach \citep{Beasley1985a,Beasley1985b} to formulate the first one and use a three indexed binary variable to describe the placement of the reference point, which is the lower left corner, of a rectangle at a position. This results in a binary integer programming problem (BIP) with very large number of variables and weak LP relaxation lower bound, which makes it difficult to solve exactly by a BB algorithm using an LP relaxation based lower bounding scheme. They also propose a second formulation under the assumption that the polygonal target region is convex. They do not use a pixel / raster method to represent the target region and define relative  position variables for modeling the interactions between the rectangles, which lowers considerably the number of variables and constraints since they are independent of the size of the target region, which is not the case with the first formulation.

\subsection{Rectangle blanket problem}
Nesting problems are two-dimensional cutting / packing problems involving irregular shapes. They can be roughly defined as the placement of small items in a configuration in a master object subject to possible constraints related to defective areas or areas with different quality of the master object and their compatibility with the small items to cut, in order to optimize an objective. In RBP there is one master object with an irregular shape. It is the image to be assigned a blanket consisting of non-overlapping axis aligned small rectangles not necessarily packed within the image, which is the main difference between RBP and irregular nesting. Besides, the objective of RBP is different. It consists of the minimization of the sum of the uncovered area of the master object, which can be treated as the total waste, and the excess of the small rectangles lying outside the image, which is the uncovering area of the blanket. However, the typical objective in nesting is to minimize the waste (e.g. \cite{DownslandDownsland1995}). In the case of nesting with defects, which occurs when the master object has defect zones, each small item has a quality value depending on its configuration in the master object and the objective becomes the maximization of the total quality value of the cutting patterns of the small items \citep{BaldacciBoschettiGanovelliManiezzo2014}.

In rectangle covering and partitioning there is one master object, i.e. target region, as RBP. The main difference is again in the objective function, i.e. the relation between the set of small items and the master object. RBP is a slightly more relaxed version of both problems. As can be observed in Figure~\ref{fig:example}, neither the shape nor the rectangles need to be fully contained within the other. In other words a rectangle blanket is neither a cover nor a partition; it partly covers and partly dissects a polygonal region, namely the target image. \textcolor{red}{Consequently, it is possible to see it as a relative of \textit{approximate (or incomplete) rectangle partitioning} or \textit{exclusive maximal rectangle covering problems}. The maximal covering problem \citep{ChurchReVelle1974, Murray2016}, which is often referred in facility location theory is quite similar except for the exclusive (or partitioning) constraints. It basically asks to maximize the covered area while using a given number of primitive shapes.} \cite{Chan2014} tackle a similar problem in the context of integrated circuit manufacturing where the rectangles are allowed to overlap. After transforming the binary input image appropriately, RBP can also be considered as a generalization of Bentley's classic \emph{maximum sum subsequence problem} to two dimensions for multiple subsequences~\citep{Bentley1984}. In fact, ~\cite{Csuros2004} have previously proposed a method to find a set of $K$ disjoint subsequences of a one dimensional array such that the sum of all elements in the set is maximized. Later, ~\cite{Bengtsson2006,Bengtsson2007} improved the efficiency of his method to have linear time complexity.

There are three coding schemes in pixel / raster method: the first one is proposed by \cite{SegenreichBraga1986}, the second one by \cite{OliveiraFerreira1993}, and the third one by \cite{BabuBabu2001}. The scheme by \cite{OliveiraFerreira1993} is the simplest and mostly preferred  in  MILP formulations; it uses 1 to code the item and 0 to represent the empty space. We formulate RBP as a BIP based on this scheme.  \textcolor{red}{This formulation resembles} Beasley's formulation \citep{Beasley1985a,Beasley1985b} at the first look. The master object, i.e. target image, is represented by the binary coefficient matrix, which is given as the part of the problem data  as will be seen in the next section. \textcolor{red}{In addition,} the binary variables are single indexed and represent whether a rectangle is selected or not in the packing constraints that allows the points of the master surface to be covered by at most one rectangle.

\subsection{Our contributions}
As we have tried to expose in detail RBP is a new problem having relations with cutting / packing, rectangle covering and partitioning problems. Following its definition, first  we introduce a binary integer linear programming (BIP) formulation, which can be classified as a set packing formulation \citep{ConfortiCornuejolsZambelli2014} extended with an additional cardinality restriction on the total number of rectangles forming the blanket. We benefit from the geometry and pixel / raster representation of the a computer image  while modeling the coefficients of the objective function and constraint matrix. Then, we develop a new branch-and-price (BP) algorithm, which we implement using two branching rules. The first one is the well-known variable branching of integer programming. The second one extends the rule by \cite{RyanFoster1981}  \textcolor{red}{for set packing}, which is originally proposed for set partitioning problems, for set packing. Our third contribution is the geometric bisection scheme that solves the pricing subproblems efficiently. The algorithmic development considers also tailing-off effect and technics for its prevention such as the the use of Lagrangean lower bounds and dual smoothing. BP algorithm can be computationally expensive for large instances. To overcome this size limitation, we suggest three heuristics. The first one is a new simulated annealing heuristic; and it is our forth contribution. The other two, adopt ideas available in the literature for other computer vision related problems. In addition, we perform extensive computational tests for assessing the performance of the algorithms.

The paper consists of six sections. The BIP formulation and BP algorithm can be found in the next two sections. The heuristics are explained in Section~\ref{sec:heuristics}. Section~\ref{sec:experiments}  essentially  reports the computational results. Finally, the last section concludes the paper.

\section{Problem formulation}\label{sec:formulation}

A \textit{polyomino} is a union of unit squares, namely a square with integer coordinates and area $1$~\citep{Schrijver2003}. It is possible to represent a shape as a polyomino $\mathcal{P}$, by replacing any pixel $\mathbf{p}$ belonging to its binary image with a unit square.

Given a polyomino $\mathcal{P}$, it is possible to construct a graph $G=(V(G),E(G))$ with vertices as all pixels contained in $\mathcal{P}$. Two vertices are the endpoints of an edge if and only if $\mathcal{P}$ contains an axis-aligned rectangle containing both pixels. $G$ is called the \textit{visibility graph} of the polyomino $\mathcal{P}$~\citep{Maire1994, MotwaniRaghunathanSaran1989, MotwaniRaghunathanSaran1990} and has interesting properties. Gy\"ori~\citep{Gyori1984} has shown that $\alpha(G)$, the minimum number of cliques that cover $G$, is equal to $\chi(G)$, the maximum number of independent vertices (i.e. the size of the maximum stable set), if each vertical or horizontal line has a convex intersection with $\mathcal{P}$, namely if $\mathcal{P}$ is \textit{orthogonally convex}. This is a consequence of the fact that $G$ is perfect for orthogonally convex images. Using the correspondence between a clique and a rectangle, and the relation $\alpha(G)\geq \chi(G)$ it is possible to obtain a lower bound on the size of a rectangle cover of $\mathcal{P}$. An intriguing question is whether there is a constant factor approximation algorithm for the rectangle cover problem~\citep{BernEppstein1997}, which is related to whether the ratio $\alpha(G) / \chi(G)$ is bounded by a constant. Although this is still an open question, numerical results~\cite{Heinrich-LitanLubbecke2006} report in their study support the fact that such an approximation exists. These results are based on the BIP formulation of a set covering problem. We also follow this line of research and propose a set packing type BIP formulation for RBP. 

A target image $\mathcal{I}$ is a union of disconnected polyominos and it is possible to associate a zero-one matrix $I\in \mathbb{B}^{W\times H}$ with the target image $\mathcal{I}$, where $W$ and $H$ are the dimensions of the target image. For example, they are the dimensions of the smallest rectangle covering the polyominos representing the target image fully. Then, for $p_1 = 1,2,\dots,W$ and $p_2 = 1,2,\dots,H$
\begin{equation}\label{equ:def-image}
I_{p_1p_2} = \left\{ \begin{array}{cl} 1 & \text{if pixel } \mathbf{p} \text{ belongs to target image } \mathcal{I}\\
0 & \text{otherwise,} \end{array}\right.
\end{equation}
is a binary matrix representing target image $\mathcal{I}$, where $\left( \begin{array}{c} p_1 \\  p_2 \end{array}\right)$ are the indices of the cell assigned to pixel $\mathbf{p}$.
We present the rectangle $r$ where the pixel values are $1$ inside the rectangle,  with the binary matrix $\mathbf{r}\in\mathbb{B}^{W\times H}$ such that 
\begin{equation}\label{equ:def-rectangle}
r_{p_1p_2} = \left\{ \begin{array}{cl} 1 & \text{if } r_\text{left} \le p_1 \le r_\text{right} \text{ and } r_\text{top} \le p_2 \le r_\text{bottom}\\
0 & \text{otherwise,} \end{array}\right.
\end{equation}
for any pixel $\mathbf{p}$ with cell coordinates $\left( \begin{array}{c} p_1 \\  p_2 \end{array}\right)$. Here, $r_\text{left},r_\text{right},r_\text{top}$ and $r_\text{bottom}$ are the coordinates of the left, right, top and bottom edges of the rectangle, respectively.  Please observe that this is exactly the pixel / raster method with \cite{OliveiraFerreira1993} coding scheme used to represent a master object, i.e. target image. It is used to represent small items on the master surface in cutting / packing \citep{BaldacciBoschettiGanovelliManiezzo2014}.  Throughout the paper we have adopted image based conventions, because the rectangle blanket problem arises from image processing and computer graphics related fields. According to this convention, y-axis of the coordinate system is oriented downwards. Therefore, $r_\text{top} \le r_\text{bottom}$ for rectangles.

Let $\mathcal{R}$ be the set of all possible rectangles in $\mathbb{B}^{W\times H}$ and has size $|\mathcal{R}| = R$. Then RBP can be formulated as the subset selection problem
\begin{equation}
\text{RBP: }\min \left\{ z(\mathcal{B}) : \mathcal{B}\subseteq \mathcal{R}, |\mathcal{B}|\leq K, \text{ and }  \mathcal{B} \text{ is a blanket}  \right\},
\end{equation}
for a given target image $\mathcal{I}$. Here a feasible subset $\mathcal{B} = \{\mathbf{r}^1,\mathbf{r}^{\mbox 2},\dots,\mathbf{r}^B\}$ of $B = |\mathcal{B}|\leq K$ \textcolor{red}{non-overlapping rectangles represents a blanket} and cost 
\begin{eqnarray}
z(\mathcal{B}) & = & \sum_{\mathbf{r}\in\mathcal{B}}\left( \sum_{p_1=1}^W\sum_{p_2=1}^H r_{p_1p_2} - \sum_{p_1=1}^W\sum_{p_2=1}^HI_{p_1p_2}\times r_{p_1p_2}\right) \nonumber \\ 
& + & \left(\sum_{p_1=1}^W\sum_{p_2=1}^HI_{p_1p_2} - \sum_{\mathbf{r}\in\mathcal{B}} \sum_{p_1=1}^W\sum_{p_2=1}^HI_{p_1p_2}\times r_{p_1p_2}\right).
\end{eqnarray}
The first term is the area (i.e. the number of pixels) of blanket $\mathcal{B}$ overflowing the target image (i.e. total uncovering area) and the second term is the area of the target image uncovered by blanket $\mathcal{B}$ (i.e. total uncovered area), since rectangles are non-overlapping. The product $I_{p_1p_2}\times r_{p_1p_2}$ is $1$ if $\mathbf{p}$ belongs to both the target image and rectangle $\mathbf{r}$; otherwise it is $0$. Besides, the first term of the second line equals to the size of the target \textcolor{red}{image}, namely $\sum_{p_1=1}^W\sum_{p_2=1}^HI_{p_1p_2} = | \mathcal{I} |$ and  the objective can be rearranged as:
\begin{equation}
z(\mathcal{B}) = | \mathcal{I} | + \sum_{\mathbf{r}\in\mathcal{B}}\left( \sum_{p_1=1}^W\sum_{p_2=1}^H r_{p_1p_2} - 2\sum_{p_1=1}^W\sum_{p_2=1}^HI_{p_1p_2}\times r_{p_1p_2}\right).
\end{equation}
Here, $ | \mathcal{I} |$ is constant and thus RBP has the equivalent form 
\begin{equation}
\text{RBP: }\min \left\{  \sum_{\mathbf{r}\in\mathcal{B}}c(\mathbf{r}) : \mathcal{B}\subseteq \mathcal{R}, |\mathcal{B}|\leq K, \text{ and }  \mathcal{B} \text{ is a blanket} \right\},
\end{equation}
with 
\begin{equation}
c(\mathbf{r}) = \sum_{p_1=1}^W\sum_{p_2=1}^H r_{p_1p_2} - 2\sum_{p_1=1}^W\sum_{p_2=1}^HI_{p_1p_2}\times r_{p_1p_2}. \label{eq7}
\end{equation}

\textcolor{red}{Before proceeding any further we would like to comment on RBP's computational difficulty. Let us first consider the decision version of a weighted variant of the set packing problem (SP), the maximum weight set packing problem (MWSP) for this purpose.}

\noindent \textcolor{red}{INSTANCE: A family $\mathcal{F}$ of finite sets, positive integer weights $c(F)$ $F\in \mathcal{F}$, and a positive integer $L$.\\
QUESTION: Does $\mathcal{F}$ contain a subset $\mathcal{B}\subseteq \mathcal{F}$ of mutually disjoint sets with total weights $\sum_{F\in \mathcal{B}} c(F) \geq L$}

\textcolor{red}{MWSP is NP-Complete since SP is a restriction with unit weights (i.e. $c(F) = 1$) $F\in \mathcal{F}$, which is shown to be NP-Complete by \cite{Karp1972}. It is also possible to state a decision version of RBP: } 

\noindent \textcolor{red}{INSTANCE: A target image $\mathcal{I}$ represented with a $W\times H$† matrix $I$ with binary entries, a family $\mathcal{F}$ of rectangles represented with binary matrices, positive integer weights $c(\mathbf{r})$ $\mathbf{r}\in \mathcal{F}$ for each rectangle, and positive integers $K$ and $L$.\\
QUESTION: Does $\mathcal{F}$ contain a subset $\mathcal{B}\subseteq \mathcal{F}$ of mutually disjoint rectangles (i.e. a blanket) with size $|\mathcal{B}|\leq K$ and total weights $\sum_{\mathbf{r}\in \mathcal{B}} c(\mathbf{r}) \geq L$}

\textcolor{red}{RBP is NP-Complete since it can be restricted to MWSP by setting $K = |\mathcal{I}|\leq W\times H$, which makes the cardinality constraint redundant, with rectangle weights calculated using formula (7). Yet, another way of reaching to the verdict that RBP is NP-hard can be by means of the planar geometric packing problem (PGP):} 

\noindent \textcolor{red}{INSTANCE: A set $\mathcal{B}$ of geometric objects, positive integer weights $c(F)$ $F\in \mathcal{B}$, a not necessarily connected region $\mathcal{R}$† in the plane.\\
QUESTION: Is it possible to determine whether the set $\mathcal{B}$ of geometric objects can be placed within region $\mathcal{R}$ in a mutually non-intersecting way?}

\textcolor{red}{\cite{FowlerPatersonTanimoto1981} have shown that PGP is NP-Complete even when the set $\mathcal{B}$ is restricted to a given number of identical squares to be placed with their sides parallel to the axes of a Cartesian coordinate system. RBP is restricted to PGP for the case  $\mathcal{B}$ is a set of $|\mathcal{B}|\leq K$ non-overlapping rectangles (i.e. a blanket) and total weight $\sum_{r\in B}c(\mathbf{r}) \geq L$. Again, the weight $c(\mathbf{r})$ is calculated according to formula (7) (or (5) for the total weight).}

As many of the combinatorial optimization problems, the above formulation can be re-expressed using binary variables and a slight change in the notation with $c(\mathbf{r}^j)$ denoting the value (cost) of rectangles $\mathbf{r}^j \in \mathcal{R}$, $j=1,2,\dots,R$, as the constrained set packing problem
\begin{eqnarray}
\text{RBP: }& \min & \sum^R_{j=1} c(\mathbf{r}^j) x_j  \label{eq3}\\
 & \text{s.t.} &  \sum^R_{j=1} x_j \leq K \label{eq4}\\
       & & \sum^R_{j=1} r_{p_1p_2}^jx_j  \leq 1  \hspace{1cm} p_1 = 1,2,\dots,W; p_2 =1,2,\dots,H \label{eq5} \\
       & & x_j \in \{0,1\}  \hspace{1.7cm} j = 1,2,\dots,R. \label{eq6}
\end{eqnarray}
Here, the decision variable $x_j$ is set to $1$ if rectangle $j$ is selected for the blanket; otherwise it is set to $0$.  Inequality~(\ref{eq4}) restricts the number of rectangles in an optimal blanket.  The given integer upper bound $K$ is actually an implicit parameter for the approximation quality: the larger it is, the better the blanket $\mathcal{B}$ approximates the target image $\mathcal{I}$.  Set packing inequalities~\eqref{eq5} allow each pixel $\mathbf{p}$ to be in at most one of the rectangles. The number of all possible configurations of $K$ rectangles is $\mathcal{O}((W\times H)^{2K})$, which can be very large depending on the size of the target image.

Mohr and Zachmann use a fitness function while assessing the quality of a given rectangle set based on the total number of uncovered and uncovering pixels. They consider RBP within the context of silhouette matching in particular, and propose a dynamic programming approach~\citep{MohrZachmann2010a} and a recursive search heuristic~\citep{MohrZachmann2010b}. \cite{DemirozSalahAkarun2014} benefit from the same fitness function for person tracking and fall detection. More general criteria are also suggested in the literature for the analytical description of relations between a given target image and the finite set of rectangles that is supposed to approximate it;  $\Phi$ function applied as a fitness criterion for object packing~\citep{BennellScheithauerStoyanRomanova2010}, 
$\Gamma$ function for polygonal region covering~\citep{StoyanRomanovaScheithauerKrivulya2011},  and a value function based on the quality indices associated with pixels for nesting with defects \citep{BaldacciBoschettiGanovelliManiezzo2014}. Unfortunately, their evaluation requires more computational effort and they do not serve our purpose better than~(\ref{eq7}) in the BIP formulation~(\ref{eq3})--(\ref{eq6}) of RBP.

\section{An exact solution method}\label{sec:BnP}

Since the number of all possible rectangles can be large depending on the size of the image matrix, column generation procedure can be applied to generate new rectangles as long as they improve the objective function. However, column generation is not directly applicable when the variables are integer. We wrap the whole column generation problem into a branch-and-bound scheme, which is also known as branch-and-price. Let RBP$^{(t)}$ denote the restricted integer programming  master problem (i.e. rectangle blanket problem) at step $t$. Then, the LP relaxation of the restricted integer programming master at step $t$ (RLPM$^{(t)}$) is the linear programming problem
\begin{eqnarray}
\text{RLPM$^{(t)}$:} &\min &\sum^{R^{(t)}}_{j=1} c(\mathbf{r}^j) x_j  \label{eq8}\\
& \text{s.t.} &  \sum^{R^{(t)}}_{j=1} x_j \leq K \hspace{6.35cm} \label{eq9}\\
& & \sum^{R^{(t)}}_{j=1} r_{p_1p_2}^jx_j  \leq 1  \hspace{0.5cm} p_1 = 1,2,\dots,W; p_2 =1,2,\dots,H \label{eq10} \\
       & & x_j \geq 0  \hspace{1.9cm} j = 1,2,\dots,R^{(t)}, \label{eq11} 
\end{eqnarray}
after relaxing binary restrictions on the variables. Here, $R^{(t)} \leq R$ is the number of columns (variables) and $\mu^{(t)}$ and $\{\pi^{(t)}_{p_1p_2}:p_1 = 1,2,\dots,W; p_2 =1,2,\dots,H\}$ are the optimal values of the dual variables  $\mu$ and $\{\pi_{p_1p_2}:p_1 = 1,2,\dots,W; p_2 =1,2,\dots,H\}$ at step $t$. Inequalities $x_j \leq 1$ are not included since they are implied by constraints~\eqref{eq10} and \eqref{eq11} . Then, the reduced cost can be expressed as

\begin{equation}
\overline{c}(\mathbf{r}^j) = c(\mathbf{r}^j) - \sum^{W}_{p_1=1} \sum^{H}_{p_2=1} r_{p_1p_2}^j \pi_{p_1p_2} - \mu\hspace{1cm} j = 1,2,\dots,R^{(t)}, \label{eq18}
\end{equation}
where the dual variables \mbox{\boldmath$\pi$} and $\mu$ are restricted to be nonpositive.

We provide a formal pseudocode of the column generation algorithm that solves linear programming relaxation of the integer programming formulation of RBP as Algorithm~\ref{alg1}. This is the most generic form and does not include improvements and implementation details such as the multiple column generation, lower bounding, and stabilization. \textcolor{red} {However, we still believe that it is worth mentioning some of them here since they can make the pseudocode easier to follow. Notice that $\mathbf{x} = \mathbf{0}$ is always a trivial solution and RLPM$^{(t)}$ is always feasible. This requires the setting of the $W\times H + 1$ slack variables to the right hand sides of the inequalities \eqref{eq9} and \eqref{eq10} as the initial basic feasible solution. Hence, the $(W\times H +1)\times (W\times H +1)$ identity matrix, which corresponds to the slack variables can always be selected as the initial basic matrix. The solution of the pricing problem at step $t$ is mainly for checking \mbox{\boldmath $\pi$}$^{(t)}$ and $\mu^{(t)}$ are dual feasible for LPM. This provides the minimum reduced cost value $\overline{c}(\mathbf{r}^{(t+1)})$ and an optimal rectangle (i.e. column) $\mathbf{r}^{(t+1)}$. Objective value $z_{RLPM^{(t)}} = \sum_{j=1}^{R^{(t)}}c(\mathbf{r}^j)x_j^{(t)}$ of the restricted LP master is obtained by only considering $R^{(t)}$ columns generated up to step $t$ and $R^{(t+1)} = R^{(t)} + 1$.}

\begin{algorithm}[h]
\begin{algorithmic}[1]
   \Procedure{SolveLinearProgrammingMaster}{LPM} \textcolor{red}{
   \State\hspace{-0.8cm} (Initialization): Set $t=0$ and solve RLPM$^{(0)}$. 
   \State\hspace{-0.8cm} (Pricing): Call Algorithm~\ref{alg:branch-and-bound-psp} to solve pricing subproblem. Let $\overline{c}(\mathbf{r}^{(t+1)})$ and $\mathbf{r}^{(t+1)}$ be the minimum reduced cost value and corresponding optimal rectangle.
    \State\hspace{-0.8cm} (Optimality check): If $\overline{c}(\mathbf{r}^{(t+1)}) = 0$, then STOP. Set  $\mathbf{x}_{LPM} = \mathbf{x}^{(t)}$ and $z_{LPM} = z_{RLPM^{(t)}}$ and go to Step 6. Otherwise, go to Step 5.
  \State\hspace{-0.8cm} (Generating a new column): If  $\overline{c}(\mathbf{r}^{(t+1)}) < 0$, then introduce the column
 \begin{equation*}
\left( \begin{array} {c} 1 \\ \mathbf{r}^{(t+1)} \end{array} \right)
\end{equation*}
with variable $x_{R^{(t)} + 1}$ and unit cost $c(\mathbf{r}^{(t+1)})$ to obtain RLPM$^{(t+1)}$. Set $t \leftarrow t + 1$,  and go to Step 3.
   \State\hspace{-0.7cm} \Return $\mathbf{x}_{LPM}$ and $z_{LPM}$.
   \EndProcedure}
\end{algorithmic}
\caption{Column generation algorithm that solves the LP relaxation of RBP.}
\label{alg1}
\end{algorithm}

The solution of the RLPM using column generation yields an optimal solution $\mathbf{x}^* = \mathbf{x}^{(t^*)}$ (at some node of the branch-and-bound tree) at step $t^*$ after solving RLPM$^{(t^*)}$ (i.e. restricted master LP with $R^{(t^*)}$ columns). If it is fractional and it is not possible to prune that node by bound or infeasibility, new branches can be created by partitioning the solution set of the integer programming master RBP, which can be done by adding constraint(s) either to the master problem or the pricing subproblems. In the branch-and-price algorithm, we let $\mathcal{L}$ be a collection of rectangle blanket problems RBP$^i$ of the form $z^i = \min \{\mathbf{c}^T\mathbf{x}: \mathbf{x}\in \mathcal{S}^i\}$ where $\mathcal{S}^i$ is a subset of the original feasible solution set $\mathcal{S}$ (i.e. the feasible solution set of the original problem RBP described with constraints~(\ref{eq4})--(\ref{eq6})). Associated with each subproblem of $\mathcal{L}$ is a lower bound $\underline{z}^i \leq z^i$. However, there are some special difficulties in combining column generation with integer programming techniques such as selecting a branching strategy and dealing with slow convergence.

\subsection{Pricing subproblem}

As a consequence of the reduced cost expression~(\ref{eq18}) the pricing subproblem can be formulated as
\begin{equation}
\overline{c}(\mathbf{r}^{(t+1)}) = \min \{c(\mathbf{r}) - \sum^{W}_{p_1=1} \sum^{H}_{p_2=1} r_{p_1p_2}\pi^{(t)}_{p_1p_2}  - \mu^{(t)}: \mathbf{r} \in \mathcal{R}\}, \label{eq19}
\end{equation}
which can be restated as
\begin{equation}
\overline{c}(\mathbf{r}^{(t+1)}) + \mu^{(t)} = \min \{\sum^{W}_{p_1=1} \sum^{H}_{p_2=1} (1 - 2I_{p_1p_2} - \pi^{(t)}_{p_1p_2})r_{p_1p_2}: \mathbf{r} \in \mathcal{R}\} \label{eq20}
\end{equation}
by using the definition of $c(\mathbf{r})$ given with expression~(\ref{eq7}) and the fact that $\mu^{(t)}$ is constant. Here, $\mathbf{r}^{(t+1)} = \arg \min\{\overline{c}(\mathbf{r}^{(t+1)}) - \mu^{(t)}\}$, namely an optimal solution of the minimization problem given as the right-hand side of expression (\ref{eq20}). Hence, if $\overline{c}(\mathbf{r}^{(t+1)}) =  z(\mathbf{r}^{(t+1)}) - \mu^{(t)} = 0$ where $z(\mathbf{r}^{(t+1)})$ is its optimal value, then the current dual solution is feasible and an optimal solution of the LP relaxation of RBP is reached. Otherwise, if $\overline{c}(\mathbf{r}^{(t+1)}) < 0$, a new column is added with the following entries: $c(\mathbf{r}^{(t+1)})$ for objective function (\ref{eq8}), $1$ for inequality (\ref{eq9}) and $\{r_{p_1p_2}^{(t+1)}: p_1 = 1,2,\dots,W; p_2 =1,2,\dots,H\}$ for inequalities (\ref{eq10}) corresponding to the $R^{(t+1)} = (R^{(t)} + 1)$\textsuperscript{th} variable. Observe that $c(\mathbf{r}^{(t+1)})$ is calculated using expression (\ref{eq7}) with $\mathbf{r}^{(t+1)}$.
As a consequence of this discussion the pricing subproblem becomes
\begin{equation}
\text{PSP}^{(t)}:\ z(\mathbf{r}^{(t+1)}) = \min \left\{\sum^{W}_{p_1=1} \sum^{H}_{p_2=1} (1 - 2I_{p_1p_2} - \pi^{(t)}_{p_1p_2})r_{p_1p_2}: \mathbf{r} \in \mathcal{R}\right\}. \label{eq21}
\end{equation}
Please recall that dual variables $\pi^{(t)}_{p_1p_2}$ are restricted to be nonpositive.

It is possible to formulate the pricing subproblem using binary variables as a BIP. Unfortunately, they could include many binary variables and \textcolor{red}{the problem} becomes inefficient to solve. Hence, for the solution of the subproblems created at each node of the RBP's search tree, we adapt a branch-and-bound method that geometrically searches within a rectangle for a subrectangle with \textcolor{red}{the} minimum cost (i.e. such that the sum of the weights of the pixels included in the subrectangle is the smallest). This method is originally proposed for determining the subwindow for localizing an object of interest within a given image~\citep{LampertBlaschkoHofmann2009}. We summarize our adaptation in the sequel of this subsection for the sake of completeness. The idea is to start with an initial rectangle set containing all possible rectangles and dividing it into two disjoint subsets at each branch. Notice that this geometric algorithm, which is essentially a two-dimensional interval bracketing (bisection) method, eventually computes an optimal solution of the pricing subproblem PSP$^{(t)}$ at iteration $t$. 

Using the definition given in~(\ref{equ:def-rectangle}) we can represent any rectangle $\mathbf{r}$ by means of  the quadruplet $(r_\text{top},r_\text{left},r_\text{bottom},r_\text{right})$ denoting respectively its top, left, bottom and right edge coordinates. As we have already mentioned, because of the conventions we adopt from image processing and \textcolor{red}{computer graphics} related fields, the orientation of the $y$-axis is downward, and $r_\text{top} \le r_\text{bottom}$ as a result.

Similarly, we represent a set $\mathcal{R}^{k}$ of rectangles associated with node $k$ of the search tree as $\mathcal{R}^{k} = (\mathbf{T}^{k},\mathbf{L}^{k},\mathbf{B}^{k},\mathbf{R}^{k})$. Here, $\mathbf{T}^{k},\mathbf{L}^{k},\mathbf{B}^{k}$ and $\mathbf{R}^{k}$ are the range of values that top, left, bottom and right edge's coordinates can take. Each of these ranges have lower and upper bounds. For example  $\mathbf{T}^{k} = [\text{top}_\text{low}^{k}, \text{top}_\text{high}^{k}]$ means that the top edge coordinate of the rectangles of set $\mathcal{R}^{k}$ are between $\text{top}_\text{low}^{k}$ and $\text{top}_\text{high}^{k}$, inclusive.  This representation is illustrated in Figure~\ref{fig:rectangle-set} with $\mathbf{r}_{\cup}$  and $\mathbf{r}_{\cap}$ denoting respectively the  subsets corresponding to the union and intersection of all rectangles in the set. 
\begin{figure}[!h]
	\begin{center}
   		\includegraphics[width=0.6\linewidth]{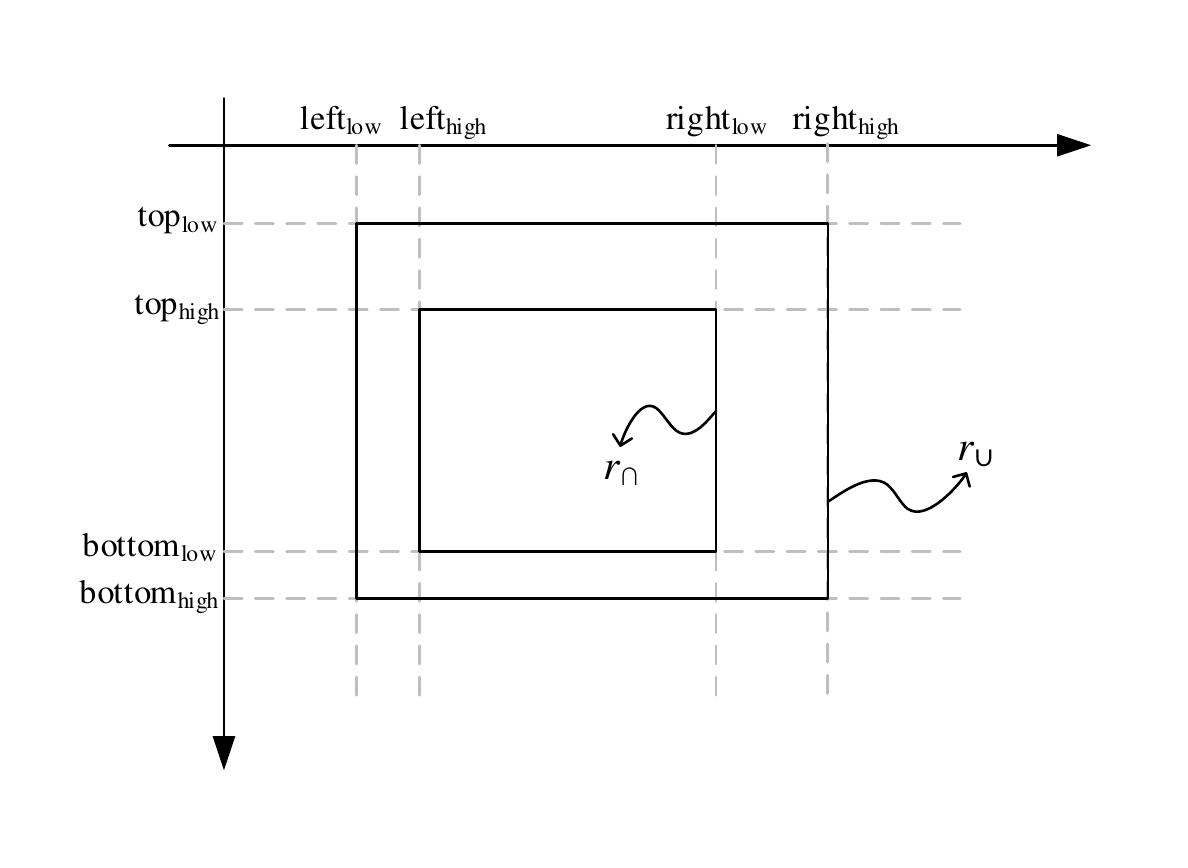}
  		\caption {Representation of a rectangle set} \label{fig:rectangle-set}
  	\end{center}
\end{figure}

At the beginning, the initial node (i.e. node $0$) of the branch-and-bound tree includes the whole set of rectangles $\mathcal{R}^{0} = ([1,H],[1,W],[1,H],[1,W])$. At each step, the largest range is bisected to form two branches such that rectangle subsets form a partition of their parent node's rectangle set. Binary search proceeds according to the best-first branching rule using a lower bound on the sum of the weights of the pixels the rectangles include for pruning, until a node representing a subset consisting of a single rectangle with the minimum cost is reached.

The sum of the negative weights inside a rectangle, is a lower bound for the sum of all the weights inside it. We are going to use a similar principle to define a lower bound for a given rectangle set. Let $z(\mathbf{r})$ be the cost associated with a given rectangle $\mathbf{r}$ and $\ell(\mathcal{R})$  be the lower bound on the costs of the rectangles belonging to the rectangle set $\mathcal{R}$. Hence,
\begin{equation}
z(\mathbf{r}) = \sum^{r_\text{right}}_{p_1=r_\text{left}} \sum^{r_\text{top}}_{p_2=r_\text{bottom}}(1 - 2I_{p_1p_2} - \pi^{(t)}_{p_1p_2})r_{p_1p_2}, \label{eq30}
\end{equation}
is  the objective function for rectangle $\mathbf{r}= (r_\text{left},r_\text{right},r_\text{top},r_\text{bottom})$. Then, $\ell(\mathcal{R})$ satisfies the following properties:
\begin{align}
	\ell(\mathcal{R}) &\leq \min_{\mathbf{r}\in \mathcal{R}}z(\mathbf{r}), \\
	\ell(\mathcal{R}) &= z(\mathbf{r}) \text{ if set } \mathcal{R} \text{ is the singleton } \{\mathbf{r}\}.
\end{align}

A value for $\ell(\mathcal{R})$ can be calculated by means of the formula
\begin{equation}
\ell(\mathcal{R}) = \ell^{-}(\mathbf{r}_{\cup}) + \ell^{+}(\mathbf{r}_{\cap}).\label{eq31}
\end{equation}
Here, $\ell^{-}(\mathbf{r})$ and  $\ell^{+}(\mathbf{r})$ are the sum of the negative and positive weights of the pixels of rectangle $\mathbf{r}$. The inequalities
\begin{eqnarray}
\ell^{+}(\mathbf{r}_{\cap})& \leq & \ell^{+}(\mathbf{r}), \\
\ell^{-}(\mathbf{r}_{\cup}) & \leq & \ell^{-}(\mathbf{r}) 
\end{eqnarray}
hold for any rectangle $\mathbf{r}$ of set  $\mathcal{R}$ because of two reasons. First, every rectangle $\mathbf{r}\in \mathcal{R}$ includes $\mathbf{r}_{\cap}$ so the sum of the positive pixels inside $\mathbf{r}$ is at least $l^{+}(\mathbf{r})$. Second, every rectangle $\mathbf{r}\in \mathcal{R}$ is included in $\mathbf{r}_{\cup}$, and following the same reasoning, it can be seen that the second inequality also holds. If we sum the inequalities we can see that $\ell(\mathcal{R}) \leq \min_{\mathbf{r}\in \mathcal{R}}z(\mathbf{r})$.

The rectangle set representation allows $\mathbf{r}_{\cup}$ and $\mathbf{r}_{\cap}$ to be calculated very quickly since, $\mathbf{r}_{\cup} = [\text{top}_\text{low}, \text{left}_\text{low}, \text{bottom}_\text{high}, \break \text{right}_\text{high}] $ and $\mathbf{r}_{\cap} = [\text{top}_\text{high}, \text{left}_\text{high}, \text{bottom}_\text{low}, \text{right}_\text{low}]$. When $\mathbf{r}_\cap$ does not define a valid rectangle, then the intersection is empty. Furthermore, $\ell^{-}(\mathbf{r})$ and  $\ell^{+}(\mathbf{r})$ can be evaluated in constant time using integral images (summed area tables) of only negative and only positive values~\citep{Crow1984}, which makes the computation of the lower bounds very efficient.

We have all the ingredients of a branch-and-bound algorithm that can solve the pricing subproblem. Let $\mathcal{L}$ be a collection of subproblems PSP$^i$ of the form $z(\mathcal{R}^i) =  \min \left\{\sum^W_{p_1=1} \sum^H_{p_2=1} w_{p_1p_2}r_{p_1p_2}: \mathbf{r} \in \mathcal{R}^i\right\}$; where $w_{p_1p_2} = (1 - 2I_{p_1p_2} + \pi^{(t)}_{p_1p_2})$ is the weight of pixel $\mathbf{p}$. Associated with each subproblem PSP$^i$ of $\mathcal{L}$ is a lower bound $\ell(\mathcal{R}^i) \leq z(\mathcal{R}^i)$ computed according to~(\ref{eq31}). We list the steps of the branch-and-bound algorithm, \textcolor{red}{which solves PSP}, formally as Algorithm~\ref{alg:branch-and-bound-psp}. 

Finally, \cite{LampertBlaschkoHofmann2009} report that their subwindow search algorithm runs in linear time or faster. Although we have not tried to show rigorously, we can say that Algorithm~\ref{alg:branch-and-bound-psp} also has a polynomial time complexity in the worst case since it is a close relative of \cite{LampertBlaschkoHofmann2009} algorithm: it is essentially a two-dimensional bisection procedure. In addition, lower bounds can be computed \textcolor{red}{efficiently}, as mentioned above.
\begin{algorithm}[h]
\begin{algorithmic}[1]
   \Procedure{SolvePricingSubproblem}{PSP}
   \State\hspace{-0.8cm} (Initialization): Set PSP$^0$ = PSP, $\mathcal{R}^0 = \mathcal{R}$,  $\underline{z}^0 = \ell(\mathcal{R}^0)$,  $\mathcal{L}$ = \{PSP$^0$\},  $\mathbf{r}^* \leftarrow (1,W,1,H)$ and \break $\overline{z} = \sum_{p_1 = 1}^W \sum_{p_2=1}^H \max(0, 1 - 2I_{p_1p_2} + \pi^{(t)}_{p_1p_2})r_{p_1p_2}^*$.
   \State\hspace{-0.8cm} (Termination test): If $\mathcal{L} = \emptyset$, then   \textcolor{red}{ rectangle $\mathbf{r}^*$} that yields $\overline{z}$ is an optimal solution of the pricing subproblem.
  \State\hspace{-0.8cm} (Problem solution and lower bound computation): Select and delete a problem from  $\mathcal{L}$, say PSP$^i$.
  \State\hspace{-0.8cm} (Pruning): 
  \begin{description}
  \item[i.  ] If  $\underline{z}^i \geq \overline{z}$, then go to Step 3.
  \item[ii. ] If rectangle set $\mathcal{R}^i$ does not consist of a single rectangle, then go to Step 6.
  \item[iii.] If rectangle set $\mathcal{R}^i$ consists of a single rectangle, say $\mathbf{r}^i$, and $z(\mathcal{R}^i) < \overline{z}$, then set $\overline{z} = z(\mathcal{R}^i)$ 
  \item\hspace{0.4cm} and $\mathbf{r}^* \leftarrow \mathbf{r}^i$, go to Step 3. Otherwise go to Step 6.
  \end{description}
\State\hspace{-0.8cm} (Division):   \textcolor{red}{ Let $\{\mathcal{R}^{ij}\}_{j=1}^2$ be a bi-partition of $\mathcal{R}^i$ and $\mathcal{R}^{i1} \cap \mathcal{R}^{i2} = \emptyset$, which is obtained by bisecting the largest range of rectangle $\mathbf{r}^i$ of set $\mathcal{R}^i$. Add problems \{PSP$^{ij}$\}$_{j=1}^2$ to $\mathcal{L}$ with lower bounds $\underline{z}^{ij} = \ell(\mathcal{R}^{ij})$ $j = 1,2$; go to Step 3.}
   \State\hspace{-0.8cm} \Return $\mathbf{r}^*$ and $\overline{z}$.
   \EndProcedure
\end{algorithmic}
\caption{Branch-and-bound algorithm that solves PSP for computing $\mathbf{\mathbf{r}}^*$.}
\label{alg:branch-and-bound-psp}
\end{algorithm}

\subsection{Branching}

First difficulty one can face in combining column generation with integer programming techniques is related to the branching rule. A scheme suitable for column generation must be devised at some node, say node $h$, of the branch-and-bound tree. We should point out that we omit index $h$ to clarify this fact in the derivations in order to increase the readability. For example, we use $\mathcal{S}$, $\mathcal{R}$, $\mathcal{S}^i$, $\mathcal{R}^i$, RLPM$_i^{(t)}$, PSP$_i^{(t)}$ instead of $\mathcal{S}^h$, $\mathcal{R}^h$, $\mathcal{S}^{hi}$, $\mathcal{R}^{hi}$, RLPM$_{hi}^{(t)}$, PSP$_{hi}^{(t)}$.

\subsubsection{Branching explicitly in the master problem}

A conventional scheme is to consider one of the fractional entries of $\mathbf{x}^*$, say $x_j^*$, and set  $x_j = 0$ in one branch and $x_j = 1$ on the other. These are usually introduced as new constraints to the master problem. Unfortunately, this is not a good choice since it yields an unbalanced branch-and-bound tree and the optimum is usually reached after many branchings as explained in~\citep{Vanderbeck2005}. 

We have observed that the use of column generation in the solution of the linear programming master LPM produces integral or close to integral optimal solutions in particular for RBP, which is an important advantage and results in very few branchings.  Hence, one can prefer branching implicitly in the master variables, since it is simple to implement and define a partition of the feasible solution set $\mathcal{S}$ as $\mathcal{S}^{i} = \{\mathcal{S}\cap \{\mathbf{x}: x_q = i\}\}$ for $i=0,1$ given that $x_q^* \not\in \mathbb{Z}$ is the fractional entry of $\mathbf{x}^*$ and $x_q$ is selected as the branching variable.

\subsubsection{Branching implicitly in the master problem}

This type of branching is preferred to the previous one since it is more likely to produce a balanced search tree as pointed by  \cite{Wolsey1998}. A well-known branching rule for the set partitioning problems is due to \cite{RyanFoster1981}. It is possible to adopt it for the solution of the set packing problem after adding slack variables and transforming packing inequalities to partitioning equalities. However, this causes a considerable increase in the number of variables. Instead, we propose a new branching scheme for RBP without considering the cardinality inequality \eqref{eq4}. It is a consequence of the next proposition and suitable for the solution of RBP by branch-and-price.

\begin{proposition}\label{prop1}
Let $\mathbf{x}$ be the fractional optimal solution to the LPM, i.e. $0 < x_k < 1$ for some $k = 1,2, \dots, R$. Then, there exists two pixels $\left( \begin{array} {c} e_1 \\ e_2 \end{array}\right)$ and $\left( \begin{array} {c} f_1 \\ f_2 \end{array}\right)$ and rectangles $j \neq k$ such that $r_{e_1e_2}^k \neq r_{e_1e_2}^j$, $r_{f_1f_2}^k = r_{f_1f_2}^j = 1$, $0 < x_j \leq 1$.
\end{proposition}
\begin{proof}
Assume that $r_{f_1f_2}^j = 0$ for all $f = \left( \begin{array} {c} f_1 \\ f_2 \end{array}\right)$ such that $r_{f_1f_2}^k = 1$, $x_j > 0$ and $j \neq k$. Then, we can improve the objective value by setting $x_k = 1$ since doing so does not violate any constraints. \textcolor{red}{This contradicts the proposition that $\mathbf{x}$ is an optimal solution. }  Hence, there must exist a rectangle $j$ and pixel $f = \left( \begin{array} {c} f_1 \\ f_2 \end{array}\right)$ with $r_{f_1f_2}^k = r_{f_1f_2}^j = 1$. Then, since there is no duplicated column in the basis, there must exist one pixel $e = \left( \begin{array} {c} e_1 \\ e_2 \end{array}\right)$ such that $r_{e_1e_2}^k \neq r_{e_1e_2}^j$.
\end{proof}

After identifying rows $e$ and $f$  we can impose the branching constraints 
\begin{eqnarray}
\sum_{k: r_{e_1e_2}^k = r_{f_1f_2}^k = 1} x_k \leq 0 & \text{ and } & \sum_{k: r_{e_1e_2}^k = r_{f_1f_2}^k = 1} x_k \geq 1
\end{eqnarray}
for respectively left and right branches. The pixels $e$ and $f$ have to be covered by different rectangles on the left branch and by the same rectangle on the right branch. 

As a consequence of this proposition if it is not possible to identify any $(e,f)$ pairs, then the solution of the master problem must be integer. The algorithm terminates in a finite number of branchings since there are only a finite number of row pairs. Besides, a large number of variables are eliminated at each branch. This rule eliminates the submatrix
\begin{equation}
\left( \begin{array} {cc} 1 & 1 \\ 0 & 1 \end{array} \right),
\end{equation}
which is precisely the excluded submatrix in the characterization of the totally balanced matrices \citep{HoffmanKolenSakarovitch1985}. Total balancedness of the coefficient matrix is a sufficient condition for the LP relaxation of a set packing problem to have integer optimal solutions.

RBP includes a cardinality constraint for the number of sets to be packed in addition to the ordinary set packing constraints. The new branching rule still remains valid for its solution by branch-and-price.

\subsubsection{Implementing the branching rules}

Notice that, new constraints added to the pricing subproblem should be compatible with the geometric solution procedure presented previously as Algorithm~\ref{alg:branch-and-bound-psp}. For example, they cannot be expressed in an algebraic form, since an explicit integer programming formulation of PSP, is not considered. Branching schemes that can be implemented by modifying the objective function \eqref{eq21} are appropriate.

However, as can be observed, any branching scheme that can be handled by only modifying the objective function \eqref{eq3} is appropriate. We restrict a pixel coverage by making its weight arbitrarily large (i.e. $w_{u_1u_2} = +\infty$) and arbitrarily small (i.e. $w_{u_1u_2} = -\infty$), respectively for $r_{u_1u_2}=0$ and $r_{u_1u_2}=1$. 

This is  what we do for $x_j = 1$ branch of the first branching rule: the weights of the pixels belonging to rectangle $j$ are set to very large numbers in the image of the pricing subproblem so that no rectangle overlapping with rectangle $j$ is generated. The right hand side of the constraints of the master problem are modified accordingly. However, for $x_j = 0$ branch the list of the prohibited rectangles belonging to the ancestors is passed to the child. So a rectangle generated by the pricing subproblem is discarded if it belongs to this list and the second best rectangle is considered. This procedure continues until the best column (i.e. the column with the smallest reduced cost) that is not included in the list is generated.

Another approach is to find the geometric representation of the new restrictions at each branch. First, columns violating the branching constraints are removed from the master problem. When the second branching rule is used, at a left branch where two pixels are forced to reside in different rectangles, we prune the node by deleting the associated rectangle set if their intersection contains these two pixels as illustrated in Figure~\ref{fig:branching_conditions}a. In case a right branch is visited we prune the node if the union of rectangles contains only one of these two pixels, since both of the pixels must reside in the same rectangle. This is illustrated in Figure~\ref{fig:branching_conditions}d and Figure~\ref{fig:branching_conditions}e. Remaining alignments, which are illustrated with  Figure~\ref{fig:branching_conditions}b, Figure~\ref{fig:branching_conditions}c and  Figure~\ref{fig:branching_conditions}f are valid for both branches.

\begin{figure}[!h]
	\begin{center}
    		\includegraphics[width=0.80\linewidth]{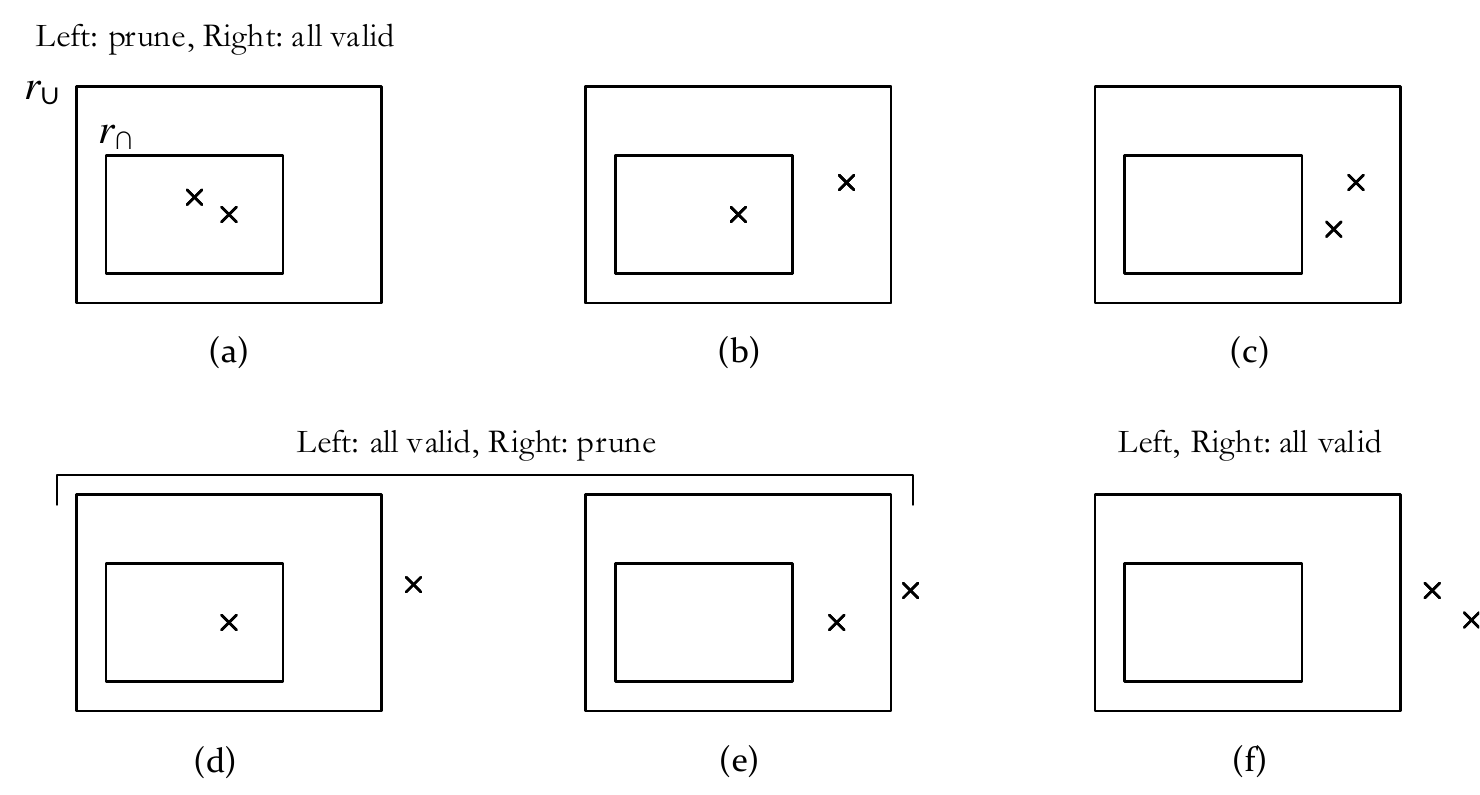}
    		\caption{Possible alignments of two pixels with respect to the rectangle set during pricing}\label{fig:branching_conditions}
		\end{center}
\end{figure} 

\subsection{Improvements for the tailing-off effect}

The next difficulty is the efficiency of the branch-and-bound algorithm. It can decrease remarkably because of the tailing-off effect of the column generation algorithm, which can be prevented by means of effective lower bounds on the optimal value of RLPM and dual smoothing.

\subsubsection{Lower bound for early stopping column generation}

It is possible to adapt the very well-known Lagrangean bounds (see page 449 in \citep{VanderbeckWolsey2010} for RBP, which we give in the following proposition.

\begin{proposition}
Let $LB(\mbox{\boldmath$\pi$}^{(t)},\mu^{(t)}) = K \overline{c}(\mathbf{r}^{(t+1)}) + z_{RLPM^{(t)}}
$, where $z_{RLPM^{(t)}}$ and $z_{LPM}$ are respectively the optimal values of the relaxed master problem at step $t$ and the LP relaxation of the original BIP formulation. Then, 
\begin{equation}
z_{RLPM^{(t)}} \geq z_{LPM} \geq LB(\mbox{\boldmath$\pi$}^{(t)},\mu^{(t)}).
\end{equation}
\end{proposition} 
As a direct consequence of this proposition, column generation can be stopped at step $t$ if $\lceil LB \rceil \geq z_{RLPM^{(t)}}$ or $\lceil LB \rceil \geq \overline{z}$ where $LB$ is the maximum of the previously computed lower bounds and $\overline{z}$ is the current integer upper bound (i.e. objective value of the incumbent). The first case implies that the restricted master is solved to optimality and the second one implies that the current node can be pruned. For the first case, we use the fact that the optimal value of RBP is integer (i.e. it is the sum of the product of a binary variable with integer coefficients).

Because of the simplified notation we use (i.e. no explicit index denoting the node of the search tree) the bounds given above may seem to be valid only for the root node. However, it is possible to state similar stopping conditions for any node of the branch-and-bound tree following a similar path.

\subsubsection{Dual smoothing}

One of the reasons behind the tailing-off effect is the oscillation of the dual variables erratically during the iterations of the column generation. One remedy is to stabilize the dual variables using the convex combination of the current and incumbent dual values as \cite{Wentges1997} and \cite{ PessoaSadykovUchoaVanderbeck2015} suggest:
\begin{equation}
  \tilde{\mbox{\boldmath$\pi$}}^{(t)} = \alpha \mbox{\boldmath$\pi$}^{(t)}_{\text{best}}+ (1-\alpha) \mbox{\boldmath$\pi$}^{(t)}.
\end{equation}
Here, $\mbox{\boldmath$\pi$}^{(t)}_{\text{best}}$ are the dual variable values that give the largest of the lower bounds (i.e. incumbent dual) calculated up to iteration $t$.  Then, the pricing subproblem is solved with the modified duals $\tilde{\mbox{\boldmath$\pi$}}^{(t)}$, instead of the current duals $\mbox{\boldmath$\pi$}^{(t)}$. If the reduced cost of the generated column is nonnegative, it is a misprice and the column is not added to RLPM$^{(t)}$. Since we are not working on a Lagrangean relaxation of RBP, there is no available subgradient to benefit from the automatic parameter tuning described in \cite{PessoaSadykovUchoaVanderbeck2015}. Nevertheless, we have experimentally noticed that setting $\alpha = 0.8$ increases the efficiency considerably.
 
\section{Heuristics}\label{sec:heuristics}

Solving RBP exactly can be computationally demanding, especially for instances of realistic size. Hence, efficient heuristics are very valuable. In this section, we introduce three of them. The first two are simple heuristics that enlarge the blanket, one rectangle at a time. They are based on the ideas proposed in the previous works by~\cite{DemirozSalahAkarun2014} and~\cite{MohrZachmann2010b}. The third one is a novel application of constrained simulated annealing.  

\subsection{Split and Fit}\label{sec:splitNfit}

Split and Fit heuristic (SF) iteratively splits rectangles and decreases the objective function using their \emph{fitness scores}. The fitness score of rectangle $\mathbf{r}$ is defined to be inversely proportional to the area of the non-overlapping region of the rectangle with the target image and can be formulated as
\begin{equation}
g(\mathbf{r})= \frac 1{\sum_{p_1 = r_{left}}^{r_{right}}\sum_{p_2 = r_{top}}^{r_{bottom}}(1-I_{p_1p_2})}\cdot\label{equ:rect-fitness}
\end{equation}
It can be computed in constant time using the integral image (i.e. summed area table) of the given image. The blanket determined by SF is represented by the rectangle set $\mathcal{R}^*$, which is initialized with the axis-aligned initial rectangle $(1,W,1,H)$ for a given object image. Another structure $\mathcal{Q}$ is maintained for keeping all possible split pairs along with their fitness scores. At each step, a rectangle $\mathbf{r}$ with the lowest fitness score and its complement rectangle $\mathbf{r}'$ (i.e. \textcolor{red}{the rectangle with the higher fitness score}
) are removed from $\mathcal{Q}$. Its parent, namely the rectangle that is split to obtain $\mathbf{r}$ is removed from $\mathcal{R}^*$ and all the siblings are removed from $\mathcal{Q}$. Finally, $\mathbf{r}$ and  $\mathbf{r}'$ are added to $\mathcal{R}^*$, and their  possible split pairs are added to  $\mathcal{Q}$. This splitting process is repeated until $K$ rectangles are obtained.  In the original work, \cite{DemirozSalahAkarun2014} used a tolerance value to enable early stopping. Unfortunately, although early stopping increases the efficiency, the solutions become less accurate. Therefore, we have run SF without early stopping (i.e. with zero tolerance) in our  experiments.

Possible split pairs of a rectangle are generated by dissecting it vertically and horizontally in different proportions. This process is controlled with parameter $\rho$. Formally, possible split ratios of a rectangle pair are $(\frac{i}{\rho}, \frac{\rho -i}{\rho})$ for $i=1,2,\dots,\rho - 1$. For example, for $\rho = 4$ possible vertical splits of a rectangle with width $W$ have widths $\{ (\frac{W}{4}, \frac{3W}{4}), (\frac{W}{2}, \frac{W}{2}), (\frac{3W}{4}, \frac{W}{4})\}$. Note that $\rho = 2$ corresponds to splitting the rectangles exactly in half as it is in the original due to \cite{DemirozSalahAkarun2014}. We have experimentally observed that setting $\rho = 3$ performs better in our experiments.

This procedure yields a set of rectangles covering the target image, i.e. the pixels are fully contained in the rectangles. To improve the solution further, as a final step, each rectangle is shrunk until $c(\mathbf{r})$ increases. \textcolor{red}{Growing / shrinking operation moves one of the edges of a rectangle by one pixel outside/ inside, as shown in Figure \ref{fig:proposals}.} The steps of the heuristic are illustrated in Figure~\ref{fig:decompose} for $\rho = 2$.

\begin{figure}
   \centering
   {\includegraphics[width=0.1\linewidth]{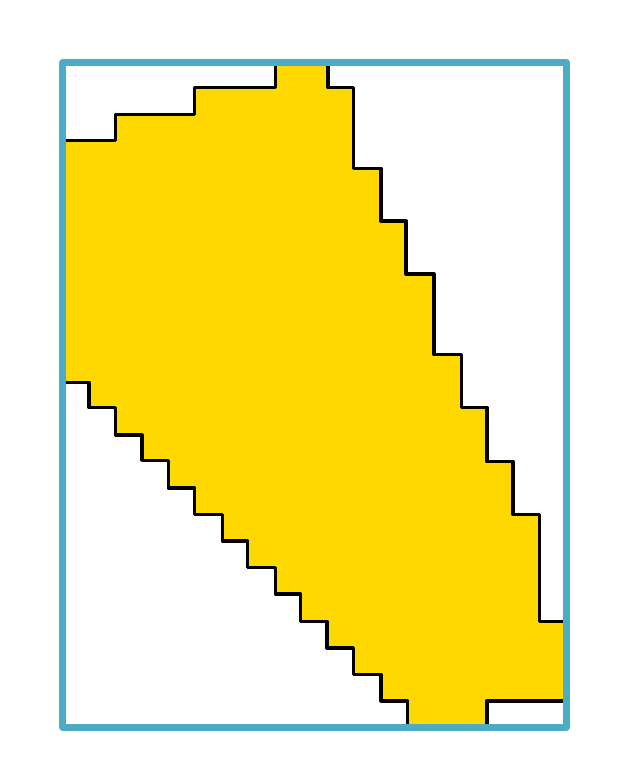} $\;$
   \includegraphics[width=0.1\linewidth]{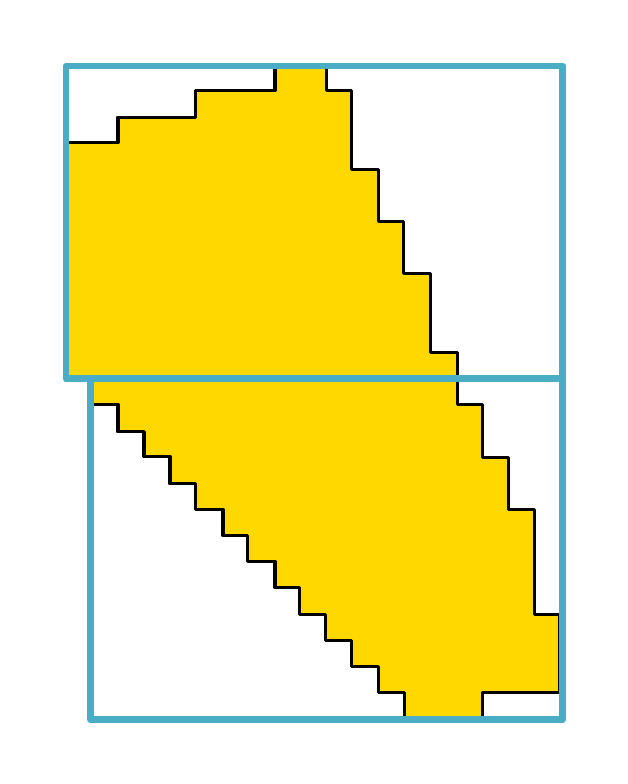} $\;$
   \includegraphics[width=0.1\linewidth]{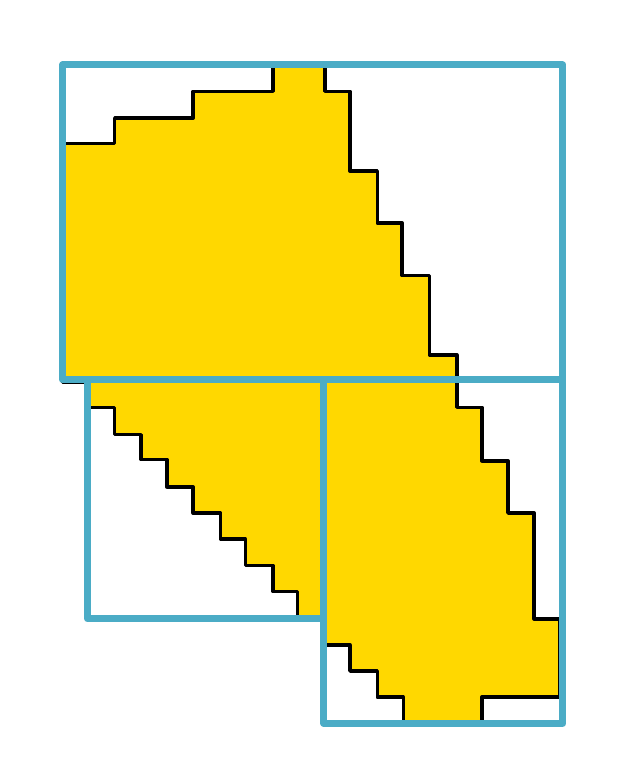} $\;$ 
   \includegraphics[width=0.1\linewidth]{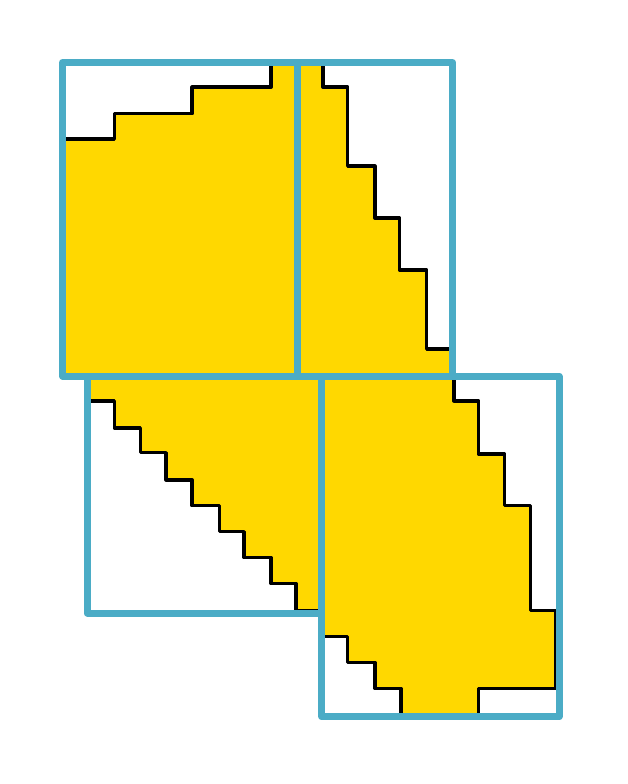} $\;$
   \includegraphics[width=0.1\linewidth]{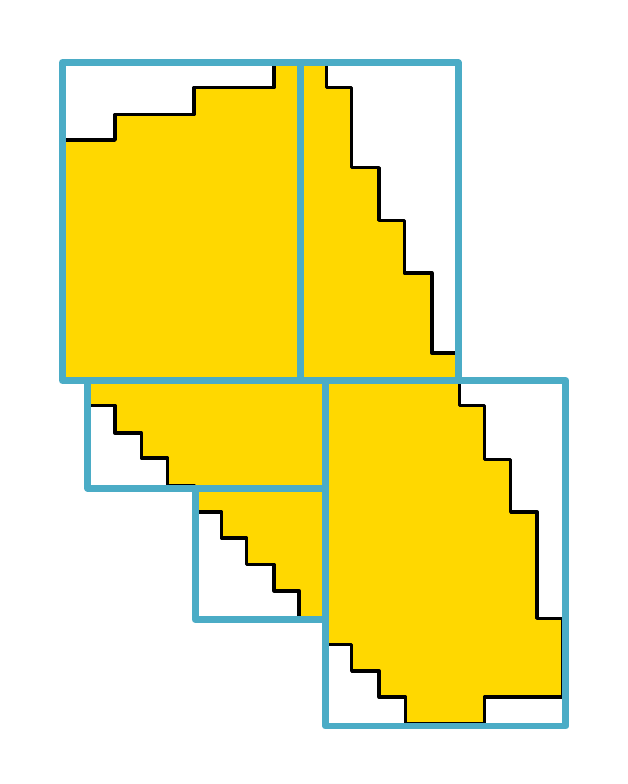} $\;$
   \includegraphics[width=0.1\linewidth]{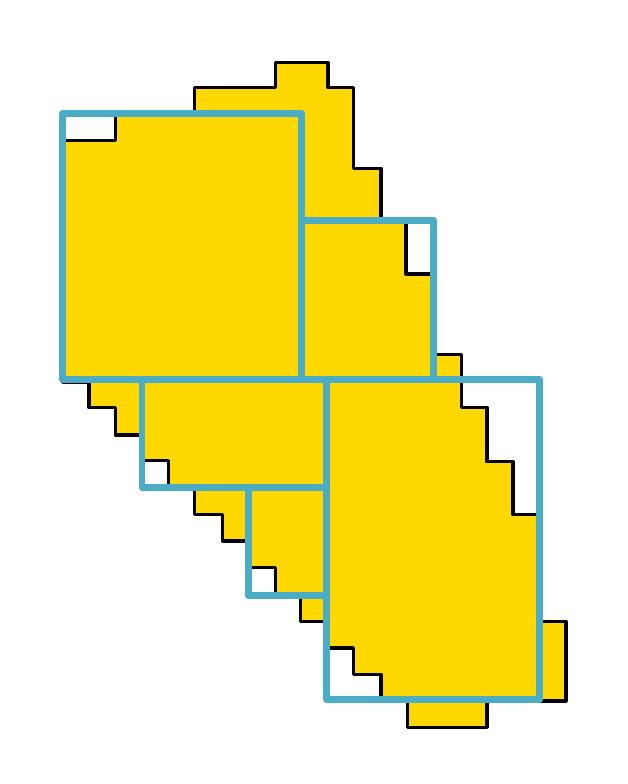}}
   \caption{Illustration for Split and Fit for $\rho=2$. \textcolor{red}{ The last image shows the output of the shrinking operation.}}
   \label{fig:decompose}
\end{figure}

\subsection{Fast Adaptive Silhouette Area Based Template Matching}

Fast Adaptive Silhouette Area Based Template matching (FAST) heuristic is proposed by \cite{MohrZachmann2010b}. It eventually determines an approximation of a given image by means of rectangles, which makes it a potential solution method for the solution of RBP. Unfortunately, the output set of rectangles do not exactly form a blanket since they are allowed to overlap, which means FAST can produce infeasible solutions for RBP in its original form.  Based on our experiments, we can say that the size of the overlaps is usually very small and does not have a serious impact on the application performance of the heuristic. Nevertheless, we have modified FAST to guarantee the feasiblity of its solutions, i.e. sets of non-overlapping rectangles, at the end.

Like SF, FAST finds rectangles iteratively, considering a \emph{benefit function}, which is similar to the fitness score~(\ref{equ:rect-fitness}), 
\begin{equation}
  f(\mathbf{r}) = \sum_{p_1 = r_{left}}^{r_{right}}\sum_{p_2 = r_{bottom}}^{r_{top}}(I_{p_1p_2}-\tau).
\end{equation}
Here, $\tau \in [0,1]$ is a parameter that controls the total uncovered area by penalizing covering a zero valued pixel. The method consists of two main steps. First, a rectangle of size $\sfrac{W}{2}\times \sfrac{H}{2}$ that maximizes $f(\mathbf{r})$ is found. $W$ and $H$ denote respectively the width and the height of the image, again.  If this rectangle has uncovered pixels, then the size is halved. This procedure is repeated until a rectangle that completely lies inside the binary image is found. \textcolor{red}{Then, the rectangle is grown until $f(\mathbf{r})$ starts to decrease}. The rectangle area is erased from the target image. These two steps are repeated until $K$ rectangles are obtained. 

\subsection{Constrained Simulated Annealing}\label{sec:csa}

Constrained Simulated Annealing CSA \citep{Wah2006} is originally proposed for solving a nonlinear programming problem consisting of the minimization of $f(\mathbf{x})$ subject to the equality constraints $h_i(\mathbf{x})=0$, $i=1,2,\dots,m$. CSA considers the Lagrange function \begin{equation}\label{eq40}
   L(\mathbf{x},\mbox{\boldmath$\lambda$}) = f(\mathbf{x}) +\sum_{i=1}^m\lambda_i h_i(\mathbf{x}).
\end{equation}
obtained by moving the constraints into the objective function with multipliers \mbox{\boldmath$\lambda$}$\in\mathbb{R}^m$.

The goal is to find $\mathbf{x}^*$ that minimizes $f(\mathbf{x})$ subject to the equality constraints $h_i(\mathbf{x})=0$, $i=1,2,\dots,m$ by finding $(\mathbf{x}^*,\mbox{\boldmath$\lambda$}^*)$ that minimizes (\ref{eq40}). Here, $\mbox{\boldmath$\lambda$}^*$ is the vector of Lagrange multipliers at which optimum solution $\mathbf{x}^*$ is obtained. In other words an equality constrained nonlinear minimization problem is solved in $\mathbf{x}$ by solving an unconstrained minimization problem in $\mathbf{y} = (\mathbf{x},\mbox{\boldmath$\lambda$})$, which is one of the classical research problems of nonlinear optimization.

In this outline, $\mathcal{N}(\mathbf{y})$, $G(\mathbf{y}' \mid \mathbf{y})$ and $A(\mathbf{y},\mathbf{y}',T)$ denote respectively the neighborhood of $\mathbf{y}$, generation probability of $\mathbf{y}'$ given $\mathbf{y}$  and the acceptance probability of the new point $\mathbf{y}'$. Then the new point $\mathbf{y}'$ can be obtained by changing $\mathbf{x}$ to $\mathbf{x}'$, i.e. at  $\mathbf{y}' = (\mathbf{x}',\mbox{\boldmath$\lambda$})$, or by changing $\mbox{\boldmath$\lambda$}$ to $\mbox{\boldmath$\lambda$}'$, i.e. at $\mathbf{y}' = (\mathbf{x},\mbox{\boldmath$\lambda$}')$. 
The algorithm is very similar to the conventional simulated annealing procedure, e.g. convergence condition can be extended to take the unchanged $\mathbf{y}$ into account in successive iterations.  The neighborhood $\mathcal{N}(\mathbf{y})$  and the acceptance probability $A(\mathbf{y},\mathbf{y}',T)$ for $\mathbf{y} = (\mathbf{x},\mbox{\boldmath$\lambda$})$ are defined as
\begin{equation}
\mathcal{N}(\mathbf{y}) =  \left\{ (\mathbf{x}',\mbox{\boldmath$\lambda$}) : \mathbf{x}' \in \mathcal{N}_1(\mathbf{x}) \right\} \cup \left\{ (\mathbf{x}, \mbox{\boldmath$\lambda$}') : \mbox{\boldmath$\lambda$}' \in \mathcal{N}_2(\mbox{\boldmath$\lambda$}) \right\},
\end{equation}
and
\begin{equation}
	A(\mathbf{y},\mathbf{y}',T) = \begin{cases} 
		\exp \left( - \frac{L(\mathbf{y}')-L(\mathbf{y})}{T} \right) & \text{if } \mathbf{y}' = (\mathbf{x}', \mbox{\boldmath$\lambda$})\\
		\\
		\exp \left( -\frac{L(\mathbf{y})-L(\mathbf{y}')}{T} \right) & \text{if } \mathbf{y}' = (\mathbf{x}, \mbox{\boldmath$\lambda$}'). \\
	\end{cases}
\end{equation}
In other words, at each iteration, a random point is generated by fixing $\mathbf{x}$ or \mbox{\boldmath$\lambda$}. 

CSA can be applied to the solution of RBP. The objective function defined as (\ref{eq3}) is $f(\mathbf{x})$ in the Lagrange function~(\ref{eq40}). There are two constraints in RBP: the maximum number of rectangles is fixed and the solution should not contain overlapping rectangles. These constraints can be incorporated into the Lagrange function as

\begin{eqnarray*}
	h_1(\mathcal{R}) &  = & \max(0, |\mathcal{R}|-K), \\
	h_2(\mathbf{r}) &  = & \sum_{\substack{\mathbf{r}^i,\mathbf{r}^j \in \mathcal{R} \\ \mathbf{r}^i \neq \mathbf{r}^j}} \sum_{p_1}\sum_{p_2}r^i_{p_1p_2}\times r^j_{p_1p_2},
\end{eqnarray*}
with multipliers $\lambda_1$ and $\lambda_2$. $h_1(\mathcal{R})$ is for penalizing the excess in the number of rectangles; it increases as the number of rectangles in the blanket exceeds the upper bound $K$. $h_2(\mathbf{r})$ is for penalizing overlapping  rectangles; its value is obtained by counting pixels in overlapping rectangle pairs. 

In this work, five primitive operations are realized using $G(\mathbf{y}' \mid \mathbf{y})$: \emph{grow}, \emph{shrink}, \emph{split}, \emph{delete} and \emph{create}. Growing/shrinking operation moves one of the edges of a rectangle $\mathbf{r}\in\mathcal{R}$ by one pixel outside/inside. Splitting operation partitions a rectangle $\mathbf{r}\in \mathcal{R}$ into two rectangles $\mathbf{r}^1$ and $\mathbf{r}^2$ ($\mathbf{r}= \mathbf{r}^1 \cup \mathbf{r}^2$ and $\mathbf{r}^1 \cap \mathbf{r}^2 = \emptyset$). Deleting simply removes a rectangle $\mathbf{r}$ from $\mathcal{R}$. Creating adds a rectangle $\mathbf{r}$ to $\mathcal{R}$ that resides in the image region (i.e. $1 \le r_\text{left} <  r_\text{right} \le W$, $1 \le r_\text{top} < r_\text{bottom} \le H$). These operations are illustrated in Figure~\ref{fig:proposals}. 

\begin{figure}[!h]
	\begin{center}   
		\includegraphics[width=0.8\linewidth]{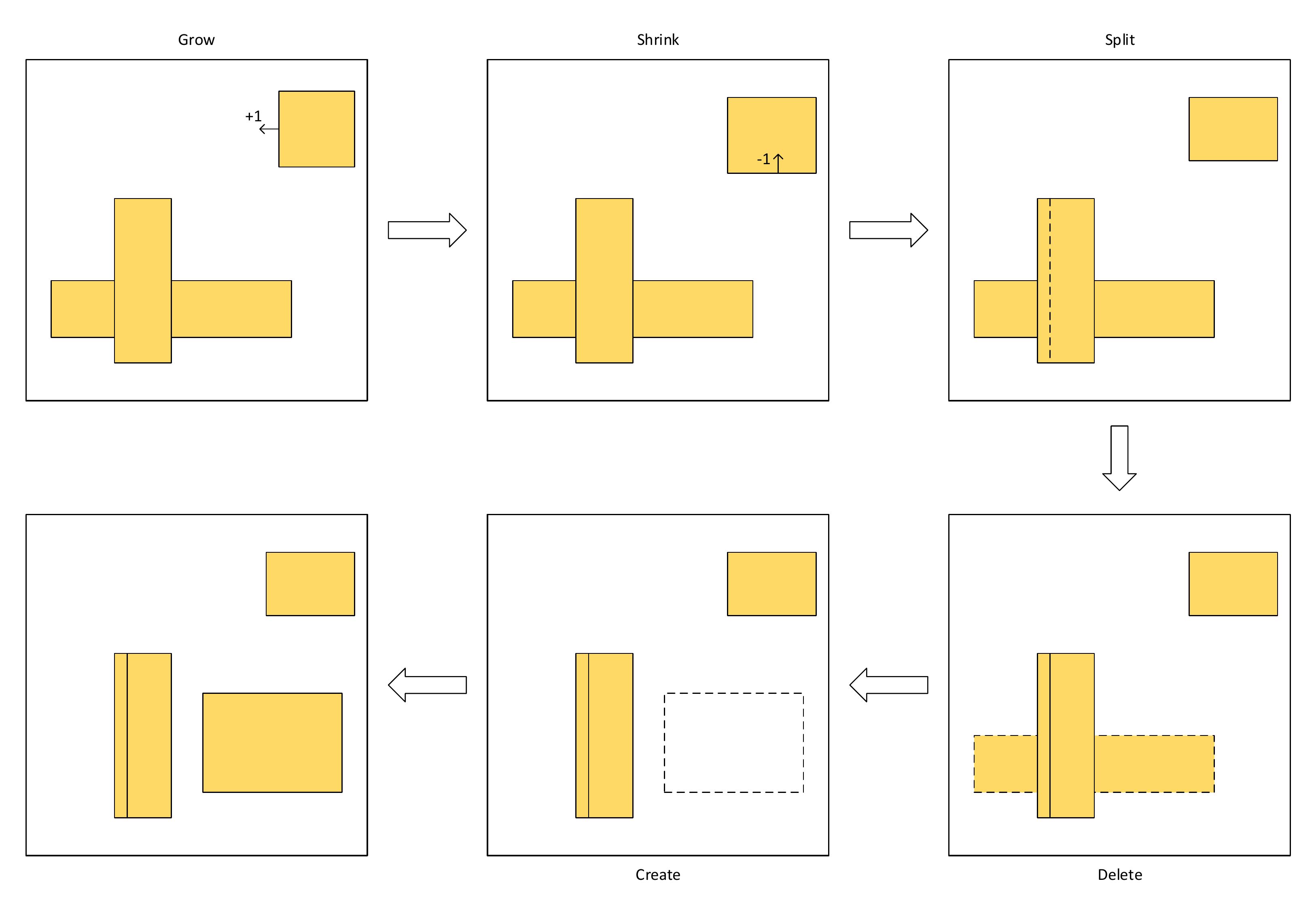}
   		\caption{Illustration for grow, shrink, split, delete and create operations} \label{fig:proposals}
	\end{center}
\end{figure}

\section{Computational results}\label{sec:experiments}

In this section we report the results obtained in the computational tests made for assessing the performance of the solution methods. 

\subsection{Test bed}
For benchmarking, five different groups of binary images are generated; they all have different properties. The benchmark instances can be seen in Figure~\ref{fig:benchmark-1}--Figure~\ref{fig:benchmark-5}. The labels and resolutions (i.e. the number of pixels in each dimension) are given below the images. The first group is \emph{ideal human silhouette}s selected from the work of ~\cite{DemirozSalahAkarun2014} on fall detection and tracking. They are simple nonconvex polygons and the size of the images are small compared to the other groups. Since the structure of all shapes used in the study are very similar, only four representative shapes are selected. The second group of benchmark images are taken from~\citep{Chan2014} where the authors created a similar benchmark for \emph{mask fracturing}, a process where complex shapes are translated into the union of simpler shapes called \textit{shots} during integrated circuit layout production. This group has the largest images and each image represents a single region. The third group contains images  generated artificially to capture certain shape properties that the first two groups do not have. This group contains convex regions, disconnected regions, regions with holes and nested disconnected regions. The fourth group of images is a subset of MPEG-7 shape dataset~\citep{latecki2000}, which are cropped and resized for our experiments. Each category in this group has 20 instances. Because the purpose of the MPEG-7 shape dataset is to evaluate shape similarity measures, we have found that using only a subset of the MPEG-7 dataset is sufficient for our study. The selected categories also capture different shape properties like the previous group. Finally, the last group consists of images belonging to 10 realistic nesting problems. They \textcolor{red}{were } obtained from leather garment and furniture with defects on the master surfaces and \textcolor{red}{have been used} to generate the results given in Table 4 of \cite{BaldacciBoschettiGanovelliManiezzo2014}. We rasterized the images included in the files in DXF format that were sent by \cite{GanovelliManiezzo2018}, using KABEJA  library written in JAVA \citep{Kabeja2018}. In short, the test bed consists of $147$ instances and each one is solved for five distinct blanket sizes (i.e. the number of rectangles in the blanket); $K\in \{3,5,10,15,20\}$. This results in $147 \times 5 = 735$ test runs for each solution method. 

\newcommand{\scale}{0.3}
\newcommand{\twidth}{3.0cm}

\begin{figure}[b]
	\centering
	\subcaptionbox*{avatar1\\18$\times$15}[\twidth]
		{\includegraphics[width=\paperwidth*\real{0.18}*\real{\scale}]{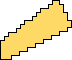}}
	\subcaptionbox*{avatar2\\19$\times$25}[\twidth]
		{\includegraphics[width=\paperwidth*\real{0.19}*\real{\scale}]{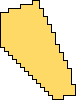}}
	\subcaptionbox*{avatar3\\28$\times$17}[\twidth]
		{\includegraphics[width=\paperwidth*\real{0.28}*\real{\scale}]{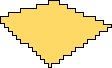}}
	\subcaptionbox*{avatar4\\21$\times$19}[\twidth]
		{\includegraphics[width=\paperwidth*\real{0.21}*\real{\scale}]{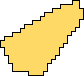}}
	\caption{Selected binary masks of \emph{ideal human silhouette}s from~\cite{DemirozSalahAkarun2014}}
	\label{fig:benchmark-1}
\end{figure}

\renewcommand{\scale}{0.05}
\renewcommand{\twidth}{2cm}

\begin{figure}
	\centering
	\subcaptionbox*{artificial1\\70$\times$70}[\twidth]
		{\includegraphics[width=\paperwidth*\real{0.70}*\real{\scale}]{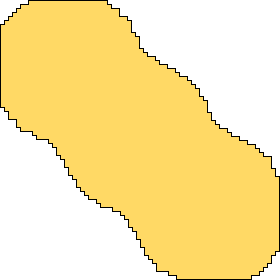}}
	\subcaptionbox*{artificial2\\309$\times$102}
		{\includegraphics[width=\paperwidth*\real{3.09}*\real{\scale}]{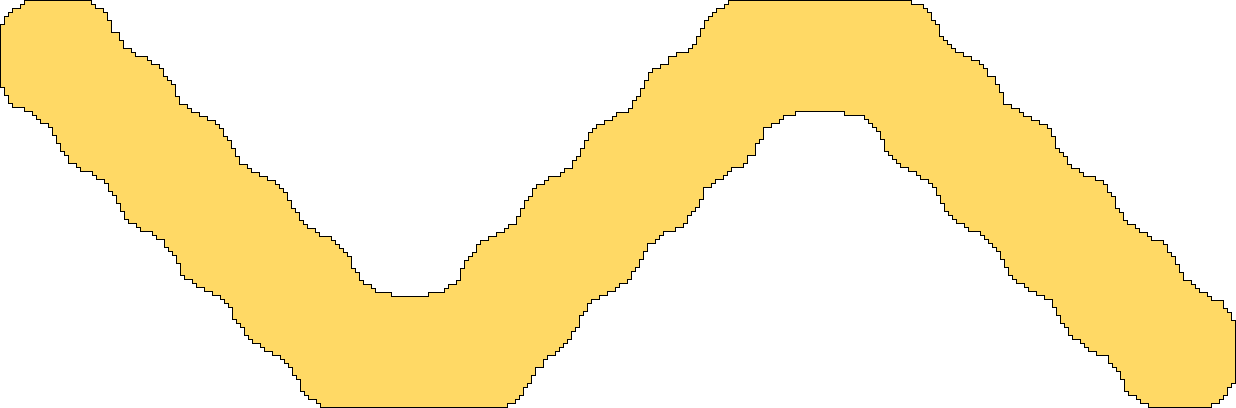}}
	\subcaptionbox*{artificial3\\170$\times$143}
		{\includegraphics[width=\paperwidth*\real{1.70}*\real{\scale}]{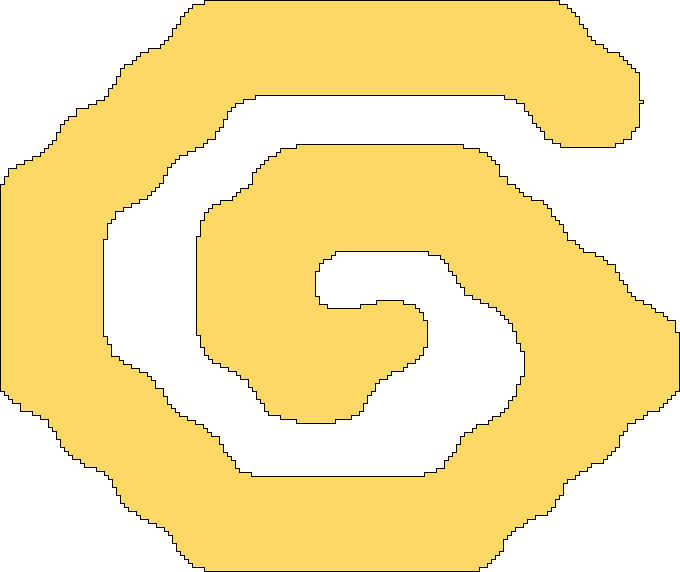}}
	\subcaptionbox*{artificial4\\145$\times$121}[\twidth]
		{\includegraphics[width=\paperwidth*\real{1.45}*\real{\scale}]{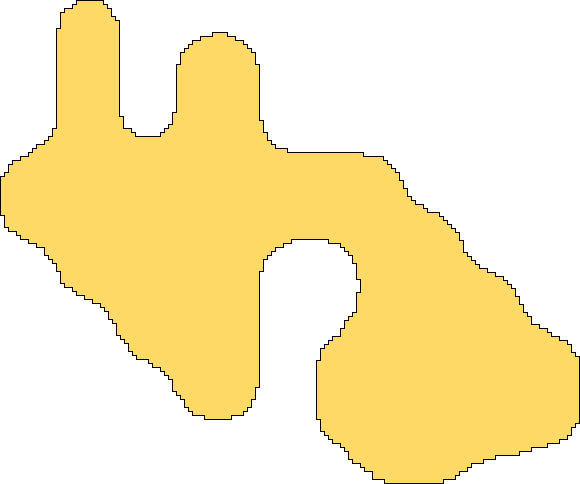}}
	\subcaptionbox*{artificial5\\65$\times$61}[\twidth]
		{\includegraphics[width=\paperwidth*\real{0.65}*\real{\scale}]{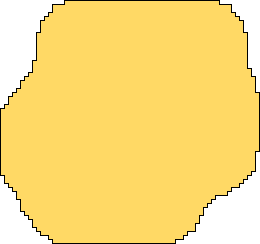}}
	\subcaptionbox*{realistic1\\110$\times$60}[\twidth]
		{\includegraphics[width=\paperwidth*\real{1.10}*\real{\scale}]{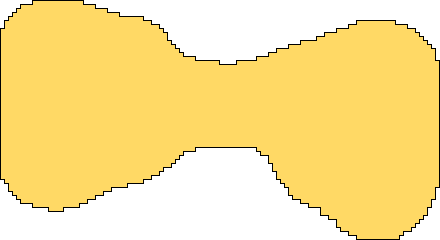}}
	\subcaptionbox*{realistic2\\162$\times$190}[\twidth]
		{\includegraphics[width=\paperwidth*\real{1.62}*\real{\scale}]{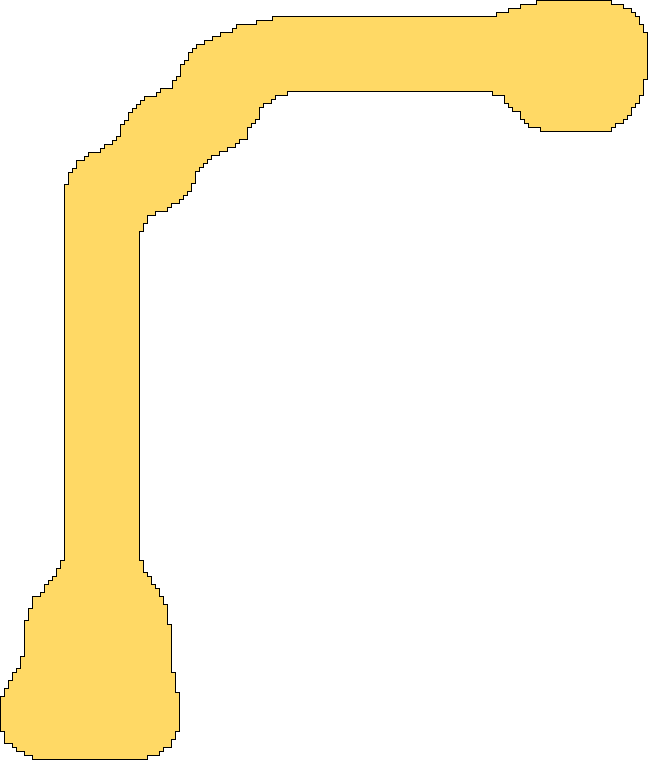}}
	\subcaptionbox*{realistic3\\80$\times$130}[\twidth]
		{\includegraphics[width=\paperwidth*\real{0.80}*\real{\scale}]{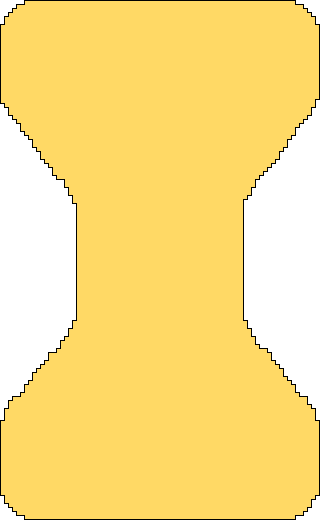}}
	\subcaptionbox*{realistic4\\87$\times$258}[\twidth]
		{\includegraphics[width=\paperwidth*\real{0.87}*\real{\scale}]{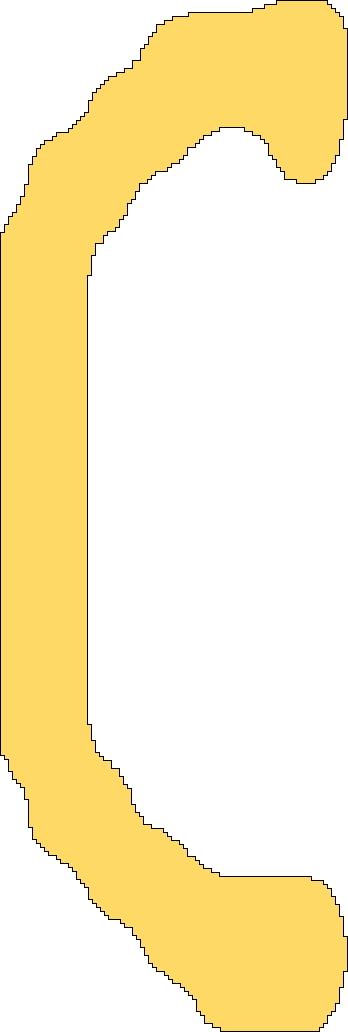}}
	\subcaptionbox*{realistic5\\152$\times$157}[\twidth]
		{\includegraphics[width=\paperwidth*\real{1.52}*\real{\scale}]{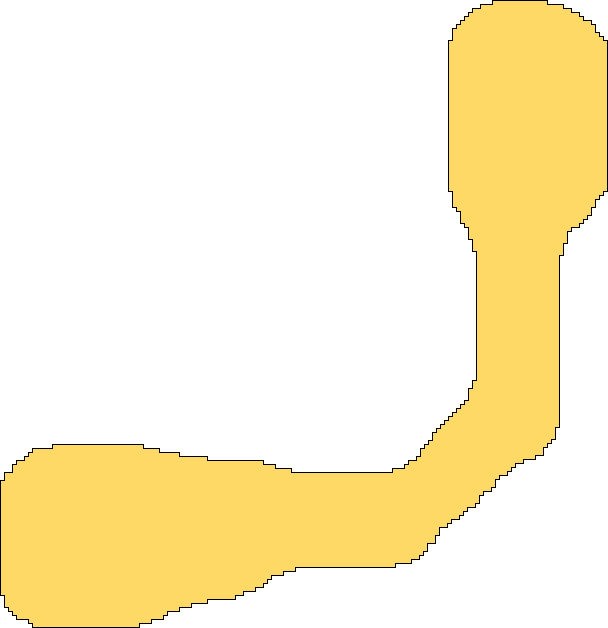}}
	\subcaptionbox*{typical1\\111$\times$55}[\twidth]
		{\includegraphics[width=\paperwidth*\real{1.11}*\real{\scale}]{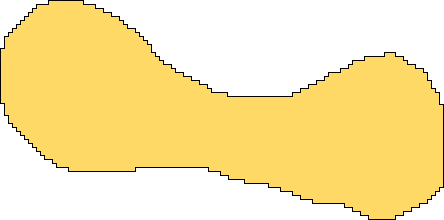}}
	\subcaptionbox*{typical10\\47$\times$190}[\twidth]
		{\includegraphics[width=\paperwidth*\real{0.47}*\real{\scale}]{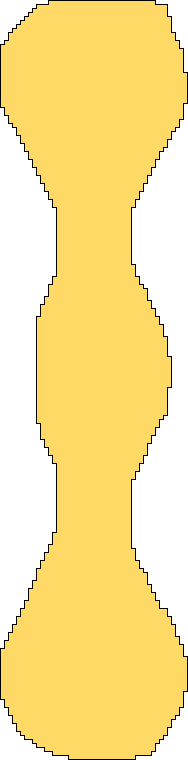}}
	\subcaptionbox*{typical2\\180$\times$222}[\twidth]
		{\includegraphics[width=\paperwidth*\real{1.80}*\real{\scale}]{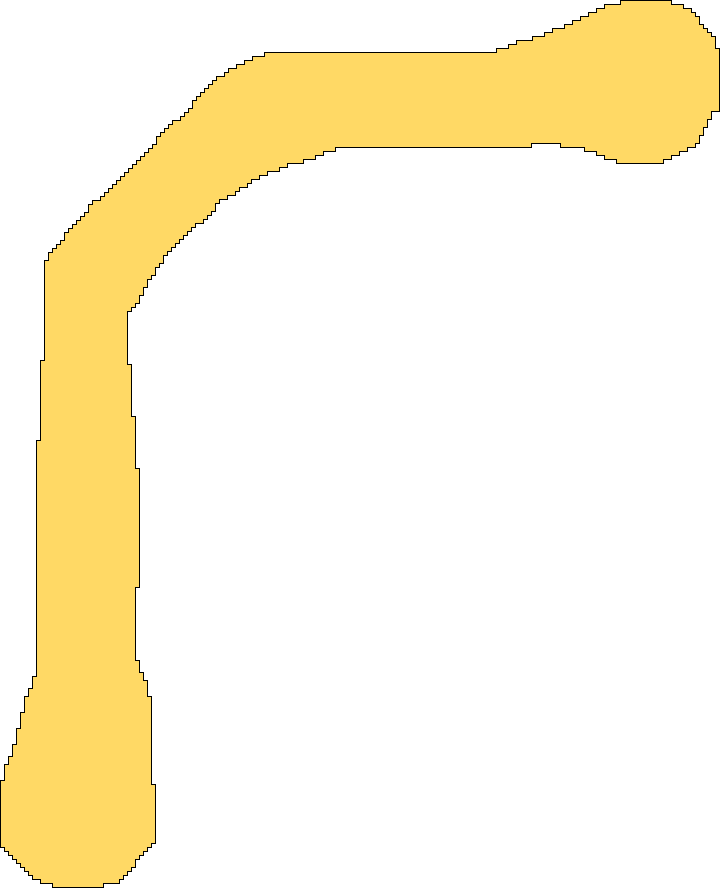}}
	\subcaptionbox*{typical3\\75$\times$95}[\twidth]
		{\includegraphics[width=\paperwidth*\real{0.75}*\real{\scale}]{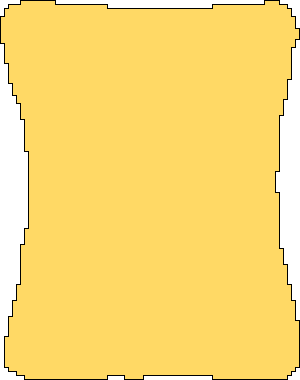}}
	\subcaptionbox*{typical4\\105$\times$280}[\twidth]
		{\includegraphics[width=\paperwidth*\real{1.05}*\real{\scale}]{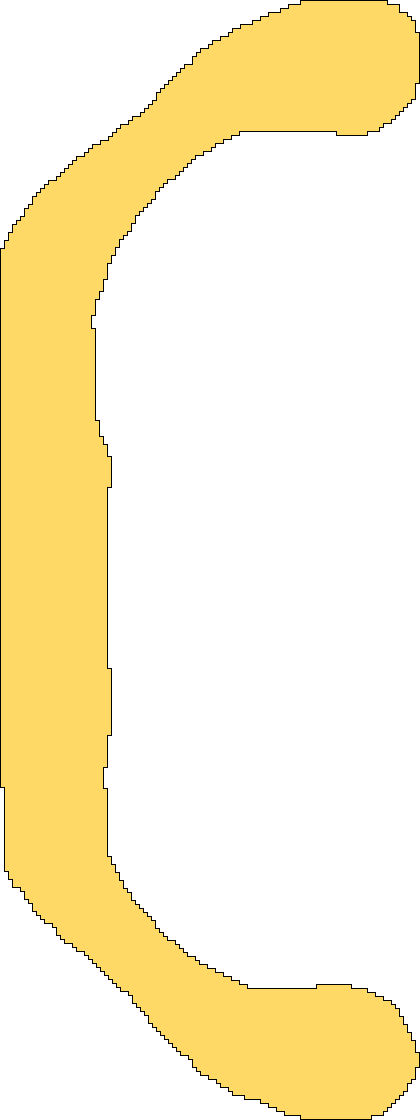}}
	\subcaptionbox*{typical5\\143$\times$125}[\twidth]
		{\includegraphics[width=\paperwidth*\real{1.43}*\real{\scale}]{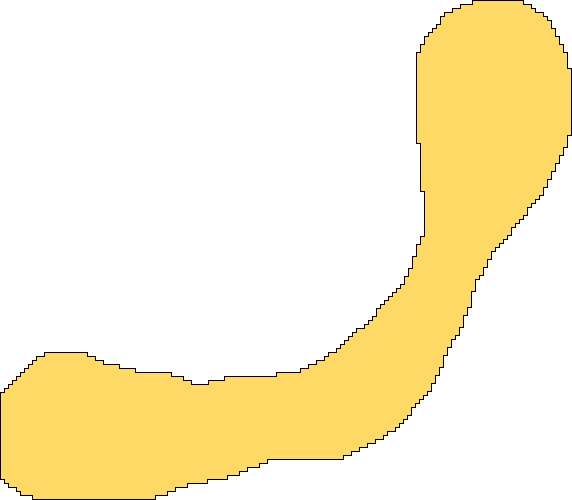}}
	\subcaptionbox*{typical6\\60$\times$60}[\twidth]
		{\includegraphics[width=\paperwidth*\real{0.60}*\real{\scale}]{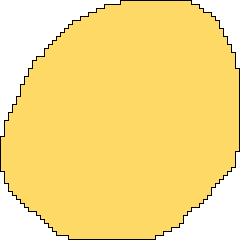}}
	\subcaptionbox*{typical7\\70$\times$75}[\twidth]
		{\includegraphics[width=\paperwidth*\real{0.70}*\real{\scale}]{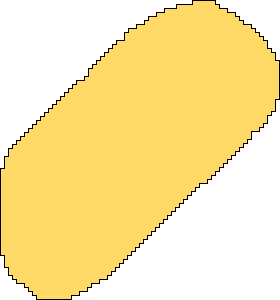}}
	\subcaptionbox*{typical8\\132$\times$137}[\twidth]
		{\includegraphics[width=\paperwidth*\real{1.32}*\real{\scale}]{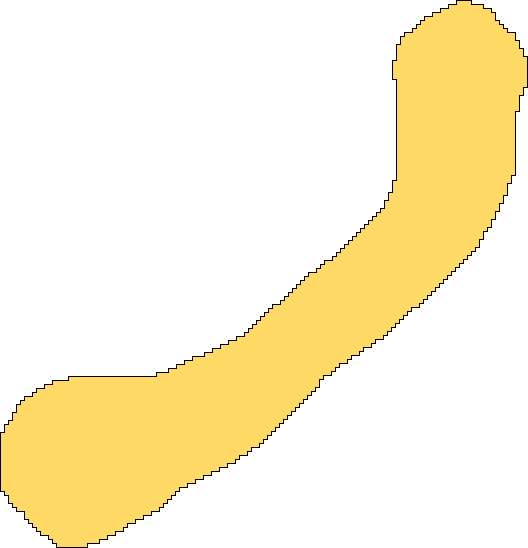}}
	\caption{Binary masks used for \emph{mask fracturing} in~\citep{Chan2014}}
	\label{fig:benchmark-2}
\end{figure}

\renewcommand{\scale}{0.03}
\renewcommand{\twidth}{2cm}

\begin{figure}
	\centering
	\subcaptionbox*{toy1\\160$\times$104}[\twidth]
		{\includegraphics[width=\paperwidth*\real{1.60}*\real{\scale}]{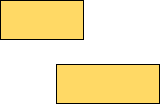}}
	\subcaptionbox*{toy2\\160$\times$140}[\twidth]
		{\includegraphics[width=\paperwidth*\real{1.60}*\real{\scale}]{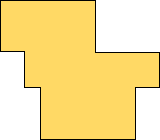}}
	\subcaptionbox*{toy3\\156$\times$160}[\twidth]
		{\includegraphics[width=\paperwidth*\real{1.56}*\real{\scale}]{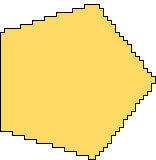}}
	\subcaptionbox*{toy4\\160$\times$140}[\twidth]
		{\includegraphics[width=\paperwidth*\real{1.60}*\real{\scale}]{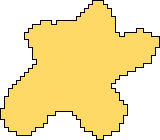}}
	\subcaptionbox*{toy5\\112$\times$160}[\twidth]
		{\includegraphics[width=\paperwidth*\real{1.12}*\real{\scale}]{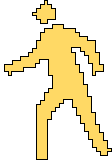}}
	\subcaptionbox*{toy6\\400$\times$348}[\twidth]
		{\includegraphics[width=\paperwidth*\real{4.00}*\real{\scale}]{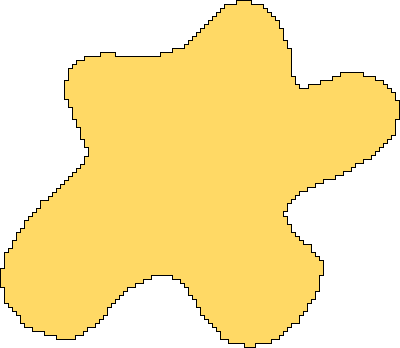}}
	\subcaptionbox*{toy7\\160$\times$156}[\twidth]
		{\includegraphics[width=\paperwidth*\real{1.60}*\real{\scale}]{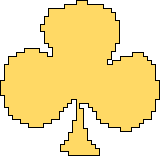}}
	\subcaptionbox*{toy8\\152$\times$160}[\twidth]
		{\includegraphics[width=\paperwidth*\real{1.52}*\real{\scale}]{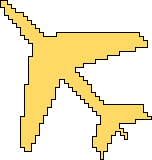}}
	\subcaptionbox*{toy9\\160$\times$128}[\twidth]
		{\includegraphics[width=\paperwidth*\real{1.60}*\real{\scale}]{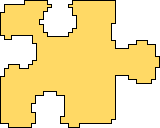}}
	\subcaptionbox*{toy10\\200$\times$200}[\twidth]
		{\includegraphics[width=\paperwidth*\real{2.00}*\real{\scale}]{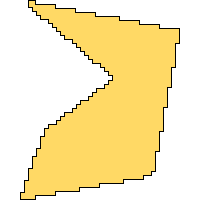}}
	\subcaptionbox*{toy11\\160$\times$160}[\twidth]
		{\includegraphics[width=\paperwidth*\real{1.60}*\real{\scale}]{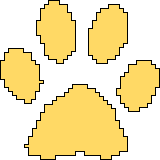}}
	\subcaptionbox*{toy12\\160$\times$140}[\twidth]
		{\includegraphics[width=\paperwidth*\real{1.60}*\real{\scale}]{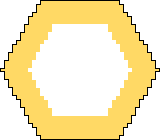}}
	\subcaptionbox*{toy13\\160$\times$140}[\twidth]
		{\includegraphics[width=\paperwidth*\real{1.60}*\real{\scale}]{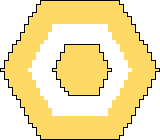}}
	\subcaptionbox*{toy14\\180$\times$200}[\twidth]
		{\includegraphics[width=\paperwidth*\real{1.80}*\real{\scale}]{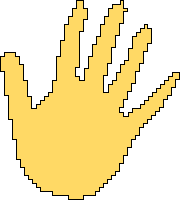}}
	\caption{Generated binary masks}
	\label{fig:benchmark-3}
\end{figure}

\renewcommand{\scale}{0.05}
\renewcommand{\twidth}{2.5cm}

\begin{figure}
	\centering
	\subcaptionbox*{bat\\121$\times$118}[\twidth]
		{\includegraphics[width=\paperwidth*\real{1.21}*\real{\scale}]{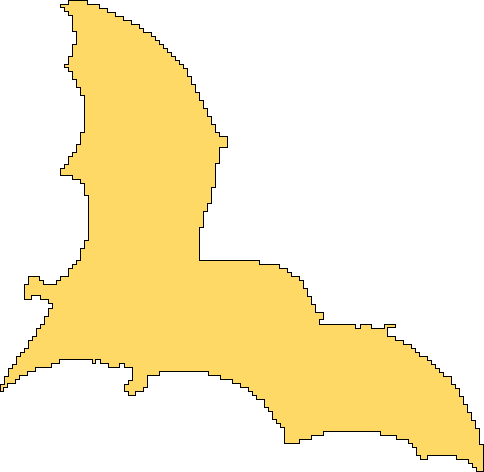}}
	\subcaptionbox*{device5\\139$\times$142}[\twidth]
		{\includegraphics[width=\paperwidth*\real{1.39}*\real{\scale}]{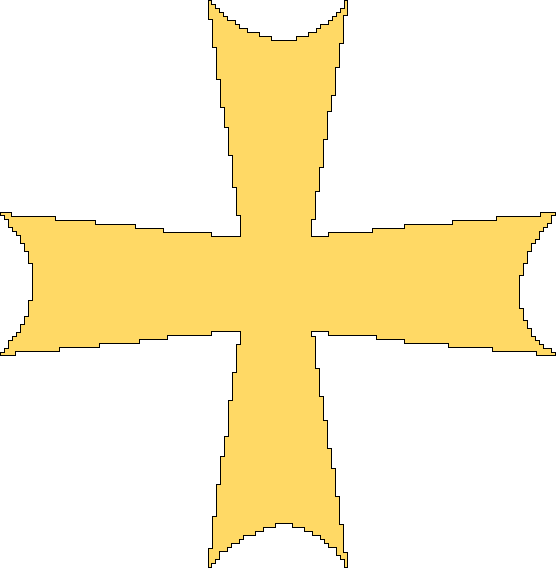}}
	\subcaptionbox*{dog\\174$\times$130}
		{\includegraphics[width=\paperwidth*\real{1.74}*\real{\scale}]{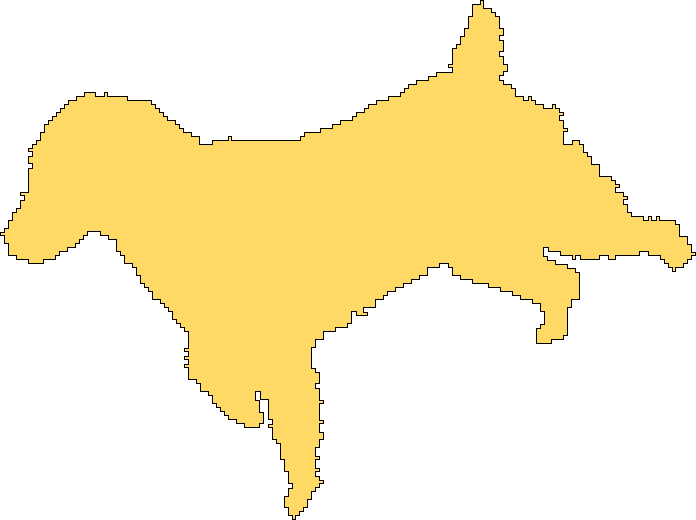}}
	\subcaptionbox*{key\\150$\times$72}
		{\includegraphics[width=\paperwidth*\real{1.50}*\real{\scale}]{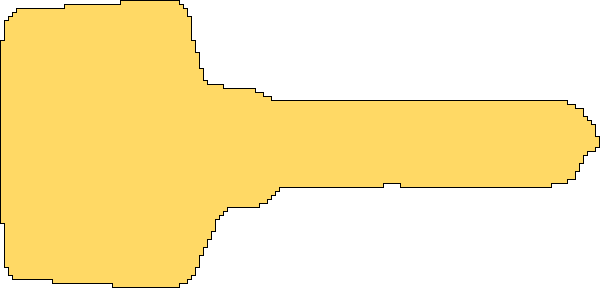}}
	\subcaptionbox*{misk\\166$\times$201}[\twidth]
		{\includegraphics[width=\paperwidth*\real{1.66}*\real{\scale}]{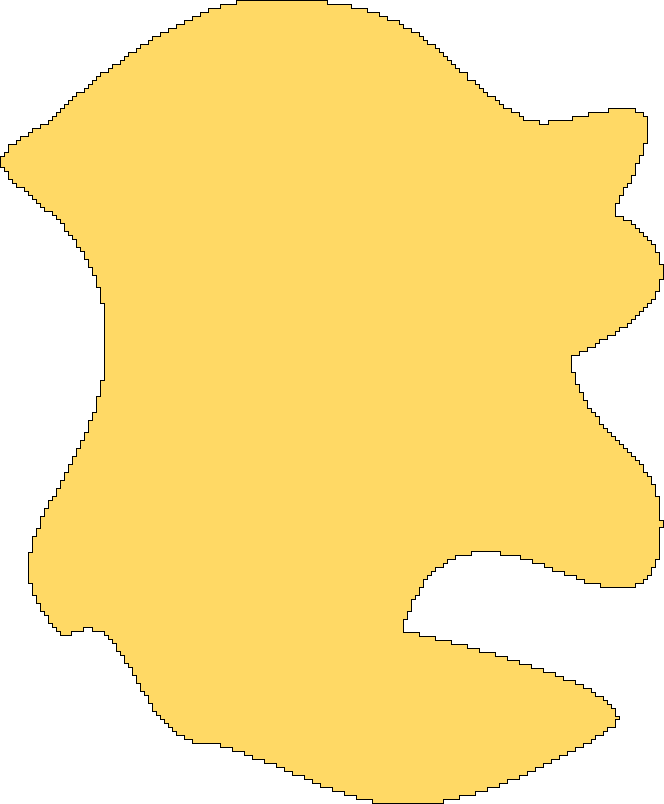}}
	\caption{Subset of MPEG7 shape dataset categories}
	\label{fig:benchmark-4}
\end{figure}	

\renewcommand{\scale}{0.03}
\renewcommand{\twidth}{3cm}

\begin{figure}
	\centering
	\subcaptionbox*{79510\\242$\times$278}[\twidth]
		{\includegraphics[width=\paperwidth*\real{2.42}*\real{\scale}]{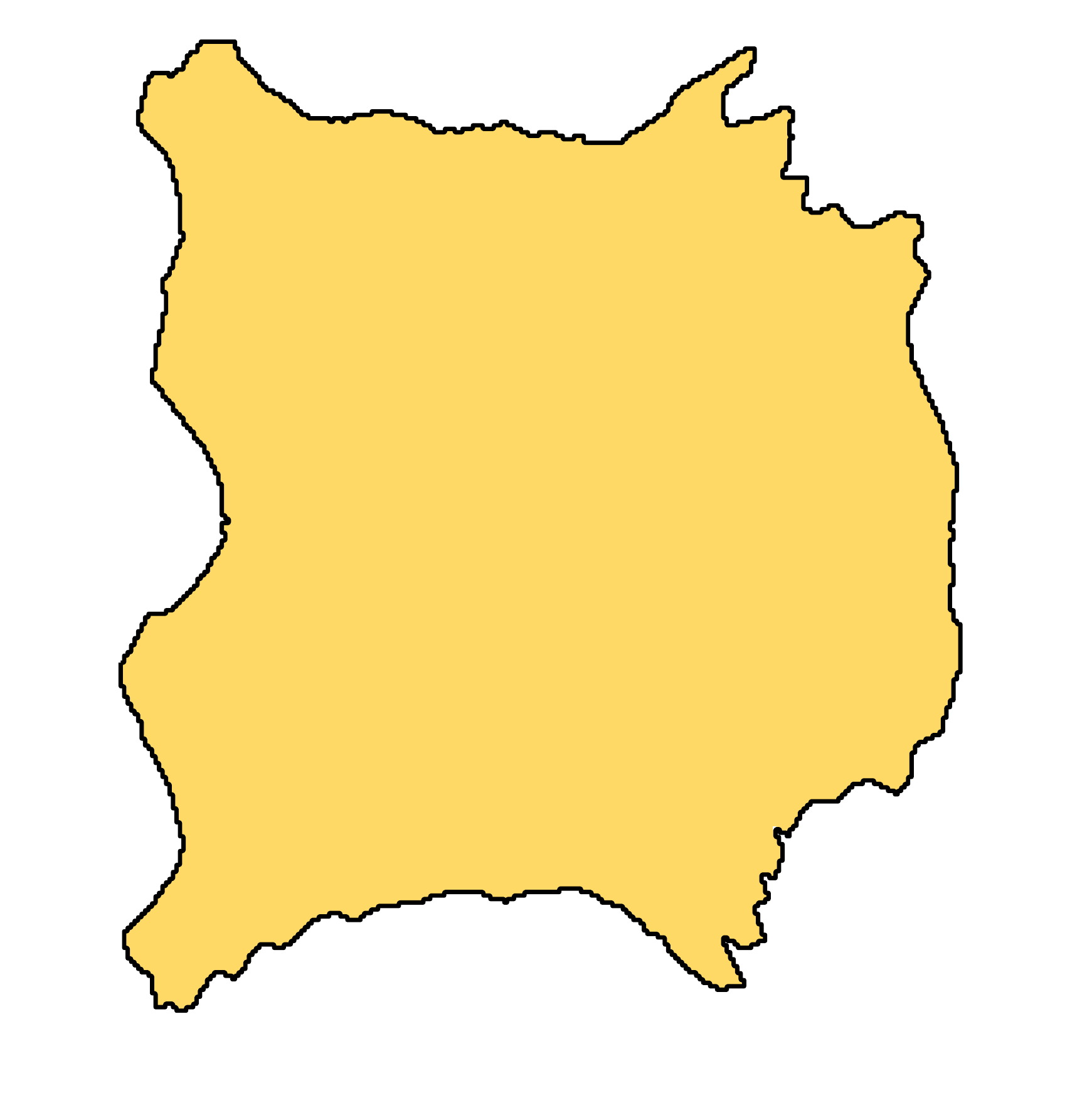}}
	\subcaptionbox*{79611\\263$\times$268}[\twidth]
		{\includegraphics[width=\paperwidth*\real{2.63}*\real{\scale}]{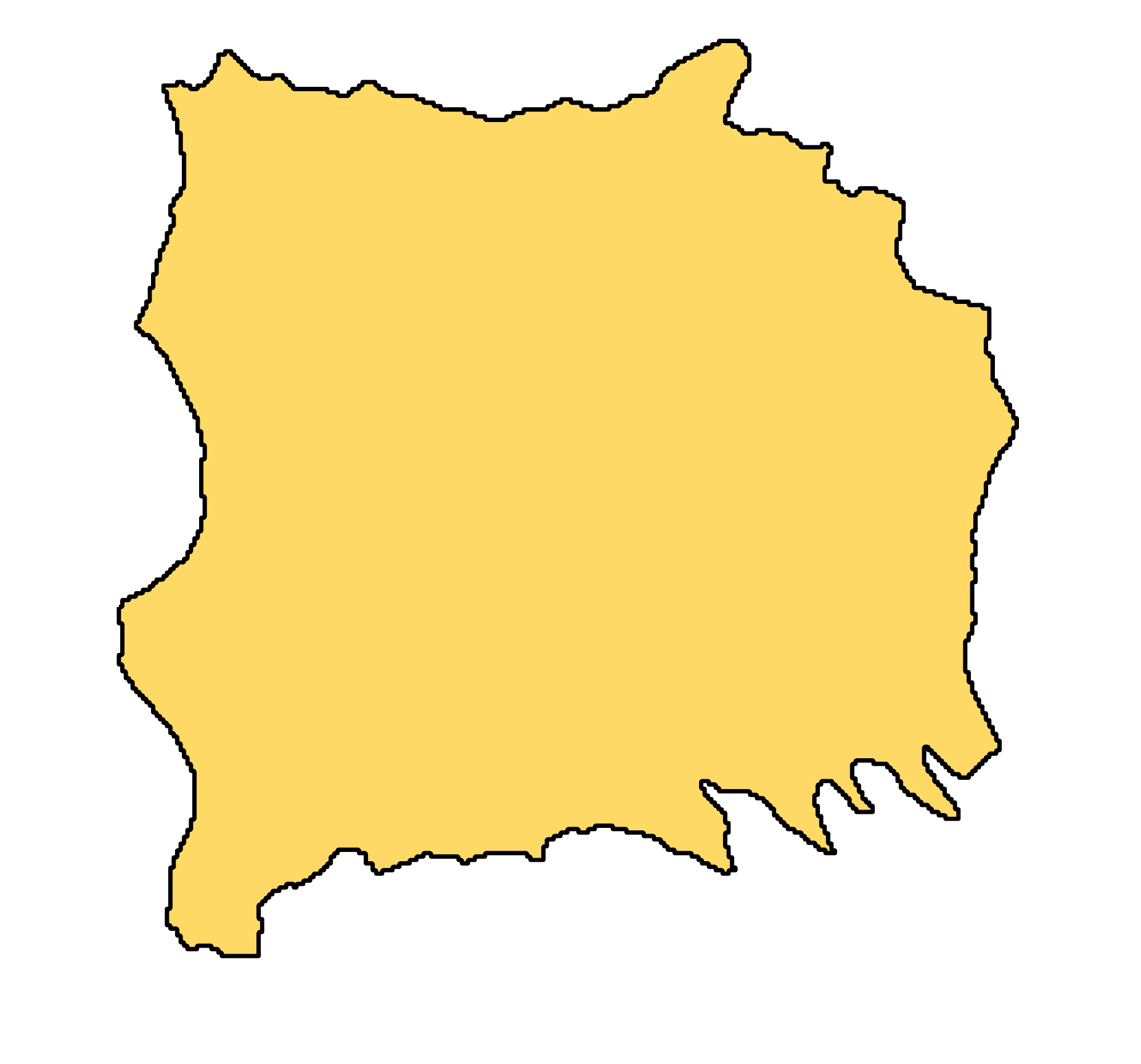}}
	\subcaptionbox*{79712\\230$\times$268}[\twidth]
		{\includegraphics[width=\paperwidth*\real{2.30}*\real{\scale}]{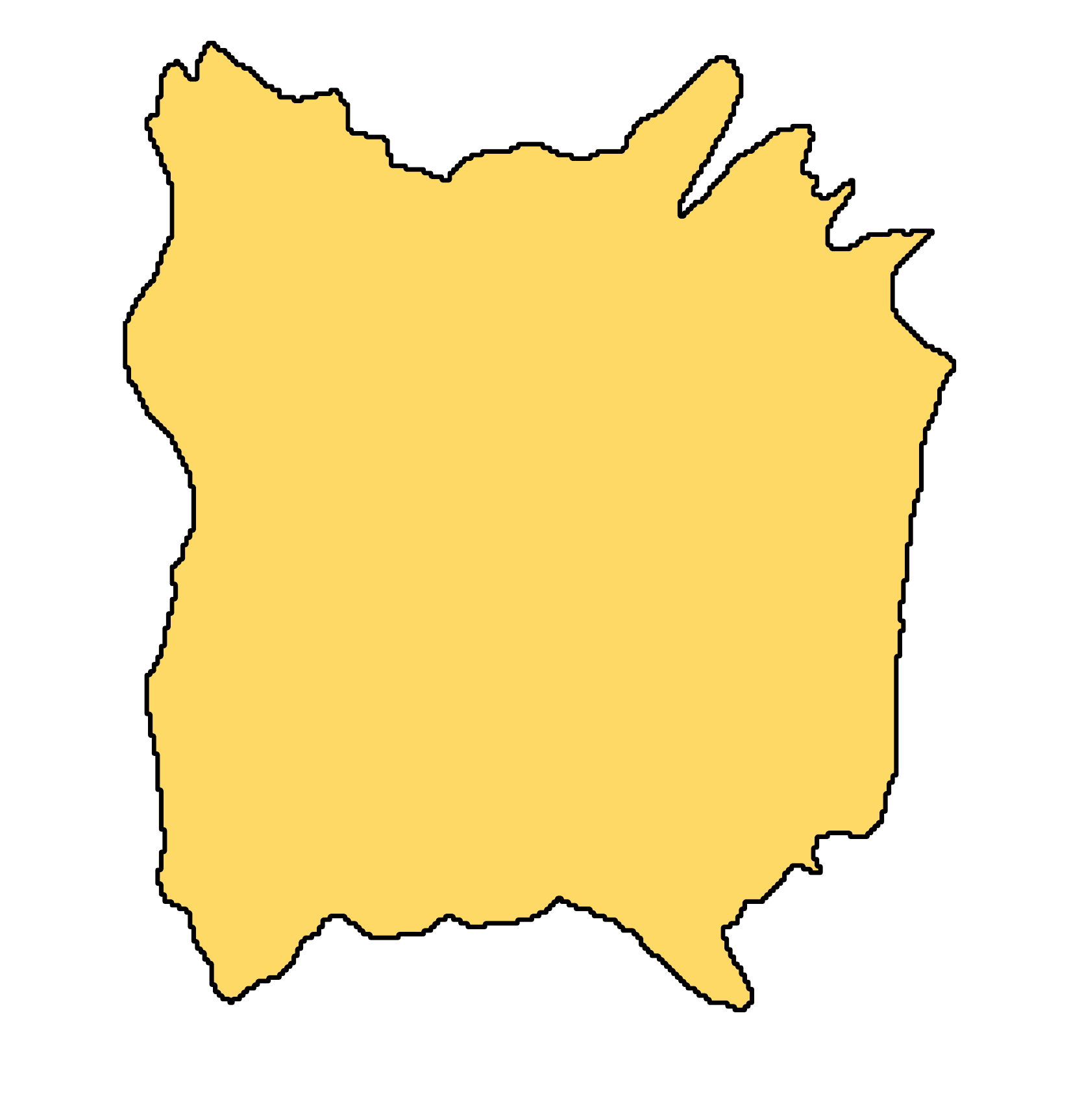}}
	\subcaptionbox*{79813\\265$\times$263}[\twidth]
		{\includegraphics[width=\paperwidth*\real{2.65}*\real{\scale}]{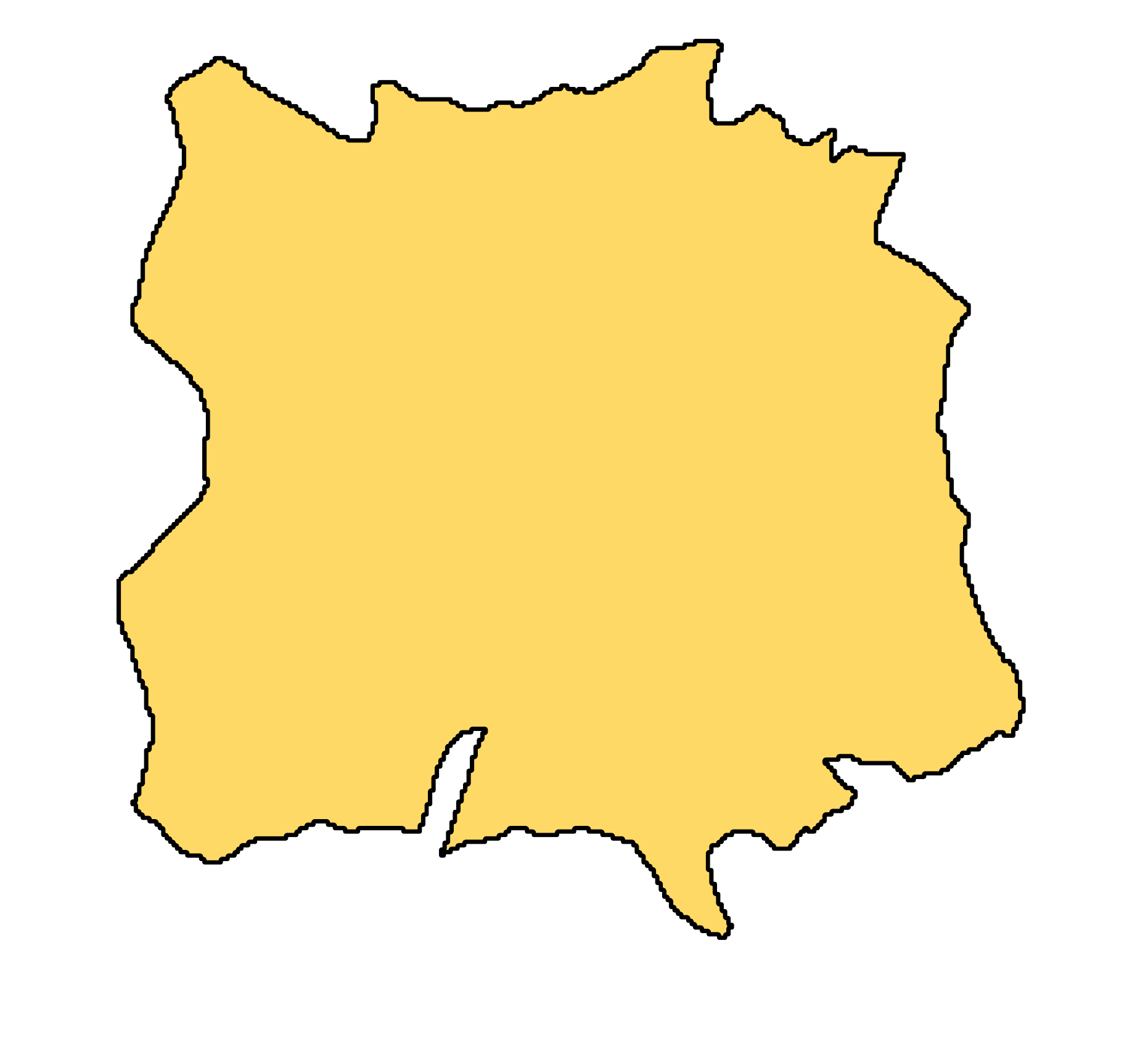}}
	\subcaptionbox*{79914\\267$\times$270}[\twidth]
		{\includegraphics[width=\paperwidth*\real{2.67}*\real{\scale}]{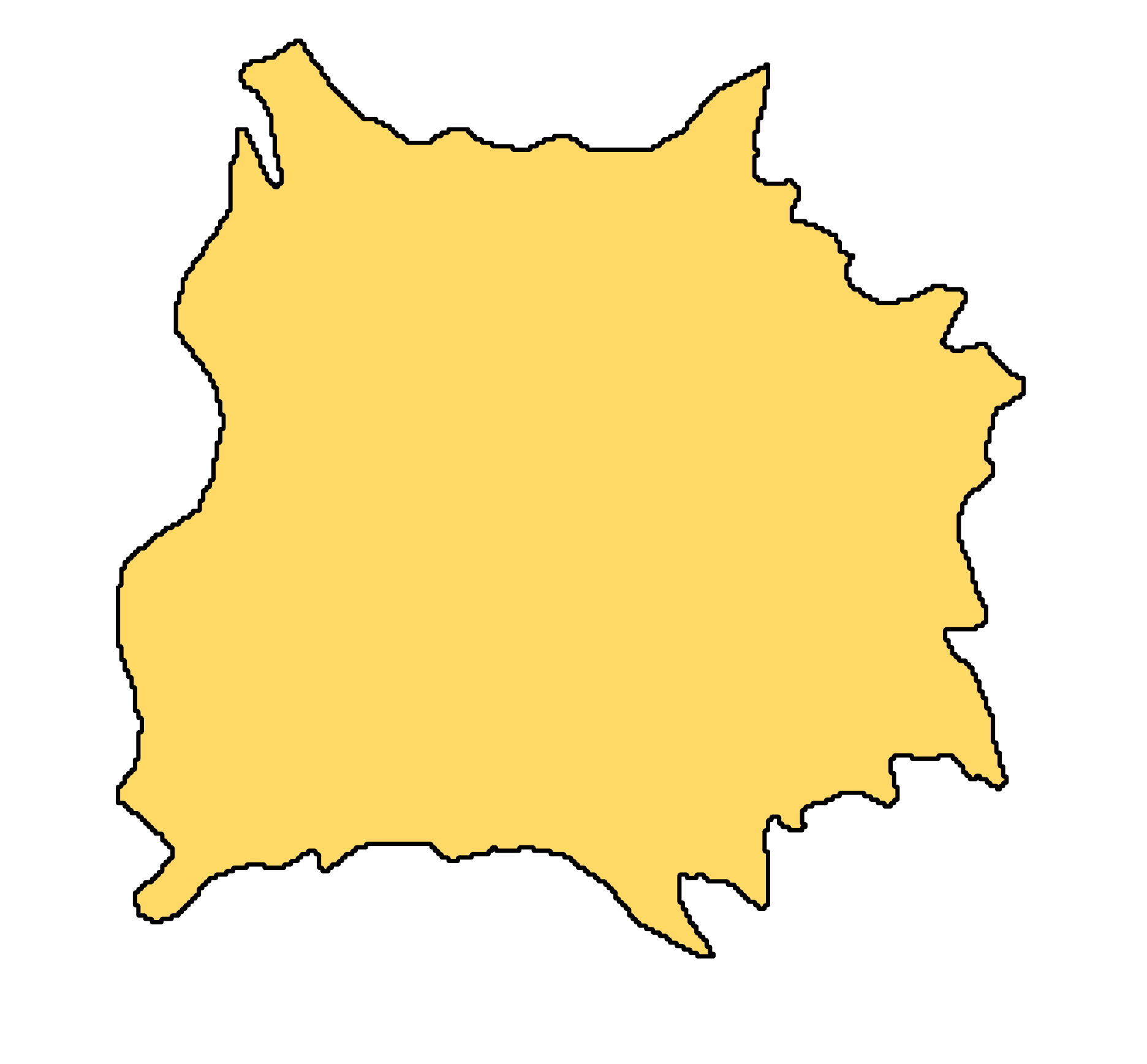}}
	\subcaptionbox*{79915\\250$\times$245}[\twidth]
		{\includegraphics[width=\paperwidth*\real{2.50}*\real{\scale}]{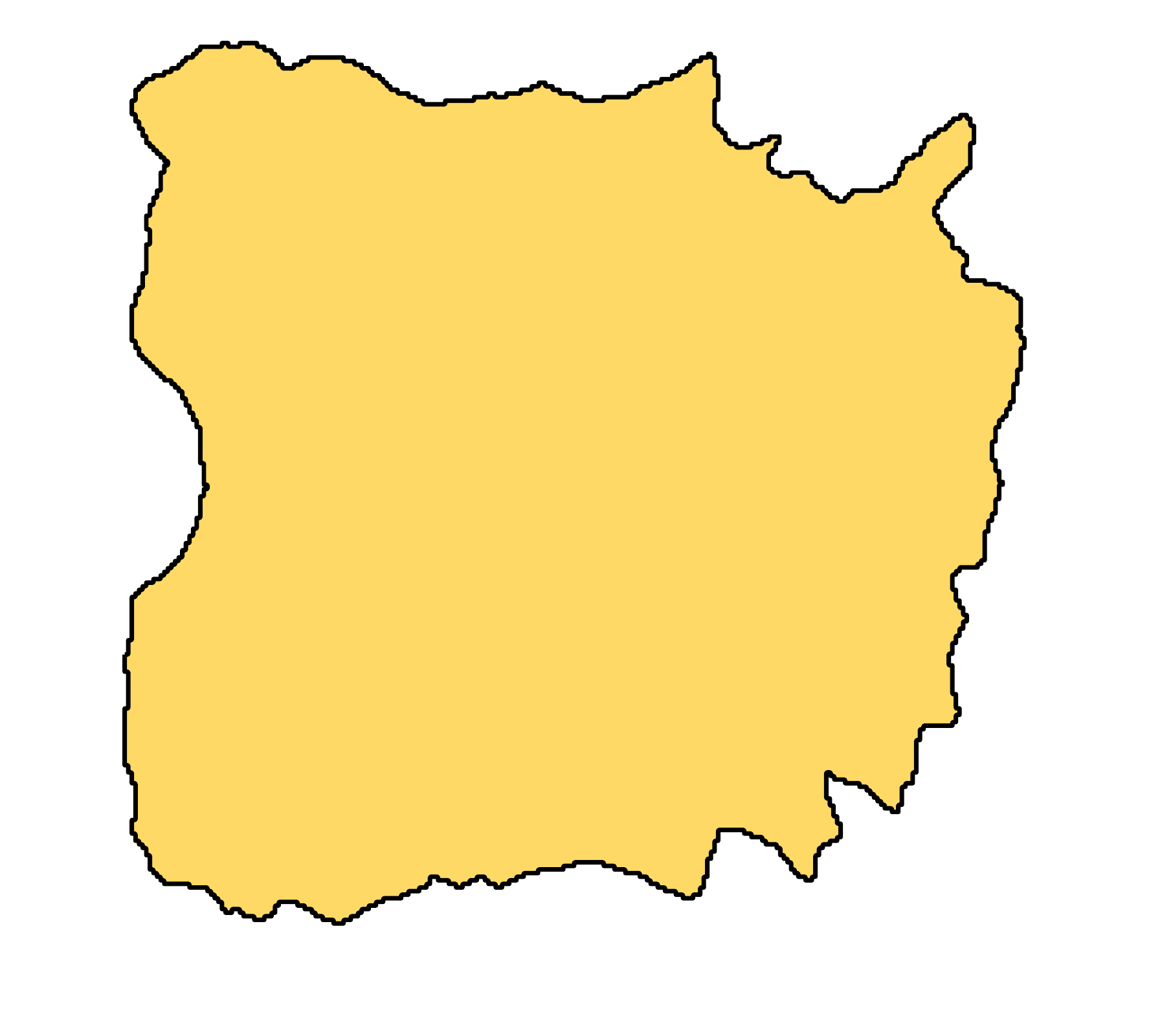}}
	\subcaptionbox*{79916\\261$\times$278}[\twidth]
		{\includegraphics[width=\paperwidth*\real{2.61}*\real{\scale}]{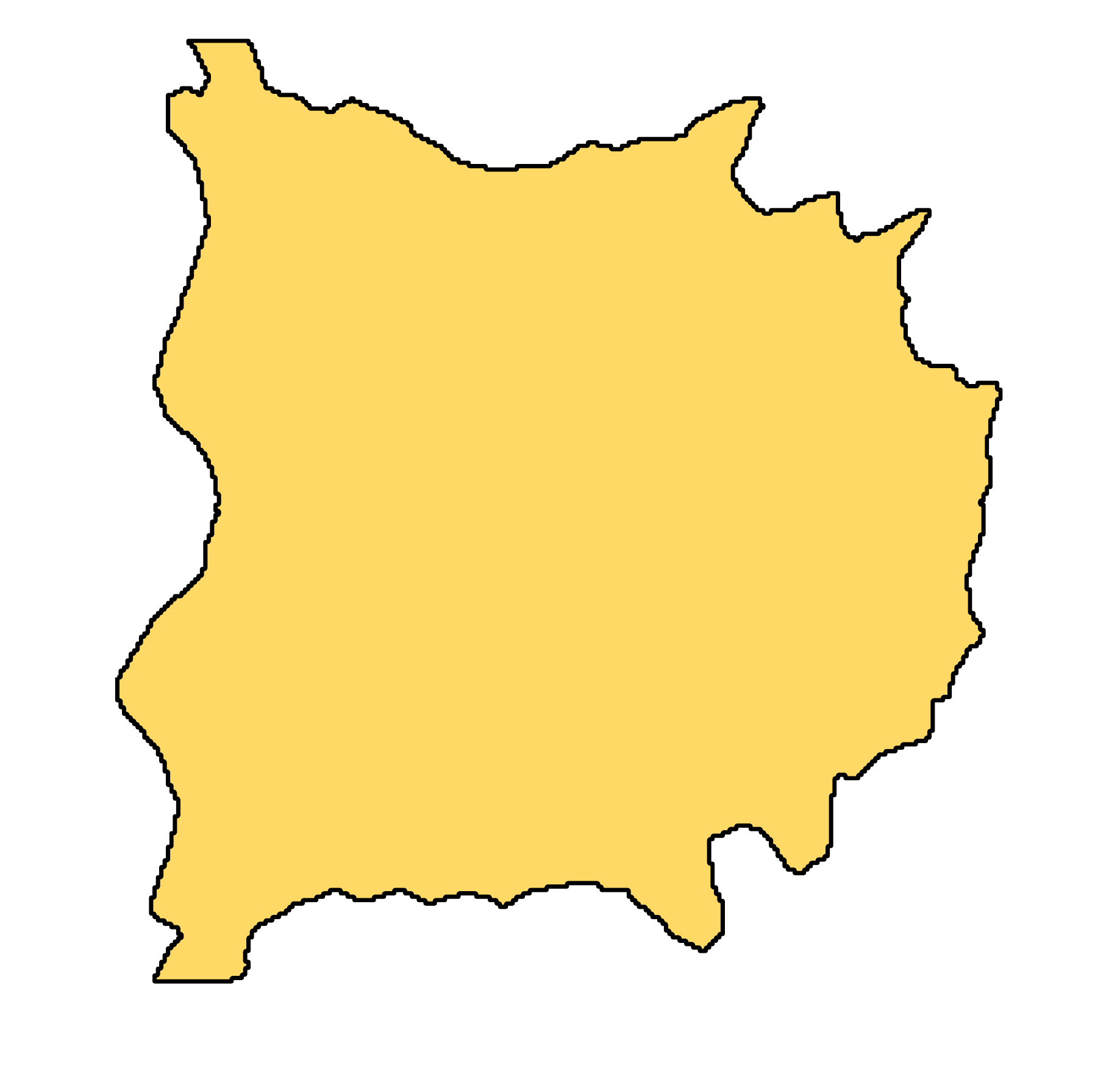}}
	\subcaptionbox*{79917\\253$\times$256}[\twidth]
		{\includegraphics[width=\paperwidth*\real{2.53}*\real{\scale}]{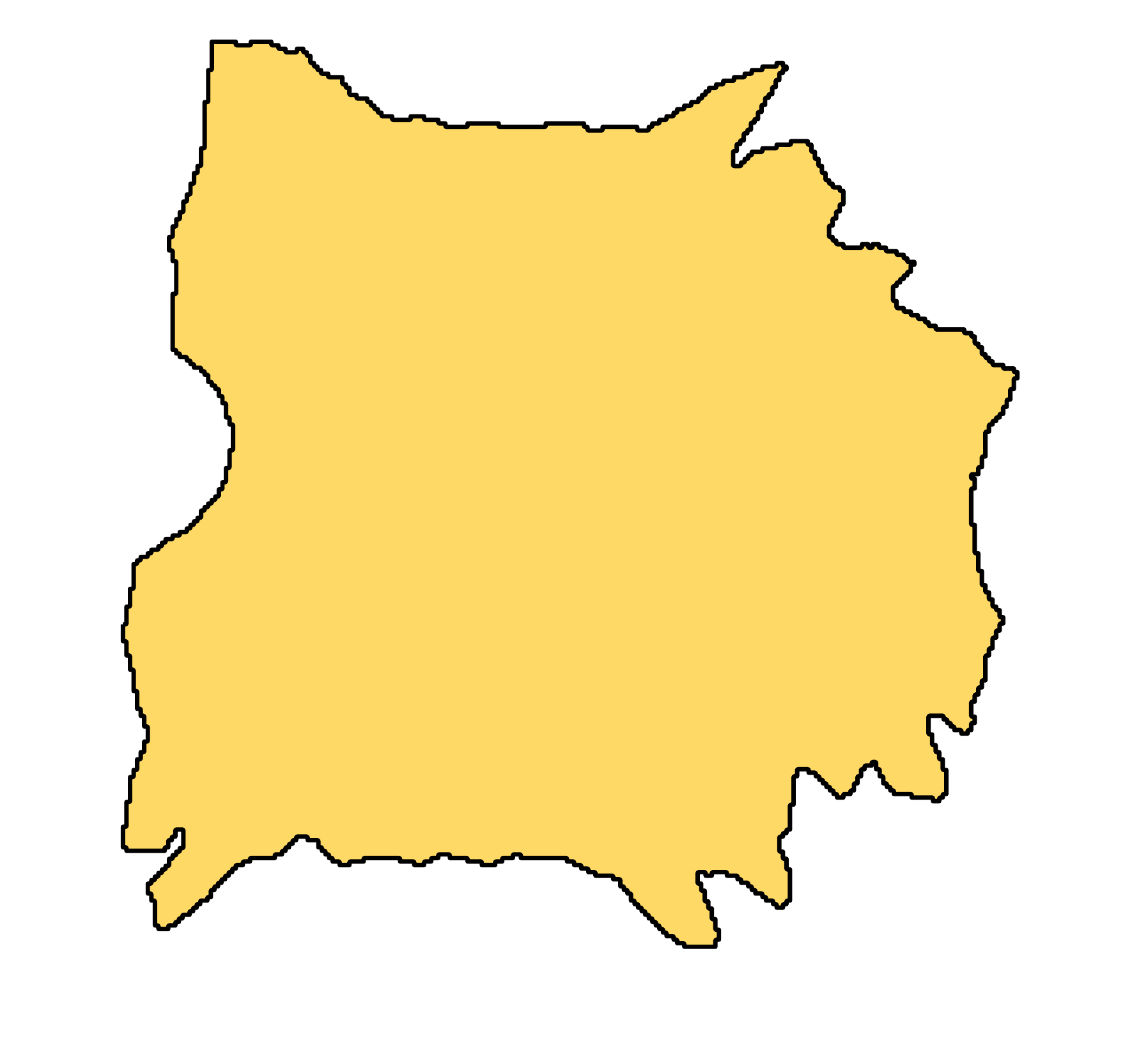}}
	\subcaptionbox*{79918\\269$\times$276}[\twidth]
		{\includegraphics[width=\paperwidth*\real{2.69}*\real{\scale}]{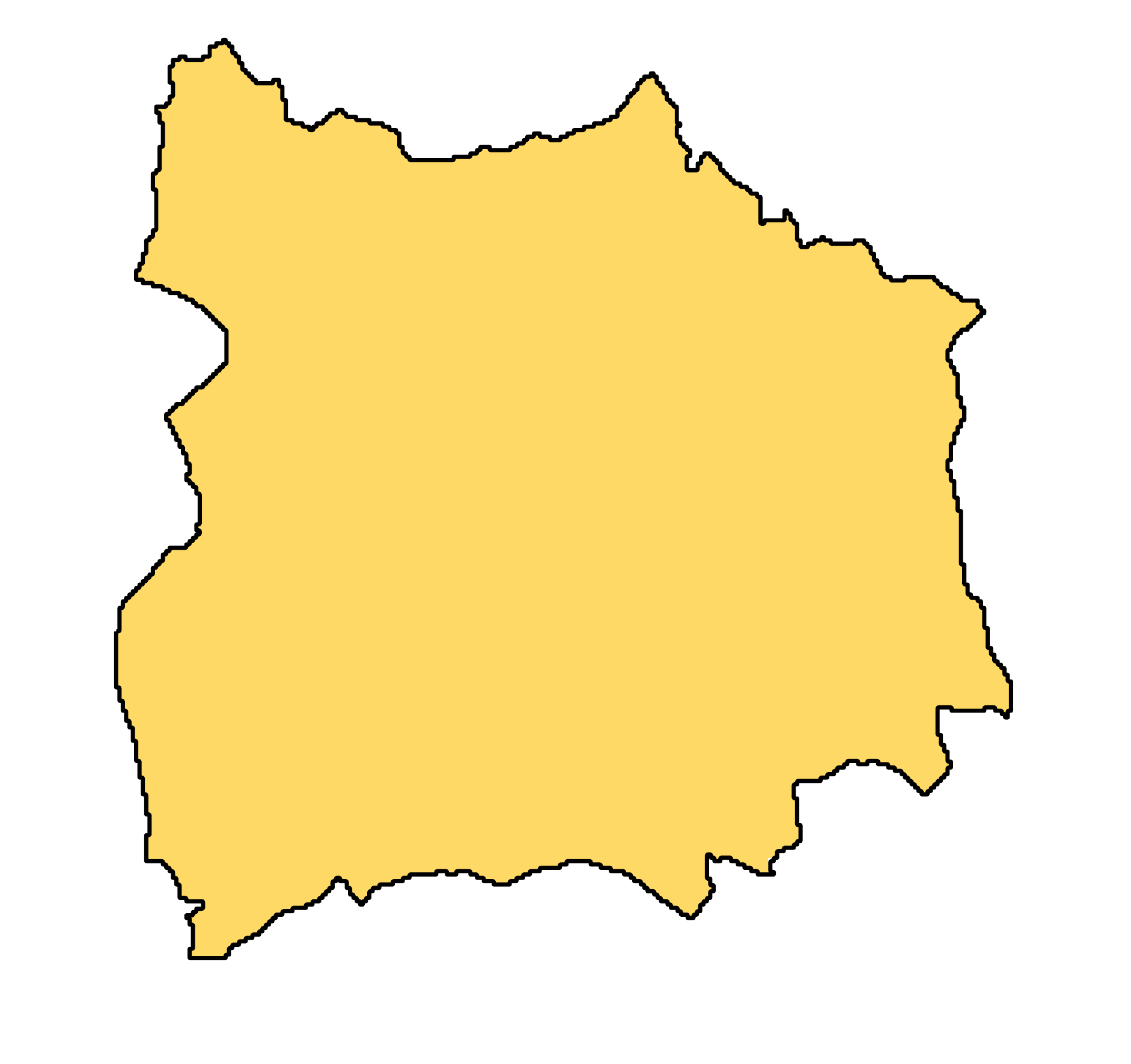}}
	\subcaptionbox*{79919\\240$\times$263}[\twidth]
		{\includegraphics[width=\paperwidth*\real{2.40}*\real{\scale}]{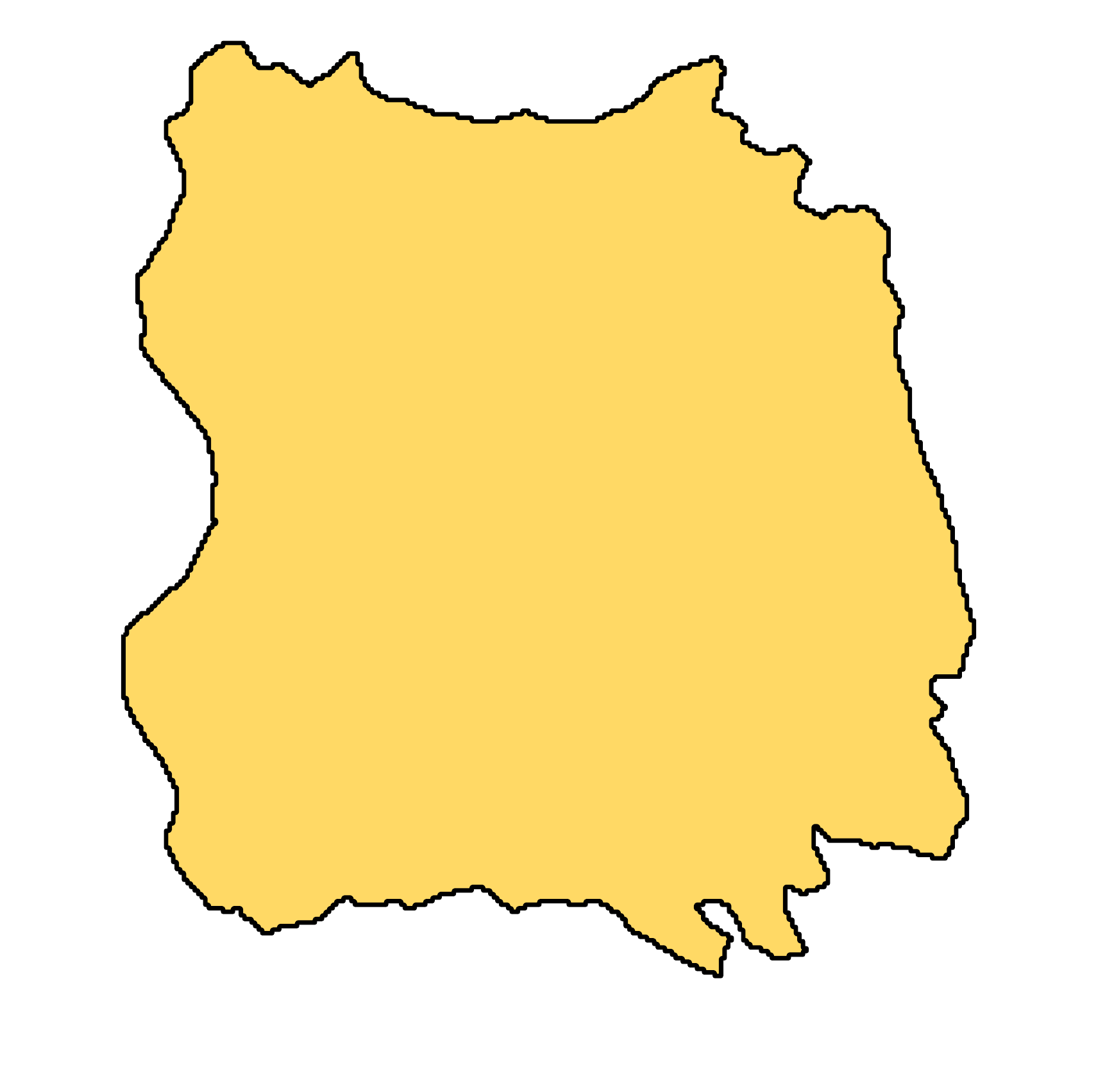}}
	 \caption{Leather master surfaces of industrial nesting problems}\label{fig:benchmark-5}
	\label{fig:benchmark-hide}
\end{figure}

\newcommand{\rulesep}{\unskip\ \vrule\ }

\begin{figure}
  \centering
  \begin{subfigure}[b]{0.12\textwidth}
    \centering
    \includegraphics[width=0.5\textwidth]{typical4}
    \caption*{Image}
  \end{subfigure}
  \rulesep
  \begin{subfigure}[b]{0.12\textwidth}
    \centering
    \includegraphics[width=0.5\textwidth]{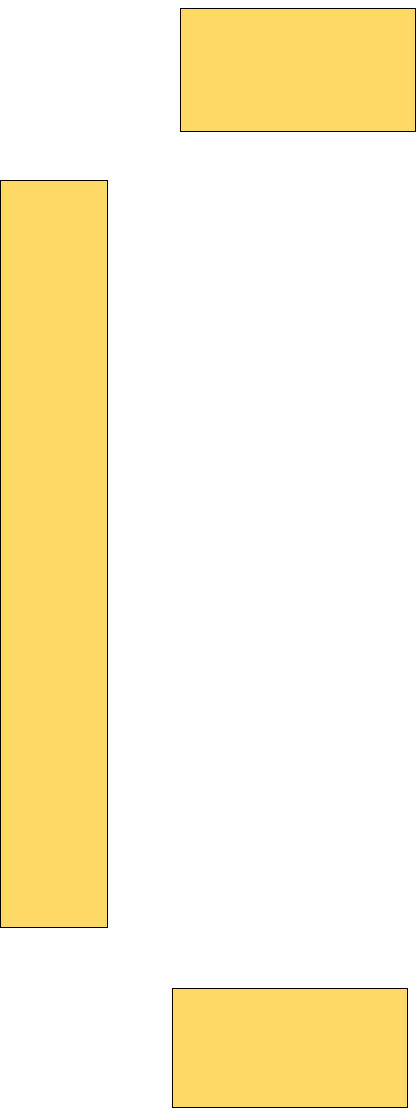}
    \caption*{BP}
  \end{subfigure}
  \begin{subfigure}[b]{0.12\textwidth}
    \centering
    \includegraphics[width=0.5\textwidth]{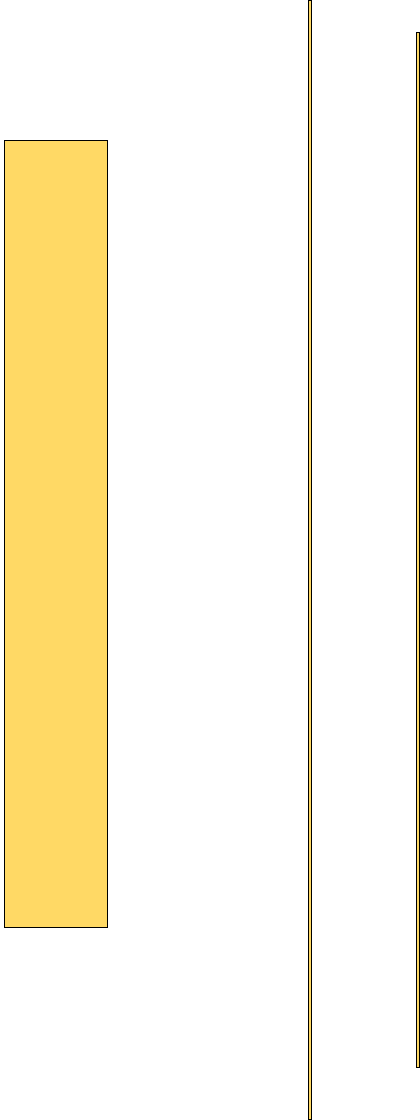}
    \caption*{SF}
  \end{subfigure}
  \begin{subfigure}[b]{0.12\textwidth}
    \centering
    \includegraphics[width=0.5\textwidth]{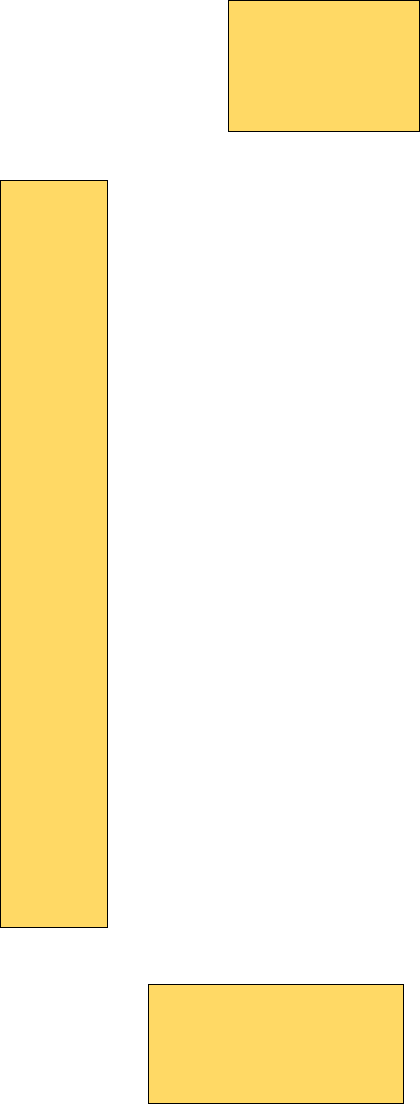}
    \caption*{FAST}
  \end{subfigure}
  \begin{subfigure}[b]{0.12\textwidth}
    \centering
    \includegraphics[width=0.5\textwidth]{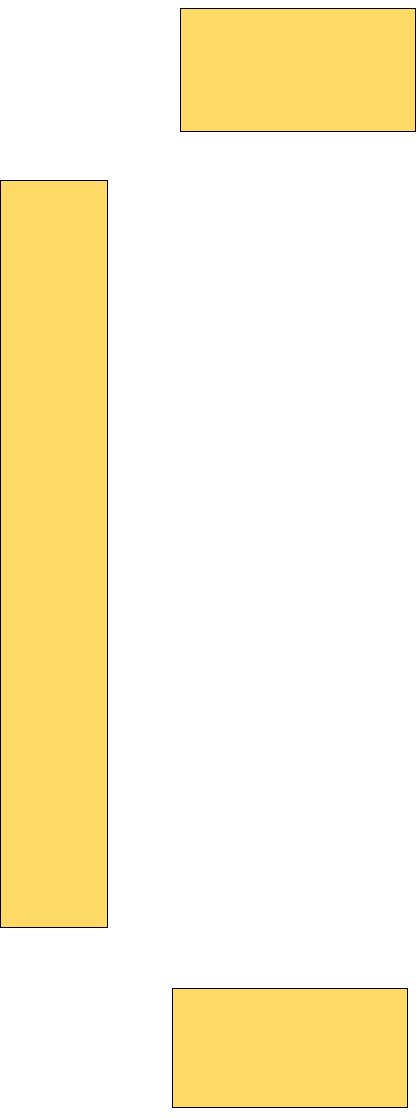}
    \caption*{CSA}
  \end{subfigure}    
  \caption{Sample nonconvex image, typical4, where SF performs poorly for $K=3$}
  \label{fig:example-1}
\end{figure}

\definecolor{lineblue}{RGB}{75,172,198}

\begin{figure}
  \centering
  \begin{subfigure}[b]{0.12\textwidth}
    \centering
    \includegraphics[width=0.5\textwidth]{artificial1}
    \caption*{Input}
  \end{subfigure}
  \rulesep
  \begin{subfigure}[b]{0.12\textwidth}
    \centering
    \includegraphics[width=0.5\textwidth]{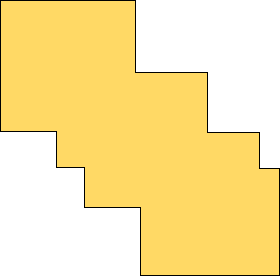}
    \caption*{BP}
  \end{subfigure}
  \begin{subfigure}[b]{0.12\textwidth}
    \centering
    \includegraphics[width=0.5\textwidth]{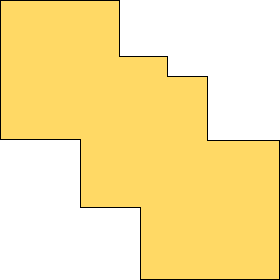}
    \caption*{SF}
  \end{subfigure}
  \begin{subfigure}[b]{0.12\textwidth}
    \centering
    \includegraphics[width=0.5\textwidth]{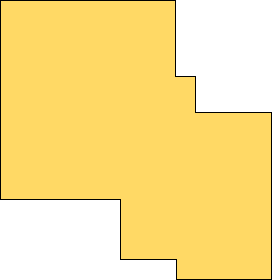}
    \caption*{FAST}
  \end{subfigure}
  \begin{subfigure}[b]{0.12\textwidth}
    \centering
    \includegraphics[width=0.5\textwidth]{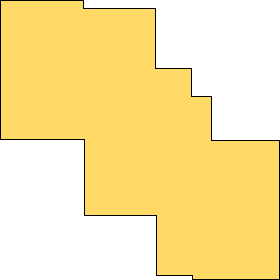}
    \caption*{CSA}
  \end{subfigure}    
  \caption{Sample image, artificial1, where FAST performs poorly for $K=5$}
  \label{fig:example-2}
\end{figure}

\begin{figure}
  \centering
  \begin{subfigure}[b]{0.15\textwidth}
    \centering
    \includegraphics[width=0.5\textwidth]{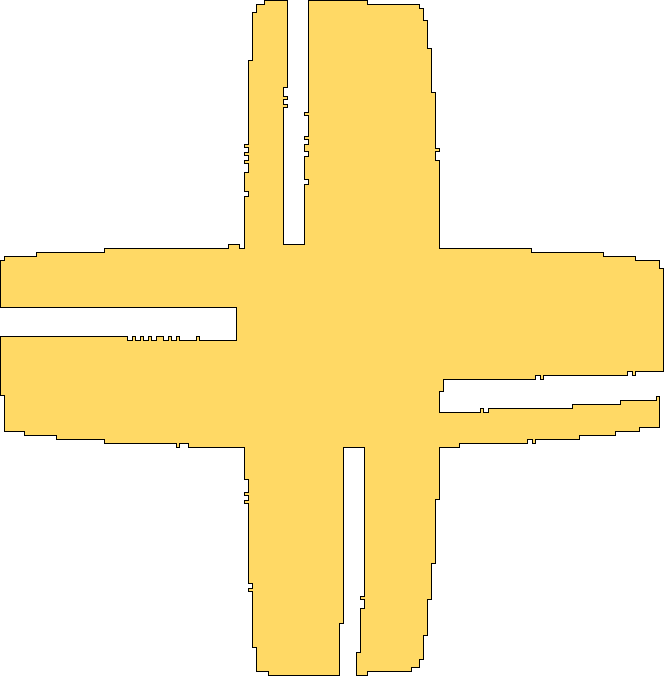}
    \caption*{Input}
  \end{subfigure}
  \rulesep
  \begin{subfigure}[b]{0.15\textwidth}
    \centering
    \includegraphics[width=0.5\textwidth]{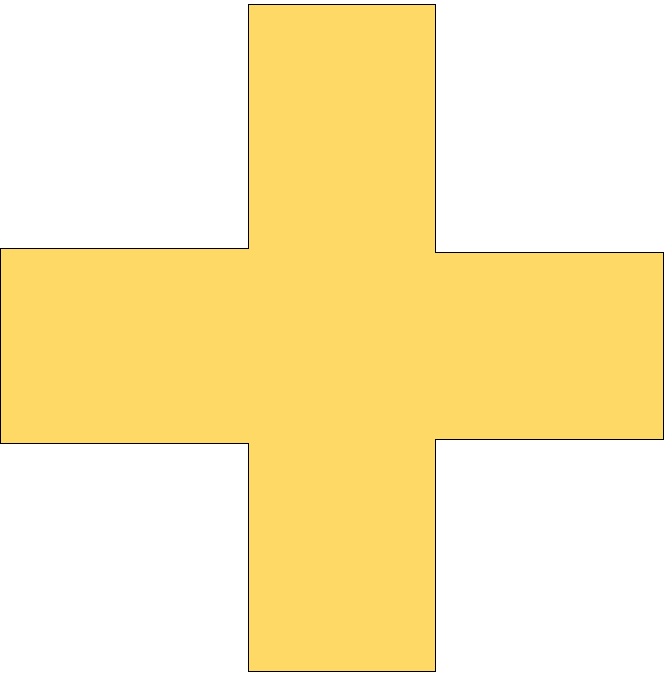}
    \caption*{BP}
  \end{subfigure}
  \begin{subfigure}[b]{0.15\textwidth}
    \centering
    \includegraphics[width=0.5\textwidth]{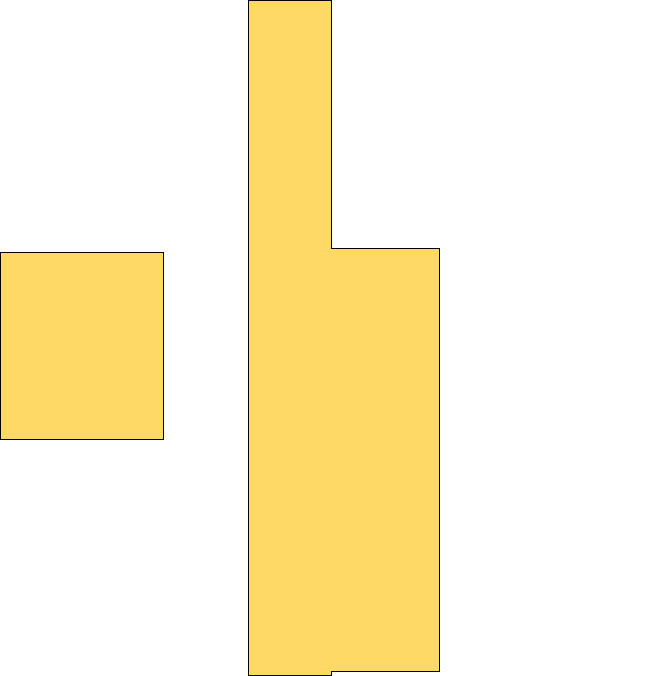}
    \caption*{SF}
  \end{subfigure}
  \begin{subfigure}[b]{0.15\textwidth}
    \centering
    \includegraphics[width=0.5\textwidth]{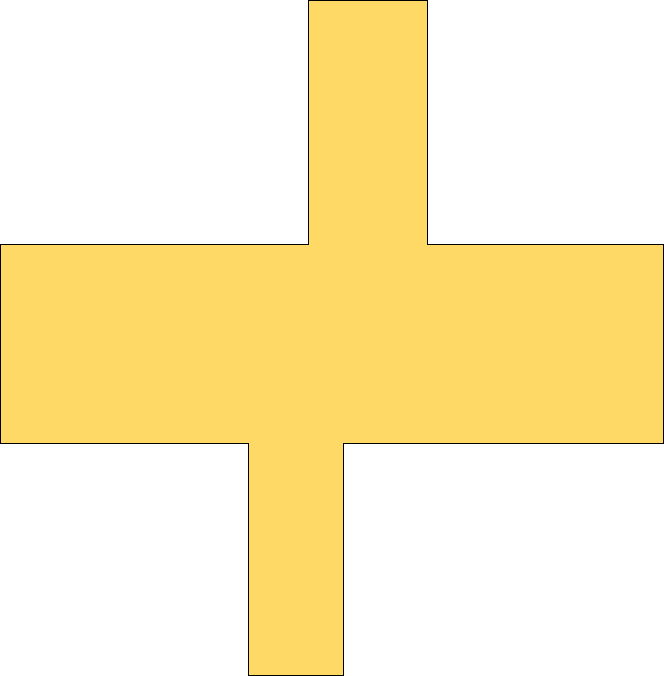}
    \caption*{FAST}
  \end{subfigure}
  \begin{subfigure}[b]{0.15\textwidth}
    \centering
    \includegraphics[width=0.5\textwidth]{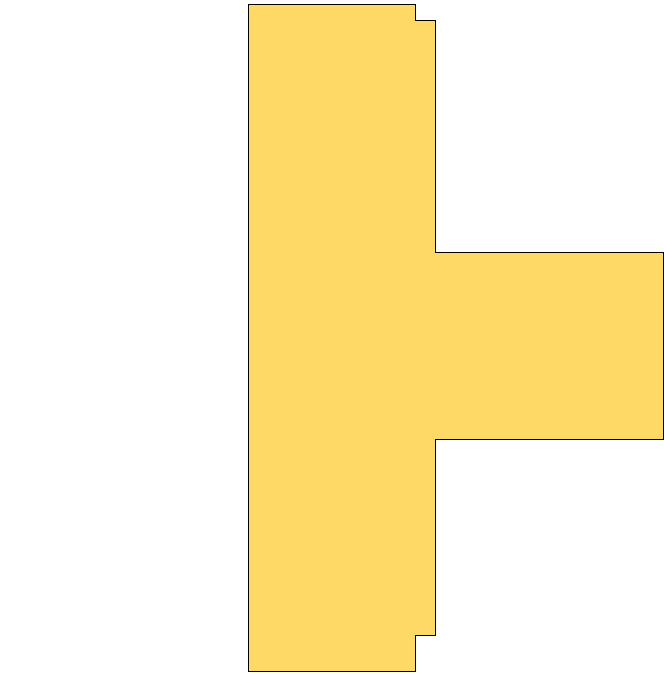}
    \caption*{CSA}
  \end{subfigure}    
  \caption{Sample image, device5-9, where BP finds the optimum for $K=3$}
  \label{fig:example-3}
\end{figure}

\subsection{Implementation details}

In the actual implementation of the heuristics, rectangles are presented as quadruplets, $(x,y,w,h)$ where $x$ and $y$ represent the top left coordinates, and $w$ and $h$ are the width and height of the rectangle. However, in BP, rectangles are presented as matrices (or vectors) with entries set according to expression~(\ref{equ:def-rectangle}).

We have noticed  that starting column configurations do not affect the number of generated columns and convergence behavior significantly in BP. Our column generation process shows rapid progress in the early iterations, especially with the dual smoothing scheme we have implemented. 

We have considered various operations to speed up the optimization procedure. Whenever RLPM is modified via column generation or branching, instead of solving it from scratch we start the optimization from the basis obtained in the previous step (i.e. warm start). Multiple columns (up to 10) are generated during pricing in order to speed up the computations. We also keep track of how long each column has been in RLPM. A column that does not enter into the basis for a certain amount of time, is discarded. If a previously discarded column is regenerated, we double the lifespan of that column when we add it back to RLPM. This technique limits the amount of columns present in RLPM and speeds up the optimization dramatically.

In the pricing problem, we have used priority queue data structure, as usually done, to get the next promising rectangle set. However, rectangle sets that contain few rectangles dominate the overall cost of the algorithm. This is because of the overhead associated with the priority queue and the data structure that holds the rectangle set. When the number of rectangles are sufficiently small, we switch to a naive algorithm that iterates through all rectangles in the set. Iterated rectangles are added to another priority queue with given capacity that determines the number of columns to return from the pricing problem. We have experimentally found that the best threshold is 256 for switching to the naive algorithm for our hardware. To make this hybrid algorithm work with the branching constraints, we also check if all the rectangles in the rectangle set satisfy the constraints before switching to the naive algorithm. This can be performed quickly by checking the pixels involved in the constraints against the rectangle set (see Figure~\ref{fig:branching_conditions} for valid configurations). 

Experiments are carried out on a computer with Core i7 3.07 GHz CPU and 8 GB RAM. BP and CSA are implemented using Java; SF and FAST  are implemented in C$^{{++}}$. Three software libraries are used: OpenCV~\citep{Kaehler2013} for operations on images, \cite{Kabeja2018} for rasterization of the images given in DXF format and Gurobi~\citep{gurobi} for BP. \textcolor{red}{ 
 The running time of all tests are limited to 1 hour, except for the runs related to the performance of the branching rules (i.e. Table 1), for which the run time limit is 1.5 hours. }

\subsection{Observations}

The computational results can be grouped in two major categories. The first one is for inspecting the effect of branching rules on the efficiency of BP. The second one is for analyzing the performance of the solution methods with respect to their efficiency and accuracy. 

\subsubsection{Branching rules}

We should point that in $652$ out of the $685$ runs  of BP on the first four data sets, the initial LP relaxation gives the integer optimal solution. \textcolor{red}{ So we only have 33 cases for comparing the branching rules. } The results are reported in Table~\ref{tab:branching-comparison}, \textcolor{red}{where the} first two columns list data sets and the corresponding blanket size. In the next two columns, we report the number of nodes visited for each rule. RULE 1 and RULE 2 denote branching explicitly or implicitly in the master, respectively. The numbers in the last two columns are simply the percent relative deviations for the LP relaxation lower bound from the best feasible solution computed in 1.5 hours CPU \textcolor{red}{time} limit using BP (i.e. this limit is 1 hour for the rest of experiments). As can be observed, they are quite small, \textcolor{red}{indicating} the tightness of the LP relaxation. Besides, ``none" means $0.00$ \% deviation. This occurs when the LP relaxation have also an integer alternative optimal solution. It is detected by BP later without changing the optimum objective value. \textcolor{red}{Observe that there is a single column for reporting the deviations in the table, because in all instances, both methods find a solution with the same objective value even if they reach the time limit and are forced to stop.  Besides, when the methods finish running within the time limit, they both find an optimal solution which makes the reported objective values equal.}

We have measured the number of nodes visited in branch-and-price using both rules. RULE 2  visited fewer nodes for 15 instances and more nodes for 11 instances. We can claim that, although there is no clear winner, RULE 2 is more efficient than RULE 1 based on the averages reported on the last row of the table. This is expected since RULE 2 is more likely to produce a balanced search tree as mentioned earlier.

\begin{table}[!htbp]
   \caption{Efficiency of the branching rules and the improvement on the initial LP relaxation}
   \begin {center}
   \begin{tabular}{rr r rr r rr r rr}
   \toprule
	& & \phantom{a} & \multicolumn{2}{c}{num. visited nodes} &  \phantom{a} &  \multicolumn{2}{c}{CPU time (secs)} &  & Relative  \\
			\cmidrule(r){4-5}	\cmidrule(r){7-8}
	name & $K$ &   & RULE1 & RULE2 &   & RULE1 & RULE2 & &Dev. (\%) \\
	\midrule
dog-12 & 3 &   & 73 & 115 & & 16 & 12 & & 1.54\\
dog-13 & 3 &   & 79 & 55 & & 18 & 8 & & 1.54\\
key-16 & 3 &   & 3 & 9 & & 210 & 310 & & none\\
key-18 & 3 &   & 3 & 3 & & 4125 & 4398 & & none\\
toy13 & 3 &   & 3 & 351 & & 1 & 17 & & none\\
toy4 & 3 &   & 89 & 43 & & 36 & 20 & & 0.75\\
avatar4 & 5 &   & 7 & 9 & & 1 & 2 & & none\\
device5-13 & 5 &   & 27 & 27 & & 853 & 1416 & & 0.28\\
device5-8 & 5 &   & 18 & 41 & & 5411 & 5402 & & 0.27\\
key-17 & 5 &   & 51 & 47 & & 2736 & 3666 & & 0.09\\
toy12 & 5 &   & 5 & 5 & & 2 & 3 & & none\\
toy14 & 5 &   & 25 & 29 & & 31 & 60 & & 0.25\\
bat-2 & 10 &   & 48 & 37 & & 5405 & 5402 & & 0.05\\
bat-6 & 10 &   & 179 & 61 & & 1920 & 1524 & & 0.13\\
device5-18 & 10 &   & 7 & 3 & & 1289 & 952 & & 0.02\\
dog-12 & 10 &   & 37 & 5 & & 20 & 6 & & 0.24\\
dog-13 & 10 &   & 7 & 17 & & 8 & 12 & & 0.24\\
toy11 & 10 &   & 3 & 5 & & 7 & 11 & & 0.15\\
typical4 & 10 &   & 6 & 6 & & 5407 & 5421 & & none\\
dog-6 & 15 &   & 5 & 7 & & 46 & 62 & & none\\
toy1 & 15 &   & 3 & 23 & & 4 & 19 & & none\\
toy10 & 15 &   & 33 & 13 & & 257 & 175 & & 0.09\\
toy11 & 15 &   & 7 & 13 & & 14 & 13 & & 0.14\\
toy14 & 15 &   & 9 & 3 & & 95 & 68 & & none\\
avatar4 & 20 &   & 9 & 3 & & 2 & 1 & & none\\
dog-12 & 20 &   & 17 & 17 & & 17 & 20 & & 0.22\\
dog-13 & 20 &   & 15 & 9 & & 12 & 14 & & 0.22\\
dog-16 & 20 &   & 7 & 3 & & 7 & 6 & & none\\
dog-17 & 20 &   & 61 & 3 & & 6 & 1 & & 0.44\\
dog-7 & 20 &   & 3 & 3 & & 86 & 90 & & none\\
toy2 & 20 &   & 958 & 5 & & 5428 & 99 & & none\\
toy4 & 20 &   & 15 & 11 & & 97 & 126 & & none\\
toy9 & 20 &   & 3 & 3 & & 61 & 66 & & none\\
\hline
\multicolumn{2}{r}{Mean} && 55.00 & 29.82 & & 1019.04 & 890.98 & &0.20 \\
  \bottomrule
   \end{tabular}
   \end{center} \label{tab:branching-comparison}
\end{table}

\subsubsection{Solution time}
According to the computational results, we can say that the heuristics are very efficient. Because the order of magnitude of running times are different, we do not provide a detailed running time comparison between BP and heuristics. FAST and SF heuristics are the fastest; and  they run in milliseconds for a typical input.  CSA takes about $2.5$ seconds to converge on the average. We force BP to stop running within 1 CPU hour.

As for the BP, column generation shows rapid progress in the early iterations, especially with the dual smoothing scheme. We illustrate this typical behavior of the optimality gap in Figure~\ref{fig:typical_behavior} by exposing the progress in the best lower bound and objective value our algorithm calculates on artificial3 data set for $K=5$. This allows stopping the process earlier using lower and upper bounds on the optimal value.

\begin{figure}[!h]
  \centering
      \begin{subfigure}[b]{0.45\textwidth}
       \includegraphics[width=\textwidth] {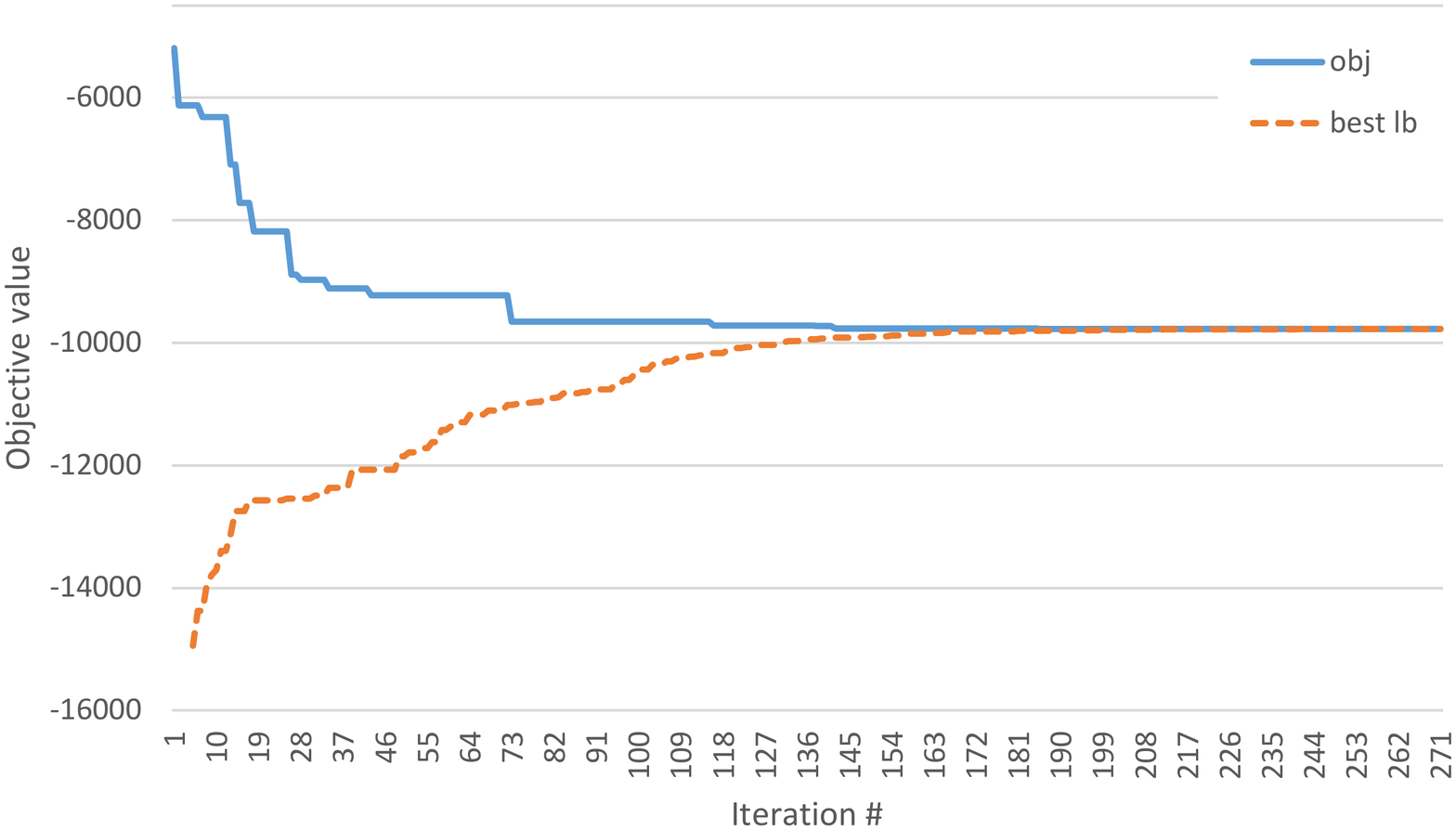}
      \caption{The objective and best lower bound values} \label{fig:typical_behavior}
    \end{subfigure}
    \begin{subfigure}[b]{0.4\textwidth}
       \includegraphics[width=\textwidth] {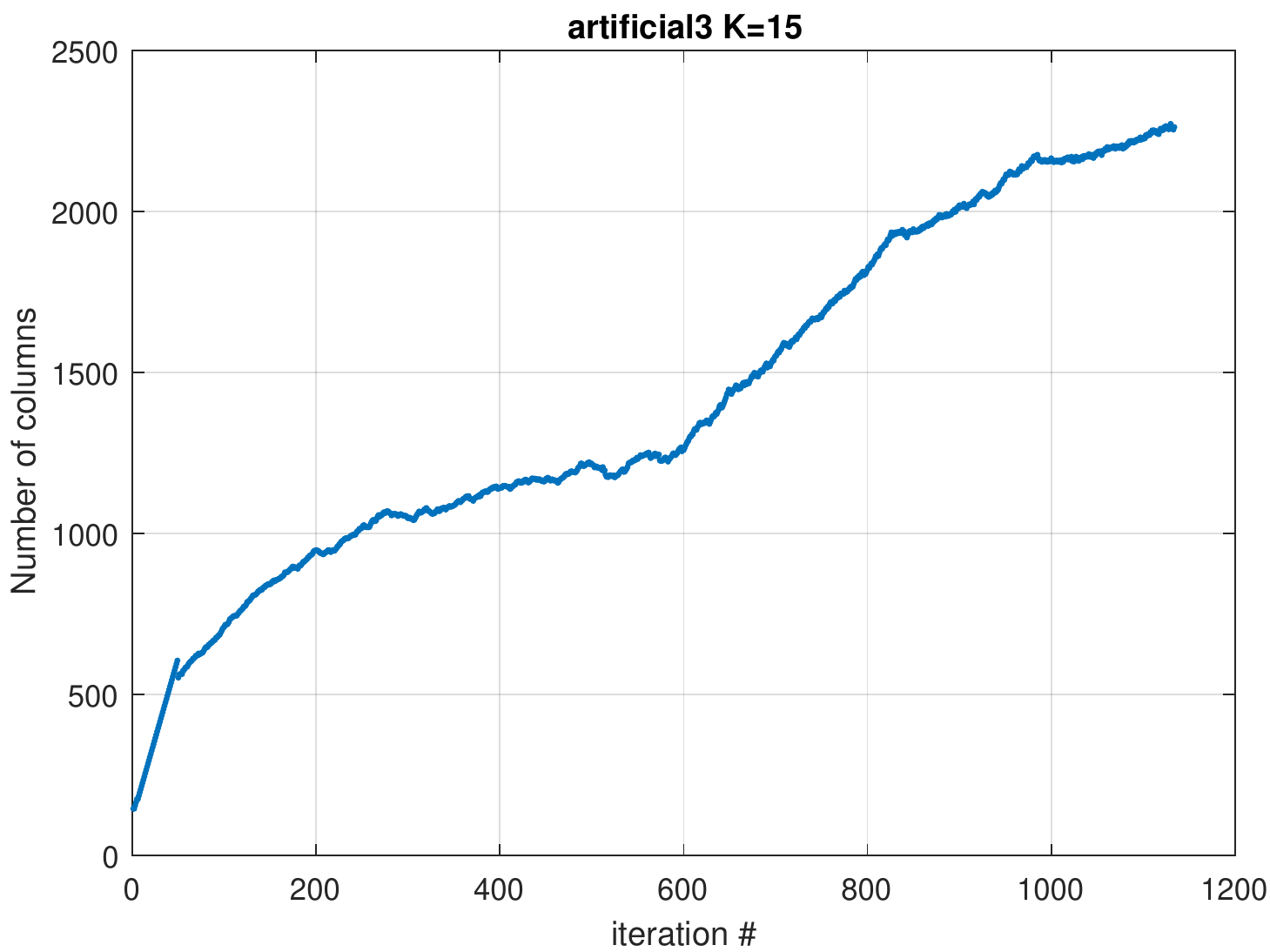}
      \caption{The number of columns in the master problem}\label{fig:column_count}
    \end{subfigure}\\
    \begin{subfigure}[b]{0.4\textwidth}
       \includegraphics[width=\textwidth] {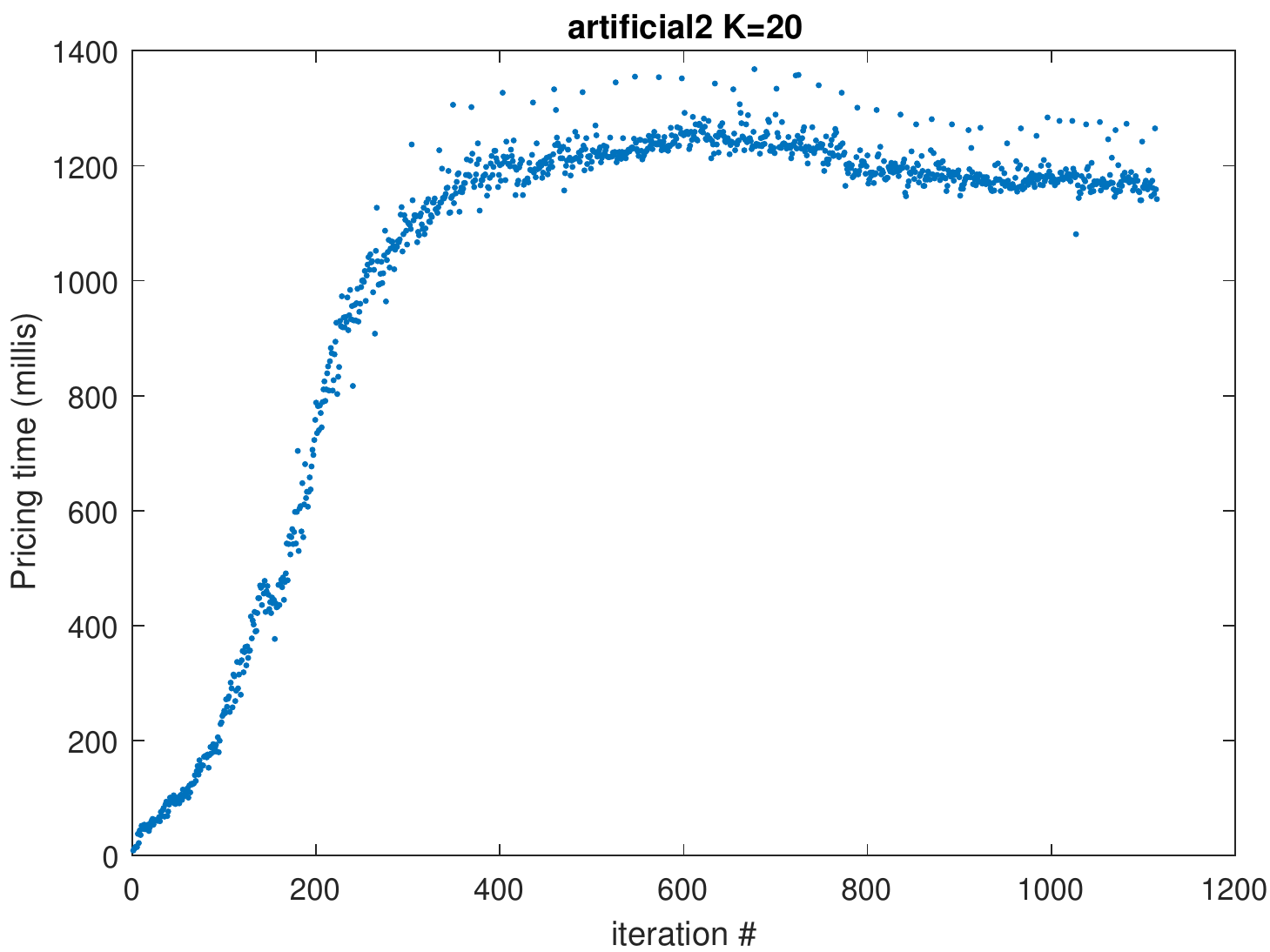}
       \caption{The solution time of the pricing subproblem}\label{fig:pricing_time}
    \end{subfigure} 
    \begin{subfigure}[b]{0.4\textwidth}
       \includegraphics[width=\textwidth] {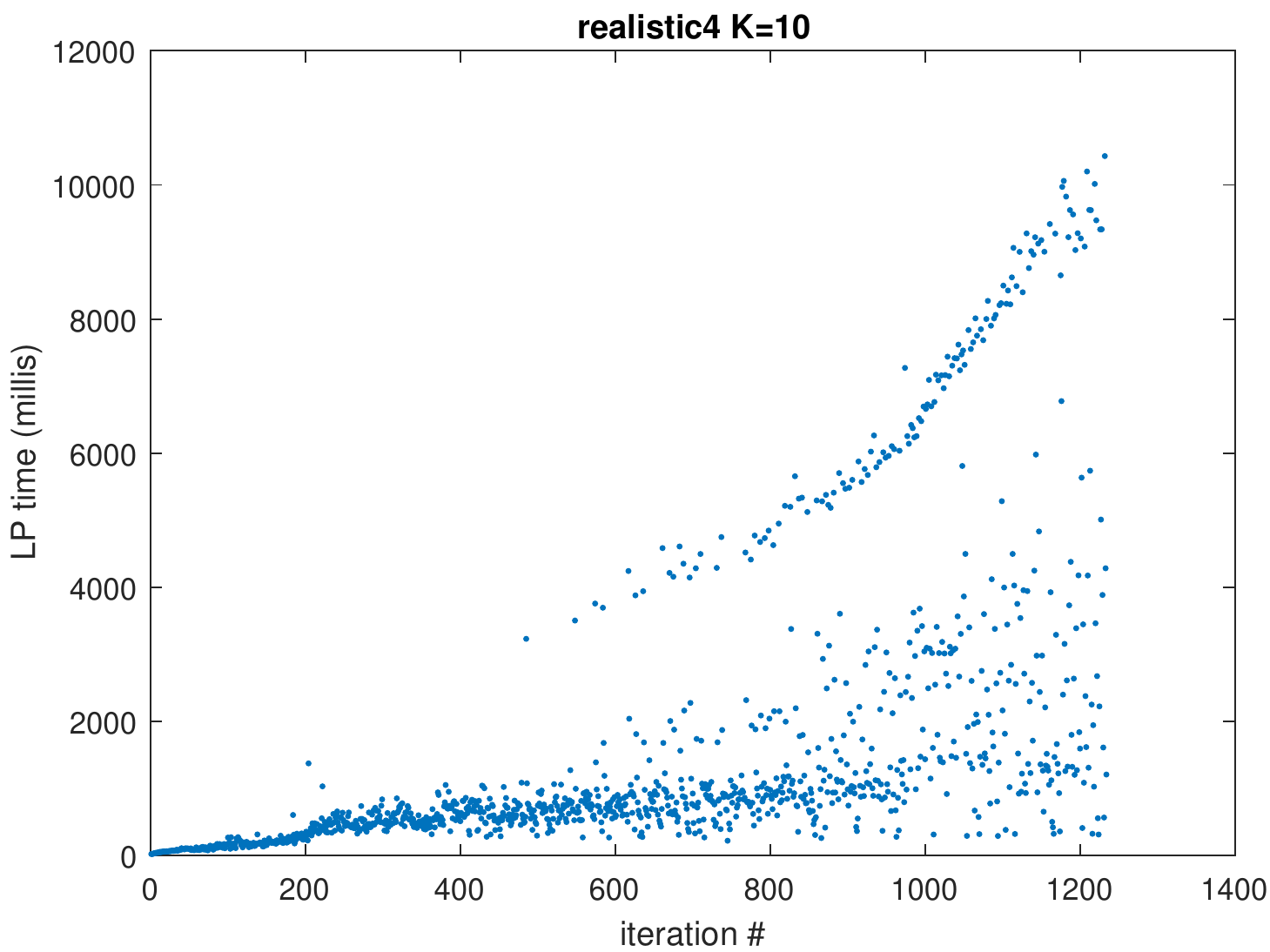}
      \caption{The LP time}\label{fig:LP_time}
    \end{subfigure}
 \caption{Typical behavior patterns}\label{fig:typical_change}
\end{figure}

We provide three more graphs related to the performance of the BP algorithm in Figure \ref{fig:column_count} - Figure \ref{fig:LP_time}. They illustrate typical behavior patterns which we observed \textcolor{red}{for almost all test instances.} The first one uses data collected on artificial3 data set for $K=15$ and plots the number of columns in the master problem during the iterations. There is a sharp increase at the very beginning, which slows down later on. \textcolor{red}{This is because column discarding is not active at the beginning and starts to operate after a large enough number of columns is generated.} We should point out that our column discarding scheme keeps the number of columns relatively lower.

The second one is obtained on artificial2 data set for $K=20$. It illustrates the typical behavior of the time devoted to the solution of the pricing subproblem throughout the iterations. It follows a regular behavior. The curve behaves parabolically at the beginning with a sharp increase. Then, it reaches a peak and settles down asymptotically. \textcolor{red}{ In short, pricing subproblems are easy in early iterations, and become harder as iterations progress.} The hardness is capped at some point.

The third one is the plot of the CPU time spent to solve the RLPMs throughout the iterations on realistic4 data set for $K=10$. First of all observe that, with each iteration the solution of the LP becomes harder. \textcolor{red}{After a while, a second trend emerges, and occasional  solutions take significantly longer times.}

Finally, we have prepared Table \ref{tab:results-stats} for the first test problem group for $K = 3,5,10,15,20$, in order to give an idea on the number of generated columns, time spent for the solution of the RLPMs ($t_\text{RLPM}$) and PSPs ($t_\text{PSP}$) in milliseconds.  The first column includes the instances, as usual. The first column of each group includes the median values for the number of columns added \textcolor{red}{at each iteration} ($\Delta_\text{col}$). The second and third columns are the median values of $t_\text{RLPM}$ and $t_\text{PSP}$. They all increase with $K$, as expected. If time spent is zero, it means the median time spent is less than a millisecond.

\begin{sidewaystable}
\caption{The efficiency of BP} 
\footnotesize
\begin{center}
{\begin{tabular}{r r rrr r rrr r rrr r rrr r rrr}
\toprule
			& \phantom{a} & \multicolumn{3}{c}{$K=3$} & \phantom{a} & \multicolumn{3}{c}{$K=5$}  & \phantom{a}  & \multicolumn{3}{c}{$K=10$}  & \phantom{a} & \multicolumn{3}{c}{$K=15$}  & \phantom{a} & \multicolumn{3}{c}{$K=20$} \\
			\cmidrule(r){3-5} 	\cmidrule(r){7-9} 	\cmidrule(r){11-13} 	\cmidrule(r){15-17} 	\cmidrule(r){19-21}
Image & & $\Delta_\text{col}$ & $t_\text{LP}$  & $t_\text{PSP}$ & & $\Delta_\text{col}$ & $t_\text{LP}$ & $t_\text{PSP}$ & & $\Delta_\text{col}$ & $t_\text{LP}$ & $t_\text{PSP}$ & & $\Delta_\text{col}$ & $t_\text{LP}$ & $t_\text{PSP}$ & & $\Delta_\text{col}$ & $t_\text{LP}$ & $t_\text{PSP}$ \\
\midrule
artificial1 & & 5 & 83 & 42 & & 4 & 228 & 90 & & 5 & 1263 & 95 & & 5 & 2371 & 92 & & 6 & 4590 & 91\\
artificial2 & & 3 & 151 & 25 & & 2 & 210 & 153 & & 3 & 426 & 410 & & 3 & 378 & 717 & & 3 & 842 & 1180\\
artificial3 & & 2 & 402 & 78 & & 2 & 332 & 362 & & 2 & 320 & 743 & & 2 & 442 & 934 & & 3 & 1500 & 1380\\
artificial4 & & 2 & 159 & 83 & & 2 & 146 & 190 & & 3 & 454 & 647 & & 7 & 1431 & 677 & & 8 & 1750 & 671\\
artificial5 & & 5 & 119 & 60 & & 6 & 1200 & 97 & & 7 & 4114 & 102 & & 9 & 5117 & 103 & & 9 & 6341 & 103\\
avatar1 & & 1 & 0 & 0 & & 2 & 1 & 0 & & 2 & 1 & 0 & & 3 & 1 & 0 & & 1 & 1 & 0\\
avatar2 & & 4 & 3 & 1 & & 7 & 5 & 1 & & 6 & 8 & 1 & & 8 & 11 & 1 & & 9 & 8 & 1\\
avatar3 & & 9 & 2 & 1 & & 3 & 3 & 1 & & 4 & 6 & 1 & & 7 & 10 & 1 & & 9 & 6 & 1\\
avatar4 & & 3 & 2 & 0 & & 1 & 3 & 0 & & 6 & 6 & 0 & & 7 & 6 & 0 & & 8 & 4 & 0\\
realistic1 & & 6 & 88 & 62 & & 5 & 214 & 163 & & 7 & 1376 & 185 & & 7 & 2062 & 193 & & 8 & 2871 & 186\\
realistic2 & & 3 & 123 & 45 & & 2 & 115 & 142 & & 4 & 1009 & 396 & & 5 & 1465 & 426 & & 8 & 2300 & 440\\
realistic3 & & 3 & 215 & 85 & & 8 & 2620 & 499 & & 9 & 2447 & 507 & & 9 & 2349 & 523 & & 9 & 2797 & 521\\
realistic4 & & 10 & 128 & 6 & & 3 & 130 & 185 & & 3 & 743 & 504 & & 3 & 1433 & 538 & & 4 & 1905 & 547\\
realistic5 & & 2 & 121 & 45 & & 2 & 165 & 235 & & 4 & 1070 & 488 & & 6 & 1345 & 508 & & 7 & 1828 & 506\\
toy1 & & 10 & 14 & 2 & & 10 & 15 & 2 & & 9 & 16 & 1 & & 1 & 18 & 0 & & 0 & 15 & 0\\
toy2 & & 3 & 5 & 5 & & 10 & 160 & 8 & & 9 & 230 & 5 & & 10 & 112 & 9 & & 10 & 118 & 9\\
toy3 & & 6 & 12 & 11 & & 8 & 34 & 12 & & 7 & 140 & 13 & & 6 & 301 & 13 & & 7 & 362 & 13\\
toy4 & & 0 & 13 & 7 & & 6 & 31 & 9 & & 6 & 63 & 10 & & 7 & 100 & 10 & & 5 & 103 & 10\\
toy5 & & 6 & 2 & 3 & & 10 & 2 & 3 & & 5 & 3 & 4 & & 7 & 5 & 4 & & 7 & 6 & 4\\
toy6 & & 3 & 304 & 145 & & 5 & 322 & 292 & & 7 & 2197 & 323 & & 8 & 1856 & 329 & & 8 & 2358 & 335\\
toy7 & & 6 & 9 & 10 & & 7 & 22 & 12 & & 7 & 46 & 12 & & 7 & 100 & 12 & & 7 & 112 & 12\\
toy8 & & 2 & 2 & 4 & & 5 & 3 & 6 & & 7 & 7 & 6 & & 7 & 12 & 6 & & 7 & 14 & 6\\
toy9 & & 0 & 7 & 2 & & 2 & 5 & 6 & & 8 & 37 & 8 & & 7 & 64 & 8 & & 6 & 85 & 8\\
toy10 & & 6 & 18 & 16 & & 6 & 37 & 19 & & 6 & 74 & 20 & & 4 & 81 & 20 & & 6 & 125 & 20\\
toy11 & & 3 & 2 & 5 & & 6 & 3 & 10 & & 4 & 8 & 9 & & 5 & 14 & 9 & & 7 & 13 & 9\\
toy12 & & 0 & 1 & 1 & & 4 & 3 & 4 & & 6 & 7 & 5 & & 7 & 8 & 5 & & 7 & 11 & 5\\
toy13 & & 0 & 4 & 2 & & 2 & 3 & 6 & & 7 & 6 & 9 & & 6 & 10 & 8 & & 7 & 12 & 8\\
toy14 & & 9 & 17 & 18 & & 0 & 21 & 7 & & 7 & 40 & 21 & & 6 & 93 & 21 & & 6 & 77 & 22\\
typical1 & & 3 & 66 & 50 & & 4 & 152 & 111 & & 5 & 1795 & 125 & & 5 & 3521 & 125 & & 7 & 3086 & 126\\
typical2 & & 2 & 148 & 171 & & 2 & 191 & 269 & & 3 & 1569 & 647 & & 4 & 2110 & 734 & & 7 & 3255 & 752\\
typical3 & & 4 & 3128 & 307 & & 8 & 4871 & 313 & & 9 & 6079 & 315 & & 9 & 6940 & 317 & & 10 & 6825 & 318\\
typical4 & & 2 & 165 & 10 & & 2 & 150 & 277 & & 3 & 680 & 587 & & 3 & 1157 & 656 & & 3 & 2061 & 711\\
typical5 & & 2 & 84 & 49 & & 3 & 165 & 131 & & 4 & 474 & 259 & & 6 & 1615 & 295 & & 8 & 996 & 296\\
typical6 & & 5 & 99 & 53 & & 6 & 685 & 74 & & 6 & 4639 & 81 & & 8 & 5887 & 82 & & 8 & 4400 & 82\\
typical7 & & 3 & 91 & 73 & & 5 & 226 & 98 & & 5 & 842 & 118 & & 4 & 1283 & 117 & & 4 & 3079 & 114\\
typical8 & & 3 & 72 & 29 & & 3 & 131 & 116 & & 4 & 364 & 276 & & 5 & 1562 & 298 & & 7 & 704 & 310\\
typical10 & & 4 & 153 & 201 & & 4 & 761 & 329 & & 7 & 1114 & 365 & & 7 & 1465 & 376 & & 7 & 2184 & 377\\
\bottomrule
\end{tabular}}
\label{tab:results-stats}
\end{center}
\end{sidewaystable}

\subsubsection{Solution quality}

As for the accuracy of the methods, the results are given in Table~\ref{tab:results-simple-1} and Table~\ref{tab:results-leather-2}. The first column of the tables lists the instances. The accuracies of each algorithm are listed for different $K$ values. The first column of each group consists of the objective values BP calculates (i.e. the value of the best feasible solution) in \textcolor{red}{1-hour CPU time limit}. The next three are the percent relative deviations of the values computed using SF, FAST and CSA, respectively. They are calculated according to the formula 
\begin{equation*}
\text{PD} = 100 \times \frac {z_\text{H} - z_\text{BP}} {z_\text{BP}},
\end{equation*}
where $z_\text{BP}$ and $z_\text{H}$ for $\text{H}\in \{\text{SF}, \text{FAST}, \text{CSA}\}$ are the objective values obtained by BP and the heuristics SF, FAST and CSA, respectively. $z_\text{H}$ values are given within parentheses in case the deviation is undefined, i.e. $z_\text{BP}=0$. Column arithmetic averages and standard deviations of the relative percent deviations are also provided on the last row for a rough comparison of the methods. BP guarantees the optimal solution. However, it can take remarkably longer.

Table~\ref{tab:results-mpeg7} is structured similarly but reports results on MPEG 7 images. Because each category contains 20 instances, only the arithmetic averages and standard deviations of the relative percent deviations are reported. However, we should point out that for the configuration $K=20$ and device5, BP performed extremely well for a particular test instance, producing huge errors for other methods. Therefore, we have treated that result as an outlier and removed it from the sample. In other words, reported arithmetic averages and standard deviations are calculated using 19 values out of 20.

\begin{sidewaystable}
\renewcommand{\arraystretch}{1.2}
\caption{Accuracy of the algorithms: $K=3,5,10$\\ (*) Proven optimal solution in $1$-hour running time}
\footnotesize 
\begin{center}
{\begin{tabular}{r r rrrr r rrrr r rrrr}
\toprule
			& \phantom{a} & \multicolumn{4}{c}{$K=3$} & \phantom{a} & \multicolumn{4}{c}{$K=5$}  & \phantom{a}  & \multicolumn{4}{c}{$K=10$} \\
			\cmidrule(r){3-6} 	\cmidrule(r){8-11} 	\cmidrule(r){13-16}
	Image	&&  BP 	& SF (\%) & CSA (\%)	& FAST (\%) &	&  BP 	& SF (\%) & CSA (\%)	& FAST (\%) &	&  BP 	& SF (\%) & CSA (\%)	& FAST (\%)  \\
	\midrule
artificial1 	&& 512$^*$	& 44.5	& 9.6	& 49.2	&&285$^*$	& 77.5	& 33.3	& 154.4	&&150$^*$	& 142.7	& 60.7	& 383.3	 \\
artificial2 	&& 6973$^*$	& 59.5	& 0.0	& 19.7	&&5287$^*$	& 46.4	& 2.2	& 28.7	&&2769$^*$	& 45.1	& 32.9	& 70.5	 \\
artificial3 	&& 6359\phantom{$^*$}	& 23.3	& 11.9	& 67.5	&&4741$^*$	& 38.8	& 7.5	& 59.8	&&2787$^*$	& 76.5	& 44.1	& 31.3	 \\
artificial4 	&& 2301$^*$	& 48.8	& 16.0	& 32.8	&&1479$^*$	& 66.7	& 59.2	& 47.5	&&578$^*$	& 227.7	& 120.2	& 200.3	 \\
artificial5 	&& 180$^*$	& 87.8	& 1.1	& 92.8	&&122$^*$	& 26.2	& 74.6	& 184.4	&&56\phantom{$^*$}	& 78.6	& 305.4	& 519.6	 \\
avatar1 	&& 27$^*$	& 25.9	& 22.2	& 25.9	&&18$^*$	& 44.4	& 16.7	& 88.9	&&7$^*$	& 85.7	& 128.6	& 385.7	 \\
avatar2 	&& 47$^*$	& 4.3	& 19.1	& 55.3	&&31$^*$	& 45.2	& 25.8	& 135.5	&&13$^*$	& 130.8	& 123.1	& 461.5	 \\
avatar3 	&& 47$^*$	& 42.6	& 12.8	& 42.6	&&33$^*$	& 33.3	& 51.5	& 103.0	&&13$^*$	& 69.2	& 184.6	& 415.4	 \\
avatar4 	&& 44$^*$	& 38.6	& 6.8	& 56.8	&&28$^*$	& 25.0	& 10.7	& 146.4	&&10$^*$	& 100.0	& 180.0	& 590.0	 \\
realistic1 	&& 499$^*$	& 29.9	& 38.5	& 78.0	&&325$^*$	& 17.2	& 38.2	& 135.4	&&184\phantom{$^*$}	& 71.2	& 160.9	& 315.8	 \\
realistic2 	&& 1430$^*$	& 27.5	& 0.8	& 106.7	&&725$^*$	& 82.8	& 23.7	& 229.2	&&373$^*$	& 144.5	& 133.8	& 480.4	 \\
realistic3 	&& 538$^*$	& 21.2	& 77.0	& 191.1	&&294\phantom{$^*$}	& 61.9	& 87.8	& 432.7	&&185\phantom{$^*$}	& 38.9	& 95.7	& 746.5	 \\
realistic4 	&& 2023$^*$	& 5.4	& 0.0	& 12.6	&&1255$^*$	& 52.1	& 31.4	& 23.9	&&593$^*$	& 176.7	& 133.7	& 147.0	 \\
realistic5 	&& 1431$^*$	& 23.5	& 12.2	& 62.6	&&679$^*$	& 86.0	& 44.8	& 242.7	&&360$^*$	& 108.1	& 141.1	& 546.4	 \\
toy1 	&& 0$^*$	& (70)	& (0)	& (0)	&&0$^*$	& (44)	& (0)	& (0)	&&0$^*$	& (0)	& (0)	& (0)	 \\
toy2 	&& 0$^*$	& (86)	& (52)	& (182)	&&0$^*$	& (51)	& (32)	& (182)	&&0$^*$	& (13)	& (32)	& (182)	 \\
toy3 	&& 107$^*$	& 60.7	& 29.9	& 93.5	&&72$^*$	& 26.4	& 62.5	& 187.5	&&35$^*$	& 45.7	& 191.4	& 491.4	 \\
toy4 	&& 143$^*$	& 18.9	& 14.7	& 74.8	&&95$^*$	& 68.4	& 63.2	& 163.2	&&59$^*$	& 78.0	& 116.9	& 323.7	 \\
toy5 	&& 158$^*$	& 41.1	& 3.8	& 10.1	&&122$^*$	& 68.0	& 22.1	& 20.5	&&71$^*$	& 109.9	& 63.4	& 94.4	 \\
toy6 	&& 921\phantom{$^*$}	& 21.4	& 22.1	& 66.9	&&601$^*$	& 72.0	& 68.6	& 118.0	&&369\phantom{$^*$}	& 82.7	& 83.2	& 244.7	 \\
toy7 	&& 136$^*$	& 119.9	& 38.2	& 20.6	&&113$^*$	& 61.1	& 46.0	& 15.0	&&72$^*$	& 138.9	& 63.9	& 80.6	 \\
toy8 	&& 212$^*$	& 16.5	& 9.4	& 19.3	&&148$^*$	& 42.6	& 19.6	& 48.6	&&89$^*$	& 92.1	& 60.7	& 109.0	 \\
toy9 	&& 186$^*$	& 83.3	& 37.6	& 15.1	&&75$^*$	& 130.7	& 108.0	& 105.3	&&42$^*$	& 169.0	& 214.3	& 266.7	 \\
toy10 	&& 232$^*$	& 16.8	& 14.2	& 23.7	&&174$^*$	& 49.4	& 24.7	& 64.9	&&96$^*$	& 63.5	& 82.3	& 199.0	 \\
toy11 	&& 297$^*$	& 17.5	& 38.4	& 18.2	&&177$^*$	& 92.7	& 110.7	& 13.0	&&115$^*$	& 63.5	& 39.1	& 59.1	 \\
toy12 	&& 259$^*$	& 42.9	& 0.0	& 28.2	&&151$^*$	& 29.8	& 7.9	& 77.5	&&72$^*$	& 98.6	& 87.5	& 163.9	 \\
toy13 	&& 328$^*$	& 44.2	& 0.0	& 17.7	&&193$^*$	& 125.9	& 30.6	& 40.9	&&101$^*$	& 240.6	& 66.3	& 111.9	 \\
toy14 	&& 366$^*$	& 26.5	& 11.2	& 23.5	&&304$^*$	& 41.4	& 24.7	& 40.8	&&197$^*$	& 94.9	& 57.4	& 93.9	 \\
typical1 	&& 535$^*$	& 11.4	& 12.7	& 43.9	&&365$^*$	& 16.7	& 40.5	& 99.5	&&206\phantom{$^*$}	& 41.7	& 75.2	& 253.4	 \\
typical2 	&& 1811$^*$	& 31.8	& 7.6	& 54.1	&&1180$^*$	& 51.5	& 36.2	& 96.0	&&716\phantom{$^*$}	& 64.0	& 79.2	& 204.3	 \\
typical3 	&& 202\phantom{$^*$}	& 9.4	& 19.8	& 88.1	&&162\phantom{$^*$}	& 24.7	& 28.4	& 134.6	&&109\phantom{$^*$}	& 64.2	& 41.3	& 248.6	 \\
typical4 	&& 2203$^*$	& 5.9	& 0.0	& 7.4	&&1590$^*$	& 33.1	& 24.7	& 17.4	&&938$^*$	& 77.9	& 92.6	& 54.8	 \\
typical5 	&& 1318$^*$	& 31.4	& 8.2	& 19.9	&&861$^*$	& 47.9	& 22.9	& 69.7	&&482$^*$	& 28.0	& 74.1	& 203.1	 \\
typical6 	&& 283$^*$	& 51.6	& 17.3	& 64.7	&&190$^*$	& 63.7	& 41.6	& 135.8	&&95\phantom{$^*$}	& 38.9	& 107.4	& 371.6	 \\
typical7 	&& 579$^*$	& 4.8	& 1.2	& 25.7	&&380$^*$	& 16.8	& 21.8	& 67.4	&&193$^*$	& 57.0	& 96.9	& 215.0	 \\
typical8 	&& 1577$^*$	& 33.5	& 0.0	& 10.3	&&986$^*$	& 45.4	& 16.9	& 51.6	&&573$^*$	& 73.8	& 67.5	& 105.1	 \\
typical10 	&& 934$^*$	& 30.6	& 11.6	& 36.3	&&681$^*$	& 43.0	& 37.6	& 48.0	&&377\phantom{$^*$}	& 100.5	& 127.1	& 153.8	 \\
\midrule
\multicolumn{3}{c}{Mean $\pm$ Std. Dev.} & 34.4 $\pm$ 25.1	 & 15.0 $\pm$ 16.0	 & 47.3 $\pm$ 37.3		&&&	53.0 $\pm$ 27.6	 & 39.0 $\pm$ 26.5	 & 103.7 $\pm$ 83.7		&&&	94.8 $\pm$ 50.4	 & 106.6 $\pm$ 57.7	 & 266.9 $\pm$ 178.8 \\
\bottomrule
\end{tabular}}\label{tab:results-simple-1}
\end{center}
\end{sidewaystable}

\begin{table}[!htbp]
\caption{Accuracy of the algorithms: $K=15, 20$\\ (*) Proven optimal solution in $1$-hour running time}
\footnotesize     
\begin{center}
\begin{tabular}{r r rrrr r rrrr}
\toprule
			& \phantom{a} & \multicolumn{4}{c}{$K=15$} & \phantom{a} & \multicolumn{4}{c}{$K=20$}  \\
			\cmidrule(r){3-6} 	\cmidrule(r){8-11} 	
	Image	&&  BP 	& SF (\%) & CSA (\%)	& FAST (\%) &	&  BP 	& SF (\%) & CSA (\%)	& FAST (\%)   \\
	\midrule
artificial1 	&& 95$^*$	& 146.3	& 148.4	& 663.2	&&65\phantom{$^*$}	& 227.7	& 240.0	& 1015.4	 \\
artificial2 	&& 1510$^*$	& 96.6	& 135.4	& 152.0	&&1147\phantom{$^*$}	& 127.0	& 129.2	& 194.8	 \\
artificial3 	&& 1631$^*$	& 138.6	& 91.8	& 86.8	&&1102\phantom{$^*$}	& 193.7	& 144.9	& 135.1	 \\
artificial4 	&& 459\phantom{$^*$}	& 155.3	& 199.6	& 278.2	&&387\phantom{$^*$}	& 146.5	& 198.7	& 348.6	 \\
artificial5 	&& 33\phantom{$^*$}	& 130.3	& 581.8	& 951.5	&&22\phantom{$^*$}	& 218.2	& 709.1	& 1477.3	 \\
avatar1 	&& 0$^*$	& (12)	& (8)	& (34)	&&0$^*$	& (12)	& (8)	& (34)	 \\
avatar2 	&& 4$^*$	& 475.0	& 625.0	& 1725.0	&&0$^*$	& (22)	& (20)	& (73)	 \\
avatar3 	&& 3$^*$	& 566.7	& 800.0	& 2133.3	&&0$^*$	& (20)	& (25)	& (67)	 \\
avatar4 	&& 2$^*$	& 600.0	& 1050.0	& 3350.0	&&0$^*$	& (12)	& (17)	& (69)	 \\
realistic1 	&& 127\phantom{$^*$}	& 112.6	& 184.3	& 502.4	&&92\phantom{$^*$}	& 158.7	& 201.1	& 731.5	 \\
realistic2 	&& 262\phantom{$^*$}	& 97.7	& 168.3	& 726.3	&&210\phantom{$^*$}	& 90.0	& 379.5	& 931.0	 \\
realistic3 	&& 127\phantom{$^*$}	& 33.1	& 203.1	& 1133.1	&&137\phantom{$^*$}	& -1.5	& 243.1	& 1043.1	 \\
realistic4 	&& 431\phantom{$^*$}	& 103.2	& 183.5	& 239.9	&&333\phantom{$^*$}	& 113.2	& 270.6	& 339.9	 \\
realistic5 	&& 259\phantom{$^*$}	& 114.7	& 132.0	& 798.5	&&190\phantom{$^*$}	& 123.2	& 254.7	& 1124.7	 \\
toy1 	&& 0$^*$	& (0)	& (0)	& (0)	&&0$^*$	& (0)	& (0)	& (0)	 \\
toy2 	&& 0$^*$	& (0)	& (9)	& (182)	&&0$^*$	& (0)	& (11)	& (182)	 \\
toy3 	&& 20$^*$	& 90.0	& 330.0	& 935.0	&&8$^*$	& 262.5	& 900.0	& 2487.5	 \\
toy4 	&& 40$^*$	& 77.5	& 120.0	& 525.0	&&26$^*$	& 165.4	& 280.8	& 861.5	 \\
toy5 	&& 47$^*$	& 148.9	& 136.2	& 193.6	&&30$^*$	& 216.7	& 196.7	& 360.0	 \\
toy6 	&& 267\phantom{$^*$}	& 58.1	& 108.2	& 376.4	&&210\phantom{$^*$}	& 93.3	& 158.6	& 505.7	 \\
toy7 	&& 50$^*$	& 152.0	& 132.0	& 160.0	&&32$^*$	& 250.0	& 203.1	& 306.3	 \\
toy8 	&& 63$^*$	& 71.4	& 81.0	& 195.2	&&46$^*$	& 95.7	& 137.0	& 304.3	 \\
toy9 	&& 27$^*$	& 170.4	& 385.2	& 470.4	&&17$^*$	& 252.9	& 282.4	& 805.9	 \\
toy10 	&& 62$^*$	& 106.5	& 141.9	& 362.9	&&37$^*$	& 175.7	& 256.8	& 675.7	 \\
toy11 	&& 86$^*$	& 77.9	& 66.3	& 112.8	&&63$^*$	& 141.3	& 119.0	& 190.5	 \\
toy12 	&& 48$^*$	& 106.3	& 118.8	& 295.8	&&28$^*$	& 196.4	& 307.1	& 578.6	 \\
toy13 	&& 73$^*$	& 305.5	& 97.3	& 193.2	&&51$^*$	& 374.5	& 103.9	& 319.6	 \\
toy14 	&& 146$^*$	& 147.9	& 90.4	& 161.6	&&111$^*$	& 184.7	& 118.9	& 244.1	 \\
typical1 	&& 138\phantom{$^*$}	& 51.4	& 153.6	& 427.5	&&100\phantom{$^*$}	& 57.0	& 156.0	& 628.0	 \\
typical2 	&& 462\phantom{$^*$}	& 122.1	& 116.0	& 371.6	&&372\phantom{$^*$}	& 113.2	& 159.7	& 485.8	 \\
typical3 	&& 65\phantom{$^*$}	& 129.2	& 256.9	& 484.6	&&55\phantom{$^*$}	& 167.3	& 280.0	& 590.9	 \\
typical4 	&& 639\phantom{$^*$}	& 182.8	& 134.3	& 115.3	&&482\phantom{$^*$}	& 127.2	& 138.0	& 185.5	 \\
typical5 	&& 329\phantom{$^*$}	& 48.0	& 91.5	& 344.1	&&253\phantom{$^*$}	& 63.2	& 131.6	& 477.5	 \\
typical6 	&& 58\phantom{$^*$}	& 46.6	& 243.1	& 672.4	&&38\phantom{$^*$}	& 105.3	& 326.3	& 1078.9	 \\
typical7 	&& 122\phantom{$^*$}	& 71.3	& 153.3	& 398.4	&&86\phantom{$^*$}	& 129.1	& 255.8	& 607.0	 \\
typical8 	&& 397\phantom{$^*$}	& 76.8	& 91.9	& 183.9	&&300\phantom{$^*$}	& 94.3	& 157.7	& 275.7	 \\
typical10 	&& 267\phantom{$^*$}	& 165.2	& 151.7	& 258.4	&&200\phantom{$^*$}	& 177.0	& 188.0	& 378.5	 \\
\midrule
\multicolumn{3}{c}{Mean $\pm$ Std. Dev.} & 152.2 $\pm$ 135.8	 & 226.6 $\pm$ 219.9	 & 587.6 $\pm$ 663.6		&&&	156.1 $\pm$ 73.9	 & 246.1 $\pm$ 166.4	 & 635.1 $\pm$ 479.3 \\	
\bottomrule
\end{tabular}\label{tab:results-simple-2}
\end{center}
\end{table}

\begin{table}[!htbp]
\caption{Accuracy of the heuristics: MPEG 7 benchmark category}
\footnotesize     
\begin{center}
\begin{tabular}{r c ccc c ccc}
\toprule
		& \phantom{a} & \multicolumn{3}{c}{$K=3$} & \phantom{a} & \multicolumn{3}{c}{$K=5$} \\
		\cmidrule(r){2-5} 	\cmidrule(r){7-9}	
	Category	& &  SF (\%) & CSA (\%)  & FAST (\%)	&	&  SF (\%) 	& CSA (\%) & FAST (\%) \\
	\midrule
bat	 & & 	51.34	$\pm$	20.72	 & 	8.99	$\pm$	7.47	 & 	27.61	$\pm$	14.22	 & & 	55.94	$\pm$	17.94	 & 	30.10	$\pm$	10.48	 & 	50.30	$\pm$	16.46	\\
device5	 & & 	182.53	$\pm$	147.85	 & 	40.29	$\pm$	57.90	 & 	36.39	$\pm$	39.29	 & & 	114.10	$\pm$	68.33	 & 	37.42	$\pm$	38.25	 & 	49.83	$\pm$	40.61	\\
dog	 & & 	28.66	$\pm$	17.02	 & 	10.90	$\pm$	7.81	 & 	26.12	$\pm$	15.09	 & & 	47.48	$\pm$	23.08	 & 	23.04	$\pm$	11.40	 & 	43.36	$\pm$	21.43	\\
key	 & & 	52.82	$\pm$	37.96	 & 	25.99	$\pm$	22.07	 & 	52.81	$\pm$	50.12	 & & 	58.38	$\pm$	25.59	 & 	43.04	$\pm$	20.21	 & 	112.02	$\pm$	74.67	\\
Misk	 & & 	43.15	$\pm$	9.05	 & 	22.67	$\pm$	9.13	 & 	39.68	$\pm$	15.75	 & & 	60.70	$\pm$	14.77	 & 	36.54	$\pm$	11.32	 & 	70.23	$\pm$	33.92	\\
	
	\midrule \midrule 
	
		& \phantom{a} & \multicolumn{3}{c}{$K=10$} & \phantom{a} & \multicolumn{3}{c}{$K=15$}   \\
		\cmidrule(r){2-5} 	\cmidrule(r){7-9}	
	Category	& &  SF (\%) &  CSA (\%)  & FAST (\%)	&	&  SF (\%) 	&  CSA (\%) & FAST (\%)  \\
	\midrule
bat	 & & 	76.70	$\pm$	34.79	 & 	62.67	$\pm$	13.53	 & 	110.94	$\pm$	39.89	 & & 	94.75	$\pm$	40.39	 & 	103.59	$\pm$	18.35	 & 	185.30	$\pm$	63.71	\\
device5	 & & 	166.01	$\pm$	77.71	 & 	84.06	$\pm$	50.95	 & 	125.26	$\pm$	91.08	 & & 	170.57	$\pm$	84.24	 & 	114.53	$\pm$	80.65	 & 	190.48	$\pm$	118.67	\\
dog	 & & 	61.99	$\pm$	23.58	 & 	47.22	$\pm$	25.55	 & 	84.88	$\pm$	45.67	 & & 	72.69	$\pm$	22.92	 & 	71.08	$\pm$	38.58	 & 	136.96	$\pm$	80.49	\\
key	 & & 	74.37	$\pm$	30.53	 & 	110.26	$\pm$	40.57	 & 	282.73	$\pm$	131.24	 & & 	97.35	$\pm$	47.52	 & 	195.41	$\pm$	79.82	 & 	470.81	$\pm$	195.04	\\
Misk	 & & 	69.96	$\pm$	34.39	 & 	83.76	$\pm$	16.23	 & 	160.79	$\pm$	69.50	 & & 	65.84	$\pm$	51.69	 & 	102.41	$\pm$	16.98	 & 	225.72	$\pm$	85.18	\\

\midrule \midrule 

		& \phantom{a} & \multicolumn{3}{c}{$K=20$} \\
		\cmidrule(r){2-5} 	
	Category	& &  SF (\%) &  CSA (\%)  & FAST (\%)	\\
	\cmidrule{2-5}
bat	 & & 	102.83	$\pm$	43.50	 & 	140.20	$\pm$	23.90	 & 	268.72	$\pm$	89.28	\\
device5	 & & 	167.56	$\pm$	105.03	 & 	155.67	$\pm$	122.09	 & 	247.48	$\pm$	144.18	\\
dog	 & & 	91.21	$\pm$	30.76	 & 	100.67	$\pm$	61.98	 & 	195.58	$\pm$	118.57	\\
key	 & & 	101.71	$\pm$	55.17	 & 	277.82	$\pm$	137.49	 & 	668.53	$\pm$	290.36	\\
Misk	 & & 	57.20	$\pm$	50.65	 & 	110.09	$\pm$	24.90	 & 	269.25	$\pm$	101.10	\\

\bottomrule
\end{tabular}
\label{tab:results-mpeg7}
\end{center}
\end{table}

We have observed that the restricted LP master tends to yield integer solutions. Approximately \textcolor{red}{$0.5$\%} of the runs yielded fractional solutions. Besides, most of the time an integer solution is acquired in very few branchings. We have also observed that, some problems with fractional optimal solutions had also alternative optimal integer solutions. This explains why only a few branchings are sufficient, mostly.

In the early iterations of the column generation, the objective value decreases rapidly and approaches the optimum value. For the aforementioned reasons, the 1-hour CPU time limit does not effect its performance dramatically for the instances of the first four sets. However, this is not true for the last set.

For the most of the instances, BP is able to find an optimal solution within the 1-hour CPU time limit except the fifth group of test problems. These cases are marked with an asterisk in the tables. Optimal rectangle blankets obtained for the artificial1 - artificial5 data sets are illustrated in Table \ref{tab:optimum_blanket}. They can be compared with the original ones given in Figure \ref{fig:benchmark-2}. Observe the increase in the quality of the approximation with the increasing $K$ values.† \textcolor{red}{ It is observed in  Table~\ref{tab:results-leather-1} and Table~\ref{tab:results-leather-2} that BP cannot find an optimal solution in 1-hour CPU time limit.} They turned out to be the most challenging test problems of the test instances. Still, BP gives the best results in one hour. → \textcolor{red}{Similarly, we observe considerable decreases in the performance of the heuristics, in parallel with the decrease in performance of the BP. } Also, it is possible to observe the considerable increase in the performance of SF. BP may have produced an optimum solution for certain instances. Nevertheless, necessary columns are not generated within the time limit to prove optimality. For example for toy2 instance (see Figure~\ref{fig:benchmark-2}), it is easy to come up with the optimum for $K>2$ manually where the objective value is zero.

\begin{sidewaystable}
\renewcommand{\arraystretch}{1.5}
\caption{Optimum rectangle blankets for artificial1 - artificial5 with $K = 3, 5, 10, 15, 20$}
\begin{center}
\begin{tabular}{c c c c c c}
\toprule
$K=3$ & $K=5$ & $K=10$ & $K=15$ & $K=20$ & \\
\midrule
\begin{tikzpicture}[yscale=-1]
\definecolor{myyellow}{RGB}{255,217,101}
\filldraw[fill=myyellow, draw=black] (0.85,1.65) rectangle (3.45,2.60);\filldraw[fill=myyellow, draw=black] (0.00,0.05) rectangle (2.20,1.65);\filldraw[fill=myyellow, draw=black] (1.75,2.60) rectangle (3.50,3.50);\end{tikzpicture} & \begin{tikzpicture}[yscale=-1]
\definecolor{myyellow}{RGB}{255,217,101}
\filldraw[fill=myyellow, draw=black] (2.60,1.75) rectangle (3.50,3.50);\filldraw[fill=myyellow, draw=black] (0.00,0.00) rectangle (1.65,1.70);\filldraw[fill=myyellow, draw=black] (1.95,0.95) rectangle (2.60,3.50);\filldraw[fill=myyellow, draw=black] (0.85,1.70) rectangle (1.65,2.60);\filldraw[fill=myyellow, draw=black] (1.65,0.60) rectangle (1.95,3.20);\end{tikzpicture} & \begin{tikzpicture}[yscale=-1]
\definecolor{myyellow}{RGB}{255,217,101}
\filldraw[fill=myyellow, draw=black] (0.25,1.50) rectangle (2.85,1.75);\filldraw[fill=myyellow, draw=black] (1.95,3.20) rectangle (3.30,3.50);\filldraw[fill=myyellow, draw=black] (1.00,2.20) rectangle (1.65,2.60);\filldraw[fill=myyellow, draw=black] (0.65,1.75) rectangle (3.25,1.95);\filldraw[fill=myyellow, draw=black] (0.20,0.00) rectangle (1.50,0.25);\filldraw[fill=myyellow, draw=black] (1.65,2.20) rectangle (3.50,3.20);\filldraw[fill=myyellow, draw=black] (0.00,1.00) rectangle (2.60,1.50);\filldraw[fill=myyellow, draw=black] (0.80,1.95) rectangle (3.40,2.20);\filldraw[fill=myyellow, draw=black] (1.75,0.75) rectangle (2.20,1.00);\filldraw[fill=myyellow, draw=black] (0.00,0.25) rectangle (1.75,1.00);\end{tikzpicture} & \begin{tikzpicture}[yscale=-1]
\definecolor{myyellow}{RGB}{255,217,101}
\filldraw[fill=myyellow, draw=black] (3.25,1.95) rectangle (3.40,3.35);\filldraw[fill=myyellow, draw=black] (0.00,0.20) rectangle (0.20,1.50);\filldraw[fill=myyellow, draw=black] (2.90,1.75) rectangle (3.25,3.50);\filldraw[fill=myyellow, draw=black] (1.75,0.65) rectangle (1.95,3.20);\filldraw[fill=myyellow, draw=black] (3.40,2.10) rectangle (3.50,3.20);\filldraw[fill=myyellow, draw=black] (0.20,0.05) rectangle (0.40,1.65);\filldraw[fill=myyellow, draw=black] (1.60,0.40) rectangle (1.75,2.90);\filldraw[fill=myyellow, draw=black] (1.95,0.80) rectangle (2.20,3.40);\filldraw[fill=myyellow, draw=black] (2.50,1.25) rectangle (2.65,3.50);\filldraw[fill=myyellow, draw=black] (2.65,1.60) rectangle (2.90,3.50);\filldraw[fill=myyellow, draw=black] (1.35,0.10) rectangle (1.60,2.65);\filldraw[fill=myyellow, draw=black] (0.95,0.00) rectangle (1.35,2.50);\filldraw[fill=myyellow, draw=black] (0.70,1.75) rectangle (0.95,2.10);\filldraw[fill=myyellow, draw=black] (0.40,0.00) rectangle (0.95,1.75);\filldraw[fill=myyellow, draw=black] (2.20,0.95) rectangle (2.50,3.50);\end{tikzpicture} & \begin{tikzpicture}[yscale=-1]
\definecolor{myyellow}{RGB}{255,217,101}
\filldraw[fill=myyellow, draw=black] (3.40,2.10) rectangle (3.50,3.20);\filldraw[fill=myyellow, draw=black] (0.00,0.25) rectangle (0.10,1.40);\filldraw[fill=myyellow, draw=black] (2.20,0.95) rectangle (2.50,1.25);\filldraw[fill=myyellow, draw=black] (1.20,0.00) rectangle (1.40,2.60);\filldraw[fill=myyellow, draw=black] (1.65,0.40) rectangle (1.75,3.00);\filldraw[fill=myyellow, draw=black] (1.50,0.20) rectangle (1.65,2.75);\filldraw[fill=myyellow, draw=black] (0.10,0.15) rectangle (0.20,1.50);\filldraw[fill=myyellow, draw=black] (0.65,0.00) rectangle (0.80,1.95);\filldraw[fill=myyellow, draw=black] (1.80,0.70) rectangle (1.95,3.25);\filldraw[fill=myyellow, draw=black] (1.40,0.10) rectangle (1.50,2.65);\filldraw[fill=myyellow, draw=black] (0.80,0.00) rectangle (0.95,2.20);\filldraw[fill=myyellow, draw=black] (3.25,1.95) rectangle (3.40,3.35);\filldraw[fill=myyellow, draw=black] (1.75,0.60) rectangle (1.80,3.10);\filldraw[fill=myyellow, draw=black] (0.20,0.05) rectangle (0.40,1.65);\filldraw[fill=myyellow, draw=black] (2.20,1.60) rectangle (2.90,3.50);\filldraw[fill=myyellow, draw=black] (1.95,0.80) rectangle (2.20,3.40);\filldraw[fill=myyellow, draw=black] (0.40,0.00) rectangle (0.65,1.75);\filldraw[fill=myyellow, draw=black] (2.90,1.75) rectangle (3.25,3.50);\filldraw[fill=myyellow, draw=black] (2.20,1.25) rectangle (2.65,1.60);\filldraw[fill=myyellow, draw=black] (0.95,0.00) rectangle (1.20,2.45);\end{tikzpicture} & \\
\begin{tikzpicture}[yscale=-1]
\definecolor{myyellow}{RGB}{255,217,101}
\filldraw[fill=myyellow, draw=black] (0.67,0.67) rectangle (1.64,1.16);\filldraw[fill=myyellow, draw=black] (1.82,0.00) rectangle (2.83,0.49);\filldraw[fill=myyellow, draw=black] (2.84,0.50) rectangle (3.34,1.00);\end{tikzpicture} & \begin{tikzpicture}[yscale=-1]
\definecolor{myyellow}{RGB}{255,217,101}
\filldraw[fill=myyellow, draw=black] (1.91,0.00) rectangle (2.83,0.48);\filldraw[fill=myyellow, draw=black] (2.84,0.50) rectangle (3.34,1.00);\filldraw[fill=myyellow, draw=black] (0.16,0.16) rectangle (0.66,0.66);\filldraw[fill=myyellow, draw=black] (1.59,0.32) rectangle (1.91,0.82);\filldraw[fill=myyellow, draw=black] (0.67,0.67) rectangle (1.59,1.16);\end{tikzpicture} & \begin{tikzpicture}[yscale=-1]
\definecolor{myyellow}{RGB}{255,217,101}
\filldraw[fill=myyellow, draw=black] (0.00,0.01) rectangle (0.42,0.32);\filldraw[fill=myyellow, draw=black] (1.89,0.00) rectangle (2.75,0.33);\filldraw[fill=myyellow, draw=black] (0.46,0.56) rectangle (0.96,0.83);\filldraw[fill=myyellow, draw=black] (0.75,0.83) rectangle (1.59,1.16);\filldraw[fill=myyellow, draw=black] (3.08,0.84) rectangle (3.50,1.14);\filldraw[fill=myyellow, draw=black] (2.54,0.33) rectangle (3.04,0.60);\filldraw[fill=myyellow, draw=black] (1.35,0.59) rectangle (1.83,0.83);\filldraw[fill=myyellow, draw=black] (0.18,0.32) rectangle (0.67,0.56);\filldraw[fill=myyellow, draw=black] (2.83,0.60) rectangle (3.32,0.84);\filldraw[fill=myyellow, draw=black] (1.62,0.33) rectangle (2.12,0.59);\end{tikzpicture} & \begin{tikzpicture}[yscale=-1]
\definecolor{myyellow}{RGB}{255,217,101}
\filldraw[fill=myyellow, draw=black] (3.08,0.84) rectangle (3.50,1.14);\filldraw[fill=myyellow, draw=black] (1.31,0.66) rectangle (1.81,0.83);\filldraw[fill=myyellow, draw=black] (1.98,0.00) rectangle (2.66,0.16);\filldraw[fill=myyellow, draw=black] (1.65,0.33) rectangle (2.14,0.50);\filldraw[fill=myyellow, draw=black] (2.84,0.66) rectangle (3.34,0.84);\filldraw[fill=myyellow, draw=black] (2.50,0.33) rectangle (2.99,0.49);\filldraw[fill=myyellow, draw=black] (1.82,0.16) rectangle (2.83,0.33);\filldraw[fill=myyellow, draw=black] (2.66,0.49) rectangle (3.16,0.66);\filldraw[fill=myyellow, draw=black] (0.16,0.32) rectangle (0.66,0.49);\filldraw[fill=myyellow, draw=black] (1.48,0.50) rectangle (1.98,0.66);\filldraw[fill=myyellow, draw=black] (0.00,0.01) rectangle (0.42,0.32);\filldraw[fill=myyellow, draw=black] (0.33,0.49) rectangle (0.83,0.66);\filldraw[fill=myyellow, draw=black] (0.50,0.66) rectangle (1.00,0.83);\filldraw[fill=myyellow, draw=black] (0.67,0.83) rectangle (1.64,1.00);\filldraw[fill=myyellow, draw=black] (0.84,1.00) rectangle (1.51,1.16);\end{tikzpicture} & \begin{tikzpicture}[yscale=-1]
\definecolor{myyellow}{RGB}{255,217,101}
\filldraw[fill=myyellow, draw=black] (2.50,0.33) rectangle (2.99,0.49);\filldraw[fill=myyellow, draw=black] (2.85,0.69) rectangle (3.34,0.83);\filldraw[fill=myyellow, draw=black] (1.99,0.00) rectangle (2.66,0.15);\filldraw[fill=myyellow, draw=black] (1.48,0.50) rectangle (1.98,0.66);\filldraw[fill=myyellow, draw=black] (0.71,0.93) rectangle (1.61,1.01);\filldraw[fill=myyellow, draw=black] (0.66,0.83) rectangle (1.67,0.93);\filldraw[fill=myyellow, draw=black] (1.89,0.15) rectangle (2.78,0.20);\filldraw[fill=myyellow, draw=black] (1.65,0.33) rectangle (2.14,0.50);\filldraw[fill=myyellow, draw=black] (0.32,0.16) rectangle (0.50,0.32);\filldraw[fill=myyellow, draw=black] (2.66,0.49) rectangle (3.16,0.63);\filldraw[fill=myyellow, draw=black] (1.81,0.20) rectangle (2.84,0.33);\filldraw[fill=myyellow, draw=black] (3.17,0.83) rectangle (3.50,1.16);\filldraw[fill=myyellow, draw=black] (2.78,0.63) rectangle (3.25,0.69);\filldraw[fill=myyellow, draw=black] (0.00,0.00) rectangle (0.32,0.32);\filldraw[fill=myyellow, draw=black] (0.50,0.66) rectangle (1.00,0.83);\filldraw[fill=myyellow, draw=black] (0.33,0.49) rectangle (0.83,0.66);\filldraw[fill=myyellow, draw=black] (0.85,1.01) rectangle (1.51,1.16);\filldraw[fill=myyellow, draw=black] (1.31,0.66) rectangle (1.81,0.83);\filldraw[fill=myyellow, draw=black] (3.00,0.83) rectangle (3.17,1.00);\filldraw[fill=myyellow, draw=black] (0.16,0.32) rectangle (0.66,0.49);\end{tikzpicture} & \\
\begin{tikzpicture}[yscale=-1]
\definecolor{myyellow}{RGB}{255,217,101}
\filldraw[fill=myyellow, draw=black] (0.76,2.43) rectangle (2.72,2.94);\filldraw[fill=myyellow, draw=black] (0.00,0.60) rectangle (0.72,2.35);\filldraw[fill=myyellow, draw=black] (1.01,0.00) rectangle (3.21,2.08);\end{tikzpicture} & \begin{tikzpicture}[yscale=-1]
\definecolor{myyellow}{RGB}{255,217,101}
\filldraw[fill=myyellow, draw=black] (0.76,2.43) rectangle (2.72,2.94);\filldraw[fill=myyellow, draw=black] (0.00,0.60) rectangle (0.72,2.35);\filldraw[fill=myyellow, draw=black] (1.03,0.74) rectangle (2.59,1.94);\filldraw[fill=myyellow, draw=black] (2.59,1.24) rectangle (3.31,2.43);\filldraw[fill=myyellow, draw=black] (0.76,0.00) rectangle (3.13,0.49);\end{tikzpicture} & \begin{tikzpicture}[yscale=-1]
\definecolor{myyellow}{RGB}{255,217,101}
\filldraw[fill=myyellow, draw=black] (2.84,0.21) rectangle (3.29,0.76);\filldraw[fill=myyellow, draw=black] (0.31,2.12) rectangle (1.15,2.45);\filldraw[fill=myyellow, draw=black] (1.36,1.54) rectangle (2.14,2.14);\filldraw[fill=myyellow, draw=black] (2.53,1.32) rectangle (3.34,2.45);\filldraw[fill=myyellow, draw=black] (1.36,0.74) rectangle (2.74,1.32);\filldraw[fill=myyellow, draw=black] (0.76,2.45) rectangle (2.72,2.94);\filldraw[fill=myyellow, draw=black] (0.31,0.49) rectangle (1.13,0.84);\filldraw[fill=myyellow, draw=black] (0.00,0.84) rectangle (0.58,2.12);\filldraw[fill=myyellow, draw=black] (1.01,1.03) rectangle (1.36,1.91);\filldraw[fill=myyellow, draw=black] (0.76,0.00) rectangle (2.84,0.49);\end{tikzpicture} & \begin{tikzpicture}[yscale=-1]
\definecolor{myyellow}{RGB}{255,217,101}
\filldraw[fill=myyellow, draw=black] (2.88,0.23) rectangle (3.29,0.76);\filldraw[fill=myyellow, draw=black] (1.32,1.91) rectangle (1.91,2.16);\filldraw[fill=myyellow, draw=black] (0.91,2.68) rectangle (2.59,2.94);\filldraw[fill=myyellow, draw=black] (1.30,0.74) rectangle (2.57,1.03);\filldraw[fill=myyellow, draw=black] (2.35,1.30) rectangle (3.19,1.56);\filldraw[fill=myyellow, draw=black] (1.03,1.03) rectangle (2.90,1.30);\filldraw[fill=myyellow, draw=black] (0.31,0.49) rectangle (1.13,0.84);\filldraw[fill=myyellow, draw=black] (0.31,2.12) rectangle (1.15,2.45);\filldraw[fill=myyellow, draw=black] (1.01,1.30) rectangle (1.65,1.56);\filldraw[fill=myyellow, draw=black] (1.03,1.56) rectangle (2.18,1.91);\filldraw[fill=myyellow, draw=black] (0.62,2.45) rectangle (2.88,2.68);\filldraw[fill=myyellow, draw=black] (0.00,0.84) rectangle (0.58,2.12);\filldraw[fill=myyellow, draw=black] (2.35,2.16) rectangle (3.19,2.45);\filldraw[fill=myyellow, draw=black] (2.66,1.56) rectangle (3.50,2.16);\filldraw[fill=myyellow, draw=black] (0.76,0.00) rectangle (2.88,0.49);\end{tikzpicture} & \begin{tikzpicture}[yscale=-1]
\definecolor{myyellow}{RGB}{255,217,101}
\filldraw[fill=myyellow, draw=black] (2.66,1.56) rectangle (3.50,2.14);\filldraw[fill=myyellow, draw=black] (0.54,1.89) rectangle (0.84,2.16);\filldraw[fill=myyellow, draw=black] (0.54,2.16) rectangle (1.15,2.45);\filldraw[fill=myyellow, draw=black] (0.91,0.00) rectangle (3.01,0.27);\filldraw[fill=myyellow, draw=black] (2.29,2.35) rectangle (3.05,2.45);\filldraw[fill=myyellow, draw=black] (2.35,1.30) rectangle (3.19,1.56);\filldraw[fill=myyellow, draw=black] (0.29,0.56) rectangle (0.54,2.39);\filldraw[fill=myyellow, draw=black] (2.39,2.14) rectangle (3.19,2.35);\filldraw[fill=myyellow, draw=black] (0.62,2.45) rectangle (2.88,2.68);\filldraw[fill=myyellow, draw=black] (1.32,1.91) rectangle (1.89,2.16);\filldraw[fill=myyellow, draw=black] (0.54,0.49) rectangle (1.15,0.78);\filldraw[fill=myyellow, draw=black] (0.00,0.82) rectangle (0.29,2.12);\filldraw[fill=myyellow, draw=black] (2.74,0.27) rectangle (3.29,0.76);\filldraw[fill=myyellow, draw=black] (0.91,2.68) rectangle (2.59,2.94);\filldraw[fill=myyellow, draw=black] (1.03,1.03) rectangle (2.90,1.30);\filldraw[fill=myyellow, draw=black] (1.03,1.56) rectangle (2.18,1.91);\filldraw[fill=myyellow, draw=black] (0.54,0.78) rectangle (0.84,1.05);\filldraw[fill=myyellow, draw=black] (1.32,0.74) rectangle (2.57,1.03);\filldraw[fill=myyellow, draw=black] (0.62,0.27) rectangle (2.74,0.49);\filldraw[fill=myyellow, draw=black] (1.01,1.30) rectangle (1.65,1.56);\end{tikzpicture} & \\
\begin{tikzpicture}[yscale=-1]
\definecolor{myyellow}{RGB}{255,217,101}
\filldraw[fill=myyellow, draw=black] (0.34,0.24) rectangle (1.57,2.22);\filldraw[fill=myyellow, draw=black] (2.05,1.50) rectangle (3.33,2.82);\filldraw[fill=myyellow, draw=black] (1.57,0.92) rectangle (2.51,1.50);\end{tikzpicture} & \begin{tikzpicture}[yscale=-1]
\definecolor{myyellow}{RGB}{255,217,101}
\filldraw[fill=myyellow, draw=black] (1.06,0.24) rectangle (1.57,2.51);\filldraw[fill=myyellow, draw=black] (0.34,0.02) rectangle (0.72,1.81);\filldraw[fill=myyellow, draw=black] (1.57,0.92) rectangle (2.51,1.50);\filldraw[fill=myyellow, draw=black] (2.05,1.50) rectangle (3.33,2.82);\filldraw[fill=myyellow, draw=black] (0.72,0.80) rectangle (1.06,2.20);\end{tikzpicture} & \begin{tikzpicture}[yscale=-1]
\definecolor{myyellow}{RGB}{255,217,101}
\filldraw[fill=myyellow, draw=black] (1.06,0.24) rectangle (1.57,2.51);\filldraw[fill=myyellow, draw=black] (0.34,0.02) rectangle (0.72,1.81);\filldraw[fill=myyellow, draw=black] (0.72,0.80) rectangle (1.06,2.20);\filldraw[fill=myyellow, draw=black] (1.57,0.92) rectangle (2.44,1.26);\filldraw[fill=myyellow, draw=black] (0.02,0.92) rectangle (0.34,1.47);\filldraw[fill=myyellow, draw=black] (3.16,2.00) rectangle (3.50,2.68);\filldraw[fill=myyellow, draw=black] (1.57,1.26) rectangle (2.73,1.47);\filldraw[fill=myyellow, draw=black] (2.85,1.67) rectangle (3.16,2.75);\filldraw[fill=myyellow, draw=black] (1.91,2.03) rectangle (2.15,2.68);\filldraw[fill=myyellow, draw=black] (2.15,1.47) rectangle (2.85,2.92);\end{tikzpicture} & \begin{tikzpicture}[yscale=-1]
\definecolor{myyellow}{RGB}{255,217,101}
\filldraw[fill=myyellow, draw=black] (2.44,1.28) rectangle (2.78,2.92);\filldraw[fill=myyellow, draw=black] (1.91,2.08) rectangle (2.05,2.63);\filldraw[fill=myyellow, draw=black] (1.06,0.24) rectangle (1.57,2.51);\filldraw[fill=myyellow, draw=black] (0.72,0.80) rectangle (1.06,2.17);\filldraw[fill=myyellow, draw=black] (1.57,0.92) rectangle (2.15,1.47);\filldraw[fill=myyellow, draw=black] (0.34,0.05) rectangle (0.46,1.74);\filldraw[fill=myyellow, draw=black] (2.05,1.93) rectangle (2.15,2.75);\filldraw[fill=myyellow, draw=black] (0.00,0.97) rectangle (0.17,1.40);\filldraw[fill=myyellow, draw=black] (0.17,0.87) rectangle (0.34,1.50);\filldraw[fill=myyellow, draw=black] (2.90,1.67) rectangle (3.14,2.75);\filldraw[fill=myyellow, draw=black] (2.15,0.94) rectangle (2.44,2.90);\filldraw[fill=myyellow, draw=black] (0.46,0.00) rectangle (0.72,1.83);\filldraw[fill=myyellow, draw=black] (3.31,2.05) rectangle (3.50,2.66);\filldraw[fill=myyellow, draw=black] (3.14,1.96) rectangle (3.31,2.70);\filldraw[fill=myyellow, draw=black] (2.78,1.54) rectangle (2.90,2.82);\end{tikzpicture} & \begin{tikzpicture}[yscale=-1]
\definecolor{myyellow}{RGB}{255,217,101}
\filldraw[fill=myyellow, draw=black] (0.34,0.00) rectangle (0.65,1.79);\filldraw[fill=myyellow, draw=black] (1.64,0.92) rectangle (2.15,1.47);\filldraw[fill=myyellow, draw=black] (0.80,0.82) rectangle (1.06,2.22);\filldraw[fill=myyellow, draw=black] (2.97,1.69) rectangle (3.14,2.75);\filldraw[fill=myyellow, draw=black] (2.15,0.94) rectangle (2.41,2.90);\filldraw[fill=myyellow, draw=black] (0.00,0.97) rectangle (0.22,1.42);\filldraw[fill=myyellow, draw=black] (1.11,2.41) rectangle (1.50,2.51);\filldraw[fill=myyellow, draw=black] (1.57,0.80) rectangle (1.64,1.57);\filldraw[fill=myyellow, draw=black] (0.72,0.77) rectangle (0.80,2.08);\filldraw[fill=myyellow, draw=black] (1.16,0.22) rectangle (1.47,0.29);\filldraw[fill=myyellow, draw=black] (0.22,0.84) rectangle (0.34,1.54);\filldraw[fill=myyellow, draw=black] (3.14,1.96) rectangle (3.31,2.70);\filldraw[fill=myyellow, draw=black] (1.06,0.29) rectangle (1.57,2.41);\filldraw[fill=myyellow, draw=black] (0.65,0.10) rectangle (0.72,1.96);\filldraw[fill=myyellow, draw=black] (2.41,1.23) rectangle (2.66,2.92);\filldraw[fill=myyellow, draw=black] (1.91,2.08) rectangle (2.05,2.63);\filldraw[fill=myyellow, draw=black] (2.05,1.93) rectangle (2.15,2.73);\filldraw[fill=myyellow, draw=black] (3.31,2.05) rectangle (3.50,2.66);\filldraw[fill=myyellow, draw=black] (2.66,1.38) rectangle (2.80,2.90);\filldraw[fill=myyellow, draw=black] (2.80,1.59) rectangle (2.97,2.80);\end{tikzpicture} & \\
\begin{tikzpicture}[yscale=-1]
\definecolor{myyellow}{RGB}{255,217,101}
\filldraw[fill=myyellow, draw=black] (0.38,2.58) rectangle (2.69,3.28);\filldraw[fill=myyellow, draw=black] (0.48,0.00) rectangle (3.34,1.08);\filldraw[fill=myyellow, draw=black] (0.00,1.08) rectangle (3.50,2.58);\end{tikzpicture} & \begin{tikzpicture}[yscale=-1]
\definecolor{myyellow}{RGB}{255,217,101}
\filldraw[fill=myyellow, draw=black] (0.00,1.08) rectangle (3.50,2.53);\filldraw[fill=myyellow, draw=black] (0.70,0.00) rectangle (3.12,0.22);\filldraw[fill=myyellow, draw=black] (0.48,0.22) rectangle (3.34,1.08);\filldraw[fill=myyellow, draw=black] (0.27,2.53) rectangle (2.85,2.91);\filldraw[fill=myyellow, draw=black] (0.48,2.91) rectangle (2.64,3.28);\end{tikzpicture} & \begin{tikzpicture}[yscale=-1]
\definecolor{myyellow}{RGB}{255,217,101}
\filldraw[fill=myyellow, draw=black] (0.27,2.64) rectangle (2.85,2.85);\filldraw[fill=myyellow, draw=black] (0.48,0.43) rectangle (3.34,1.02);\filldraw[fill=myyellow, draw=black] (3.45,1.24) rectangle (3.50,2.05);\filldraw[fill=myyellow, draw=black] (0.43,2.85) rectangle (2.69,3.12);\filldraw[fill=myyellow, draw=black] (0.27,1.02) rectangle (3.45,1.29);\filldraw[fill=myyellow, draw=black] (0.65,3.12) rectangle (2.48,3.28);\filldraw[fill=myyellow, draw=black] (0.16,2.48) rectangle (3.12,2.64);\filldraw[fill=myyellow, draw=black] (0.00,1.29) rectangle (3.45,2.48);\filldraw[fill=myyellow, draw=black] (0.54,0.22) rectangle (3.28,0.43);\filldraw[fill=myyellow, draw=black] (0.70,0.00) rectangle (3.12,0.22);\end{tikzpicture} & \begin{tikzpicture}[yscale=-1]
\definecolor{myyellow}{RGB}{255,217,101}
\filldraw[fill=myyellow, draw=black] (0.70,3.12) rectangle (2.48,3.28);\filldraw[fill=myyellow, draw=black] (0.86,0.00) rectangle (2.96,0.05);\filldraw[fill=myyellow, draw=black] (0.00,1.40) rectangle (0.11,2.48);\filldraw[fill=myyellow, draw=black] (3.45,1.24) rectangle (3.50,2.05);\filldraw[fill=myyellow, draw=black] (3.28,0.43) rectangle (3.34,2.42);\filldraw[fill=myyellow, draw=black] (3.34,0.92) rectangle (3.45,2.37);\filldraw[fill=myyellow, draw=black] (0.43,0.81) rectangle (0.48,3.07);\filldraw[fill=myyellow, draw=black] (0.11,1.24) rectangle (0.27,2.58);\filldraw[fill=myyellow, draw=black] (3.12,0.22) rectangle (3.28,2.53);\filldraw[fill=myyellow, draw=black] (0.27,1.02) rectangle (0.43,2.91);\filldraw[fill=myyellow, draw=black] (2.85,0.05) rectangle (3.12,2.64);\filldraw[fill=myyellow, draw=black] (0.48,0.43) rectangle (0.54,3.12);\filldraw[fill=myyellow, draw=black] (0.70,2.80) rectangle (2.69,3.12);\filldraw[fill=myyellow, draw=black] (0.54,0.22) rectangle (0.70,3.18);\filldraw[fill=myyellow, draw=black] (0.70,0.05) rectangle (2.85,2.80);\end{tikzpicture} & \begin{tikzpicture}[yscale=-1]
\definecolor{myyellow}{RGB}{255,217,101}
\filldraw[fill=myyellow, draw=black] (3.23,0.27) rectangle (3.28,2.48);\filldraw[fill=myyellow, draw=black] (0.70,3.23) rectangle (2.37,3.28);\filldraw[fill=myyellow, draw=black] (0.86,0.00) rectangle (2.96,0.05);\filldraw[fill=myyellow, draw=black] (0.00,1.40) rectangle (0.11,2.37);\filldraw[fill=myyellow, draw=black] (3.45,1.24) rectangle (3.50,2.05);\filldraw[fill=myyellow, draw=black] (0.54,0.22) rectangle (2.75,0.43);\filldraw[fill=myyellow, draw=black] (0.27,1.02) rectangle (0.43,2.91);\filldraw[fill=myyellow, draw=black] (3.39,1.02) rectangle (3.45,2.32);\filldraw[fill=myyellow, draw=black] (0.11,1.24) rectangle (0.27,2.58);\filldraw[fill=myyellow, draw=black] (0.70,0.05) rectangle (3.12,0.22);\filldraw[fill=myyellow, draw=black] (0.43,0.81) rectangle (0.48,3.07);\filldraw[fill=myyellow, draw=black] (3.12,0.16) rectangle (3.23,2.53);\filldraw[fill=myyellow, draw=black] (2.75,0.22) rectangle (2.85,2.80);\filldraw[fill=myyellow, draw=black] (0.65,3.18) rectangle (2.48,3.23);\filldraw[fill=myyellow, draw=black] (0.54,3.12) rectangle (2.53,3.18);\filldraw[fill=myyellow, draw=black] (3.34,0.92) rectangle (3.39,2.37);\filldraw[fill=myyellow, draw=black] (3.28,0.43) rectangle (3.34,2.42);\filldraw[fill=myyellow, draw=black] (0.48,0.43) rectangle (2.64,3.12);\filldraw[fill=myyellow, draw=black] (2.64,0.43) rectangle (2.75,2.91);\filldraw[fill=myyellow, draw=black] (2.85,0.22) rectangle (3.12,2.64);\end{tikzpicture} & \\
\bottomrule
\end{tabular}\label{tab:optimum_blanket}
\end{center}
\end{sidewaystable}

FAST performs slightly better than SF for small $K$ values. As mentioned earlier, SF performs poorly on nonconvex shapes as illustrated in Figure~\ref{fig:example-1} for a particular example. On some simple instances FAST may perform worse than other heuristics as a consequence of its greedy nature, as illustrated in Figure~\ref{fig:example-2}. FAST has another drawback: the maximum number of rectangles that can be placed may be limited depending on the image. After placing fewer than $K$ rectangles, all the pixels of the target region might be covered with rectangles preventing FAST to add a new rectangle, because FAST does not modify rectangles placed in previous steps. For example, for a simple small shape like avatar1, it cannot place rectangles to improve the objective value and has the same value for increasing $K$. Although SF performs poorly for small $K$, the quality of the approximation increases as $K$ increases.  It even outperforms CSA for simple shapes when $K$ is large.

\begin{sidewaystable}[!htbp]
\renewcommand{\arraystretch}{0.8}
\caption{Comparison of the objective values}
\footnotesize
\begin{center}
{\begin{tabular}{r r rrrr r rrrr r rrrr}
\toprule
			& \phantom{a} & \multicolumn{4}{c}{$K=3$} & \phantom{a} & \multicolumn{4}{c}{$K=5$}  & \phantom{a}  & \multicolumn{4}{c}{$K=10$} \\
			\cmidrule(r){3-6} 	\cmidrule(r){8-11} 	\cmidrule(r){13-16}
	Image	&&  BP 	& SF (\%) & CSA (\%)	& FAST (\%) &	&  BP 	& SF (\%) & CSA (\%)	& FAST (\%) &	&  BP 	& SF (\%) & CSA (\%)	& FAST (\%)  \\
	\midrule
79510 	&& 4678	& 49.9	& 29.2	& 41.2	&&3940	& 22.8	& 38.0	& 47.1	&&2912	& 37.8	& 49.7	& 65.0	 \\
79611 	&& 4952	& 21.4	& 5.2	& 27.4	&&3834	& 23.8	& 17.6	& 57.0	&&2664	& 49.2	& 63.1	& 105.3	 \\
79712 	&& 3981	& 26.9	& 34.3	& 30.4	&&3420	& 23.4	& 37.0	& 33.5	&&2559	& 39.3	& 56.2	& 47.6	 \\
79813 	&& 4806	& 30.5	& 11.6	& 20.7	&&4045	& 21.5	& 22.6	& 31.9	&&3339	& 31.9	& 33.1	& 42.9	 \\
79914 	&& 4879	& 74.1	& 29.0	& 51.0	&&4259	& 33.6	& 28.2	& 57.0	&&3513	& 8.5	& 48.2	& 64.6	 \\
79915 	&& 4375	& 19.8	& 13.9	& 51.5	&&3406	& 28.3	& 36.4	& 79.7	&&2383	& 37.1	& 59.9	& 138.9	 \\
79916 	&& 5639	& 31.0	& 14.1	& 49.6	&&4487	& 18.8	& 43.0	& 71.5	&&3334	& 14.1	& 43.7	& 102.0	 \\
79917 	&& 4676	& 54.0	& 31.7	& 34.4	&&3689	& 33.6	& 24.2	& 57.1	&&2889	& 35.1	& 30.3	& 79.3	 \\
79918 	&& 5160	& 38.6	& 17.8	& 55.0	&&4336	& 19.4	& 26.3	& 66.4	&&3509	& 25.0	& 39.3	& 89.2	 \\
79919 	&& 4288	& 35.4	& 14.2	& 18.7	&&3369	& 24.8	& 39.1	& 34.6	&&2486	& 39.0	& 70.9	& 47.5	 \\
\midrule
\multicolumn{3}{c}{Avg. $\pm$ Std. Dev.} & 38.2 $\pm$ 16.8	 & 20.1 $\pm$ 10.0	 & 38.0 $\pm$ 13.5		&&&	25.0 $\pm$ 5.3	 & 31.2 $\pm$ 8.5	 & 53.6 $\pm$ 16.6		&&&	31.7 $\pm$ 12.4	 & 49.4 $\pm$ 13.2	 & 78.2 $\pm$ 30.9 \\
\bottomrule
\end{tabular}}
\label{tab:results-leather-1}
\end{center}
\end{sidewaystable}

\begin{table}[!htbp]
\caption {Comparison of the objective values}
\footnotesize     
\begin{center}
\begin{tabular}{r r rrrr r rrrr}
\toprule
			& \phantom{a} & \multicolumn{4}{c}{$K=15$} & \phantom{a} & \multicolumn{4}{c}{$K=20$}  \\
			\cmidrule(r){3-6} 	\cmidrule(r){8-11} 	
	Image	&&  BP 	& SF (\%) & CSA (\%)	& FAST (\%) &	&  BP 	& SF (\%) & CSA (\%)	& FAST (\%)   \\
	\midrule
79510 	&& 2090	& 76.0	& 94.4	& 115.4	&&2492	& 23.3	& 80.0	& 76.8	 \\
79611 	&& 2505	& 22.5	& 56.2	& 117.4	&&2166	& 28.6	& 93.9	& 151.5	 \\
79712 	&& 2772	& 17.9	& 38.0	& 34.9	&&2186	& 17.1	& 91.2	& 71.1	 \\
79813 	&& 2668	& 40.3	& 71.0	& 78.9	&&2587	& 25.0	& 57.0	& 84.5	 \\
79914 	&& 2666	& 23.7	& 76.6	& 105.4	&&2540	& 21.2	& 82.0	& 111.5	 \\
79915 	&& 2736	& 8.6	& 53.2	& 108.0	&&1759	& 49.9	& 143.0	& 223.6	 \\
79916 	&& 2255	& 49.4	& 112.6	& 187.5	&&2671	& 14.9	& 73.0	& 142.7	 \\
79917 	&& 2460	& 22.0	& 49.1	& 100.4	&&2329	& 20.0	& 72.4	& 111.7	 \\
79918 	&& 3274	& 22.1	& 36.2	& 102.7	&&3099	& 15.0	& 50.2	& 114.2	 \\
79919 	&& 1905	& 22.2	& 108.2	& 82.6	&&2044	& 3.4	& 93.0	& 70.2	 \\
\midrule
\multicolumn{3}{c}{Mean $\pm$ Std. Dev.} & 30.5 $\pm$ 19.7	 & 69.6 $\pm$ 27.9	 & 103.3 $\pm$ 38.2		&&&	21.8 $\pm$ 12.1	 & 83.6 $\pm$ 25.5	 & 115.8 $\pm$ 47.4\\
\bottomrule
\end{tabular}
\end{center}
\label{tab:results-leather-2}
\end{table}

BP always outperforms other methods in exchange for increased running time, even on instances where other methods are trapped in a local minimum, as illustrated in Figure~\ref{fig:example-3} for a particular example. BP's inefficiency is not a problem for real world computer vision applications. Although they are mostly based on online scenarios, a rectangle blanket is first computed offline, which is then repeatedly used online. For example ~\cite{MohrZachmann2010a, MohrZachmann2010b}  first precompute rectangle blankets for different hand shapes and store them for real time use to speed up hand shape matching.  ~\cite{DemirozSalahAkarun2014} determine rectangle blankets for crude human silhouettes at different locations offline before using them in a generative model to compute occupancy probability at a location. In a slightly different work, \cite{Chan2014} find the minimum number of rectangles that best fit an integrated circuit layout off-line prior to the minimization of the mass production time. 

\subsubsection{The number of rectangles in a blanket}

\begin{figure}[!ht]
  \centering
    \begin{subfigure}[b]{0.49\textwidth}
       \includegraphics[width=\textwidth] {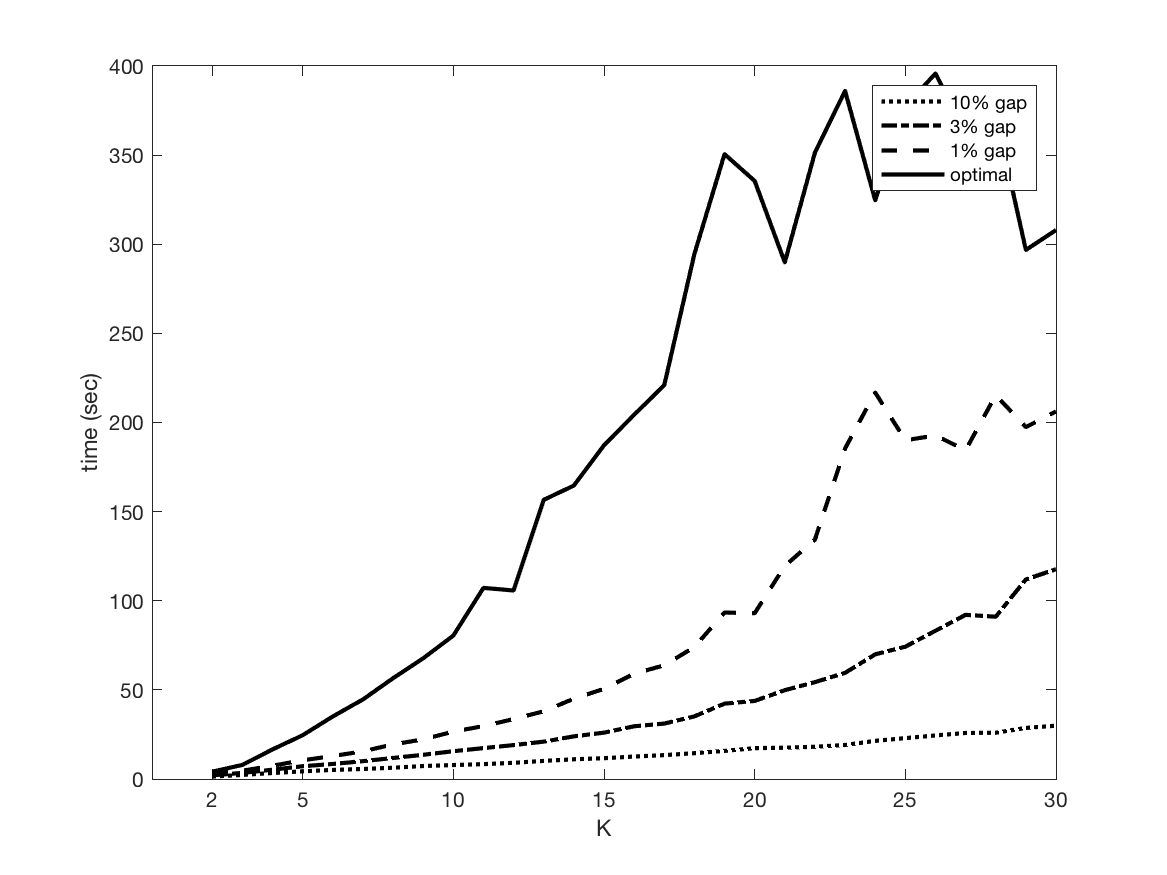}
      \caption{ CPU times to reach optimum and $1$†\%,  $3$ \%, $10$ \% gaps vs. $K$}\label{fig_V_1a}
    \end{subfigure}
    \begin{subfigure}[b]{0.49\textwidth}
       \includegraphics[width=\textwidth] {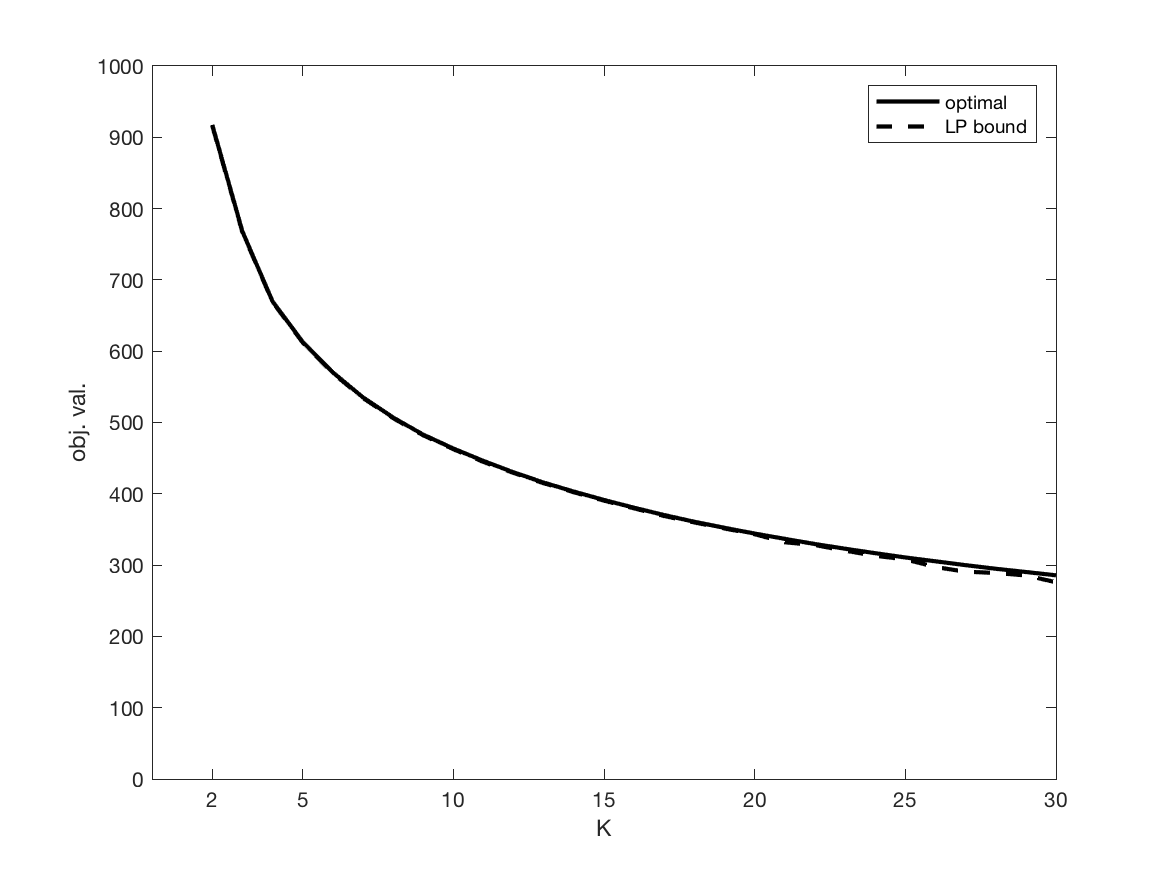}
       \caption{Optimum values and LP lower bounds vs. $K$ \\   }\label{fig_V_1b}
    \end{subfigure} 
 \caption{\textcolor{red}{Average values of 37 different test instances}}\label{fig_V_1}
\end{figure}

\textcolor{red}{To give theoretical results suggesting a particular value for $K$ does not seem trivial. However, it is not very difficult to see that RBP's optimum value is a non-increasing function of $K$ and reaches its lower bound $0$ for $K$ equals to the size of the target image, namely the number of $1$s in matrix $I$. This follows from the fact that every pixel can be treated as a $1\times 1$ rectangle and that many non-overlapping rectangles form a blanket fitting perfectly the target image. Setting $K$ equals to the smallest of $W$ and $H$ under the assumption that the target image is orthogonally convex and packed tightly within a $W\times H$ rectangle, also gives the smallest optimum value $0$, since a blanket consisting of the rectangles obtained by slicing the target image horizontally or vertically into $W$ or $H$ one-pixel wide strips, fits the target image perfectly. It is also possible to say that the optimum value of the LP relaxation is a convex function of $K$, since $K$ belongs actually to the right-hand side of the formulation and it has been known that the optimal value of a minimizing LP is a convex function of the resource vector \citep{CharnesCooper1962}. Hence, we expect similar behavior for the integer optimum values of RBP as well. Also, it is not surprising to see that the running time is a non-decreasing function of $K$, since higher $K$ means better fit, or equivalently higher blanket quality, which we have to pay for.} 

\textcolor{red}{We have conducted computational tests to be able to make concrete suggestions on the choice of $K$ values. The motivation behind these experiments is based on two main questions: 1. Is there a value of $K$ beyond which the decrease in the value of the optimum objective function becomes significantly smaller? 2. Is there a value of $K$ beyond which the running time for the exact solution becomes significantly larger? We think a yes answer to both of them makes a $K$ value promising, which points that it is not worth considering larger blanket sizes.} 

\textcolor{red}{We selected $37$ test instances (i.e. $7$ bat, $2$ device5, $18$ dog and $10$ toy instances) for which our branch-and-price algorithm computes an optimal solution within 1-hour CPU time limit, and set $2\leq K \leq 30$. This makes possible not only the determination of the times where the objective function reaches proven optimal values, but also certain optimality gaps effectively. For each instance we group the results in two: running times to reach optimum and $10$ \%, $3$ \%, $1$ \% gap values, and integer and LP relaxation optimum values. This gives $2\times37 = 74$ result groups at sum. Then we combined them by taking the averages of the $37$ values collected for each $K$ per group. The resulting plots are given in Figure \ref{fig_V_1}.}

\textcolor{red}{First of all one can easily observe the convex behavior of the LP lower bounds and optimum values as a function of $K$. Besides they are decreasing and approaching asymptotically to $0$. Observe also the quality of the LP lower bounds: they are very close to the optimum values. It seems that setting $5 \leq K \leq 10$ is not a bad choice, since the increase in the running times, and decreases in the optimal value and LP bound become sharper beyond and prior to these values.}

\begin{figure}[!ht]
  \centering
    \begin{subfigure}[b]{0.49\textwidth}
       \includegraphics[width=\textwidth] {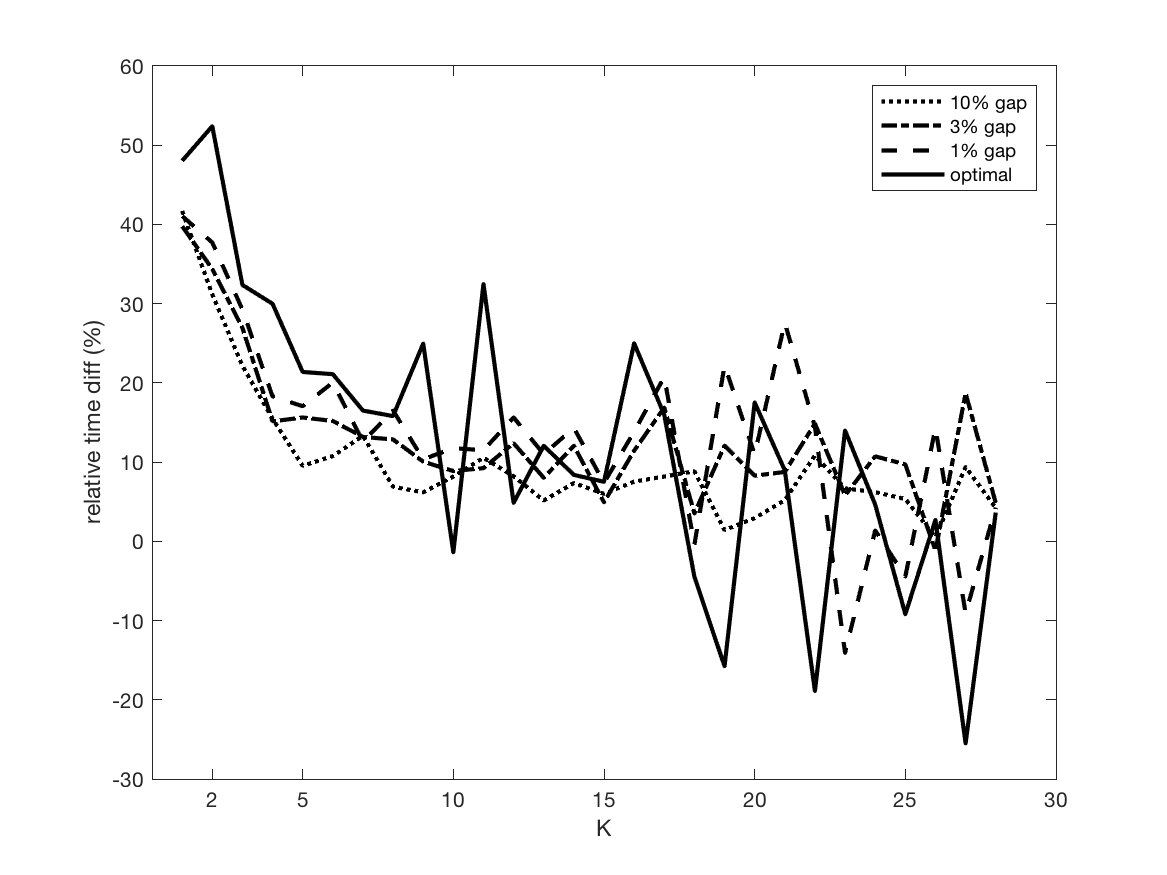}
      \caption{CPU times to reach optimum and $1$ \%, $3$ \%, $10$ \% gaps vs. $K$}\label{fig_V_2a}
    \end{subfigure}
    \begin{subfigure}[b]{0.49\textwidth}
       \includegraphics[width=\textwidth] {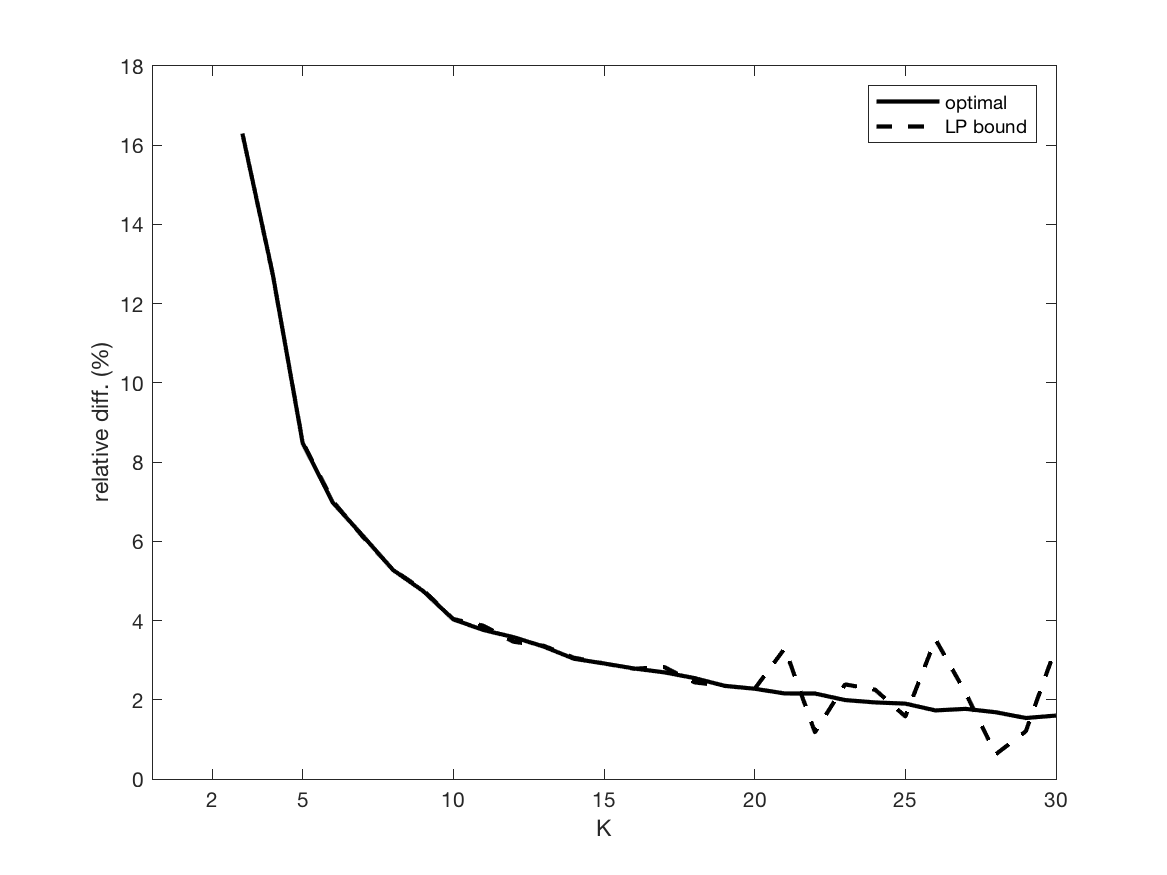}
       \caption{Optimum values and LP lower bounds vs. $K$}\label{fig_V_2b}
    \end{subfigure} 
 \caption{\textcolor{red}{ Average relative changes of 37 different test instances}}\label{fig_V_2}
\end{figure}

\textcolor{red}{We also consider percent relative increases for times, and percent relative decreases for integer optimum values and LP bounds. They are calculated according to formulae
$$100\times\frac{t_K - t_{K-1}}{t_{K-1}} \text{   and   }100\times\frac{z^*_{K-1} - z^*_K}{z^*_{K-1}},$$
respectively, for times and optimum values, for each $K$ and for each data set. The averages taken over the test instances for each $K$  are given in Figure \ref{fig_V_2}: It can be observed that the average relative deviations of the running times seem to reach and oscillate around a steady state of $10 - 15$ \% after $K = 8$; they are decreasing for $ K \leq 7$. As for the objective values, the decrease in the averages is sharper for $K \leq 10$. In short, roughly, it is possible to say for $K > 10$ the increase in the blanket quality becomes lower with an increasing cost.
}

\section{Conclusions}\label{sec:conclusions}

\textcolor{red}{ The problem of representing a target in a binary image as a collection of non-overlapping rectangles is often encountered in computer vision.} In this study, we have formulated the rectangle blanket problem as finding a set of non-overlapping rectangles that minimizes the non-overlapping area between the target image and the rectangles in the blanket. We have developed a branch-and-price algorithm to determine an optimal rectangle blanket. In the column generation phase, to solve the pricing subproblem, we have proposed a geometric branch-and-bound scheme where we start with all possible rectangles and split the rectangle set into two disjoint subsets to branch.

We have also introduced three heuristics to solve the rectangle blanket problem approximately. The first 
two of them, SF and FAST are developed adopting the ideas available in \cite{DemirozSalahAkarun2014} and \cite{MohrZachmann2010b}, respectively. Both are simple yet efficient algorithms. Unfortunately, they are not very accurate. The third one is a novel constrained simulated annealing heuristic, which tries to minimize a function obtained by adding penalty terms forcing the upper bound on the number of rectangles and punishing overlapping rectangles. 

We have prepared benchmarks and compared the performances of the four different methods. The experiments showed that SF and FAST are comparable to each other for small $K$.  CSA performs better than SF and FAST with a large margin for small $K$ in the expense of more computational power.  As $K$ increases SF closes the gap between CSA and starts performing well. BP always produces better or equally good results. Besides, for all of the instances the results BP produces are proven to be optimal. \textcolor{red}{The overall performance of the heuristics are not bright with respect to the accuracy of their results. This is a side effect of their naive design. However, this feature makes them extremely efficient and attractive for those looking for an order of magnitude faster ``quick and dirty"  solution.} 

\textcolor{red}{There are a couple of issues which we can mention as potential future research directions. First of all we plan to adapt different computer vision problems so that they can benefit from the proposed algorithms, and try to find new valid inequalities for increasing the efficiency of the BP algorithm. Also, it can be interesting to use BP for generating initial configuration of the small items on the master surface for the solution of irregular cutting / packing problems and study the impact of this approach on the solution efficiency. Finally, one can study whether the application of $\Phi$ and $\Gamma$ functions, which are used in geometry modeling for the covering problems, is also possible for modeling and evaluating rectangle blankets.}

\section*{Acknowledgements}
We gratefully acknowledge the support of Bogazici University project BAP-6754 and State Planning Organization project DPT-TAM  2007K120610.

\bibliographystyle{apalike}
\bibliography{RBP}
\end{document}